\pdfoutput=1
\documentclass[a4paper,oneside,10pt]{article}

\usepackage[utf8]{inputenc}
\usepackage[a4paper, total={6in, 8in}]{geometry}
\usepackage{nicefrac}
\usepackage{amsmath,amssymb,amsthm}
\usepackage{dsfont}
\usepackage{bbm}
\usepackage{array}
\usepackage{booktabs}
\usepackage{arydshln}
\usepackage{longtable}
\usepackage{multirow}
\usepackage{graphicx}
\usepackage{xcolor}
\usepackage{url}
\usepackage{ifthen}
\usepackage[round,semicolon]{natbib}
\usepackage{caption}
\usepackage{subcaption}
\usepackage{marvosym}
\usepackage{cases}
\usepackage{enumitem}
\usepackage[colorlinks, linkcolor = black, citecolor = black, filecolor = black, urlcolor = blue]{hyperref}
\newcommand{\doi}[1]{doi: \href{https://doi.org/#1}{\urlstyle{rm}\nolinkurl{#1}}}

\theoremstyle{plain}
\newtheorem{theorem}{Theorem}
\newtheorem{lemma}[theorem]{Lemma}

\newtheorem{corollary}[theorem]{Corollary}
\newtheorem{assumption}{Assumption}
\theoremstyle{definition}
\newtheorem{remark}[theorem]{Remark}

\theoremstyle{plain}
\newtheorem{Slemma}{Lemma}
\theoremstyle{definition}
\newtheorem{Sremark}[Slemma]{Remark}

\newcommand{\rd}{\mathrm{d}}
\newcommand{\rP}{\mathrm{P}}

\newcommand{\dense}{{\text{\tiny$[d]$}}}
\newcommand{\sparse}{{\text{\tiny$[s]$}}}

\newcommand{\llslim}{\mathrel{\vcenter{\offinterlineskip\vskip--0.3em\hbox{$\ll$}\vskip0.1em\hbox{$\,\sim$}}}}

\newcommand{\R}{{\mathbb{R}}}
\newcommand{\N}{{\mathbb{N}}}
\newcommand{\Z}{{\mathbb{Z}}}

\newcommand{\Hclass}{\mathcal{H}(\alpha; M)}

\newcommand{\expec}{{\mathbb{E}}}
\newcommand{\prob}{{\mathbb{P}}}
\newcommand{\var}{{\operatorname{\mathbb{V}\mathrm{ar}}}}
\newcommand{\cov}{{\operatorname{\mathbb{C}\mathrm{ov}}}}
\newcommand{\one}{\mathbbm{1}}

\newcommand*{\defeq}{\mathrel{\vcenter{\baselineskip0.5ex \lineskiplimit0pt
			\hbox{\scriptsize.}\hbox{\scriptsize.}}}%
=}
\newcommand*{\defeql}{ = \mathrel{\vcenter{\baselineskip0.5ex   	\lineskiplimit0pt
			\hbox{\scriptsize.}\hbox{\scriptsize.}}}%
   }

\newcommand{\mb}{\mathbf}

\title{Beyond average warming: Two-sample inference for dense-sparse functional data reveals changes in intraday temperature patterns}

\author{Kevin Wilk and Hajo Holzmann\footnote{Corresponding author. Prof. Dr. Hajo Holzmann, Department of Mathematics and Computer Science, Philipps-Universit\"at Marburg, Hans-Meerweinstr., 35043 Marburg, Germany}\\
    \small Department of Mathematics and Computer Science\\
    \small Philipps-Universit\"at Marburg\\
    \small \{wilk,  holzmann\}@mathematik.uni-marburg.de }

\date{}

\begin{document}

\maketitle

\begin{abstract}
Modern weather stations in Germany record daily temperatures every 10 minutes, whereas measurements from historical reference periods are often only available at much coarser temporal resolutions, typically hourly. This discrepancy must be accounted for when comparing historical and current daily temperature patterns. Motivated by this problem, we develop two-sample inference procedures for functional data under sampling schemes where one sample is densely observed while the other is relatively sparse. Building on recent ideas from transfer learning for functional data, our method exploits that the  difference of the mean functions is typically smoother than the mean functions themselves. We establish a functional central limit theorem in the space of continuous functions and develop multiplier bootstrap methods for constructing uniform confidence bands. Extensions to functional time series are also discussed. Applying the methodology to daily temperature curves from four German weather stations, analyzed separately by month, shows that the warming was not uniform over the day: the intraday temperature pattern changed significantly at all four stations, in all or almost all months except at the maritime station Hamburg, where it changed from April to September. In Berlin, daytime temperatures rose more than nighttime temperatures, in particular in spring and early summer, while the timing of the changes over the day differs between the stations.
\end{abstract}

\vspace{3mm}

\noindent \textbf{Keywords:} climate change, functional time series, transfer learning, two-sample problem, uniform confidence bands

\section{Introduction}

Modern weather stations in Germany record temperatures at high temporal resolution, typically every 10 minutes, whereas measurements from historical reference periods are often only available on much coarser, frequently hourly, grids. This discrepancy must be accounted for when comparing historical and recent daily temperature curves. Motivated by this problem, we develop two-sample inference procedures for functional data under asymmetric dense-sparse sampling schemes. Applying the proposed methodology to temperature series from major German cities  we find that climate change has altered not only average daily temperatures in almost every month, but also intraday temperature patterns.

The classical two-sample problem for functional mean functions has been studied in \citet{horvath2012}, Chapter 5, using $L_2$-type statistics under the assumption of continuously observed processes. This has been extended in various directions: \citet{qiu2021two, cao2012simultaneous} allow discrete, densely-observed curves and also consider sup-norm statistics, \citet{duncan2022} construct kernel-based tests,  \citet{chen2026} investigate robustness and minimax optimality, and \citet{pomann2016two} construct tests for equal distributions using functional principal components based on discrete observations, among others. Recently, \citet{cai2024} considered a setting in which one sample is densely observed while the other is sparse, framing the problem in a transfer-learning context and deriving optimal $L_2$ convergence rates for mean estimation. However, simultaneous inference and confidence bands in such settings have not yet been investigated.

Temperature series have time-series structure, see \citet{hoervathkokoska, bosq} for general references on functional time series. The two-sample problem for functional time series has been investigated in
 \citet{horvath2013estimation} for $L_2$-type statistics, and in
\citet{dette2020} as well as in \citet{dette2022detecting} using sup-norm statistics. These papers focus on continuously observed curves, with discretely-observed functional time series, and the dense-sparse setting in particular, not having been thoroughly investigated in the literature so far.

Our main contributions are as follows. We develop simultaneous inference for the difference of the mean functions in a two-sample problem in which only one sample is densely observed. To this end we combine the idea of transfer learning from \citet{cai2024}, where the difference is estimated in a second step from residuals, with the linear smoothers and their analysis in the supremum norm under synchronous designs from \citet{berger2024dense}.%
This approach allows to make use of higher-order smoothness of the difference as compared to the individual mean functions, which is natural in applications where the difference function exhibits less pronounced local structure than the mean functions themselves. For the resulting estimator we derive rates of convergence in the supremum norm and a functional central limit theorem at the parametric rate, also in a regime in which smoothing both samples separately does not permit asymptotic inference. On this basis we construct studentized uniform confidence bands for the difference function and a test for the hypothesis that it equals some constant, not necessarily zero. In contrast to \citet{cao2012simultaneous}, we can
then also treat the setting with only one densely-observed sample together with a sparsely-observed sample.
 Finally, we show how to extend the methodology to discretely observed functional time series, together with a dependent multiplier bootstrap, as required for the analysis of temperature data.

The paper is structured as follows. In Section \ref{sec:apllilust} we present the main results of our data analysis for temperature series in Berlin.
Section \ref{section:model,class,linear estimator} develops the methodology and contains the main theoretical results. In Section \ref{section:quantiles} we provide a simulation study. Section \ref{section:Real data} resumes the data analysis of Section \ref{sec:apllilust}, and provides results for temperature series at weather stations in the three further major German cities Hamburg, Munich and Frankfurt. %
Section \ref{sec:discuss} concludes. Appendix~\ref{app:rates} contains the smoothness classes, the dense and sparse regimes and the limit theorem for the studentized process, while technical assumptions, rates of convergence, proofs, the extension to functional time series, further numerical results and further details on the methodology used in the data analysis are provided in the supplement. The \texttt{R}-code used for the simulations and the data application is provided in a GitHub repository\footnote{github.com/KevinWilk/Beyond-average-warming \href{https://github.com/KevinWilk/Beyond-average-warming}{[github]}}.

\section{Comparing daily temperature series observed at different frequencies}\label{sec:apllilust}

Our study is motivated by comparing recent and historical daily temperature patterns, which we aggregate for each month. While an average increase has been well-established, we are particularly interested in investigating whether this increase is distributed uniformly over the day, or whether the daily temperature pattern itself also has changed. Here we present our results for the temperatures at a weather station in Berlin. Further results for weather stations in Frankfurt, Hamburg and Munich are provided in Section \ref{section:Real data}, while methodological details on the implementation and selection of the tuning parameters are given in Section~\ref{sec:detailsmehtod} of the supplement.

We analyze time series of daily air temperatures for each month, containing on the one hand the temperatures of the recent  period from $2000$ to $2025$, and on the other hand those from a historical reference period from  $1952$ to $1972$. The data are obtained from the Deutsche Wetterdienst (DWD)\footnote{opendata.dwd.de/climate\_environment/CDC/observations\_germany/climate \href{https://opendata.dwd.de/climate_environment/CDC/observations_germany/climate/10_minutes/air_temperature/historical/}{[10 minutes]} and \href{https://opendata.dwd.de/climate_environment/CDC/observations_germany/climate/hourly/air_temperature/historical/}{[hourly]}}.
Our interest is to compare the mean daily temperature curves in each month.
In the current time period we have $26$ segments for each of the $12$  time series covering the respective month, each segment with $30$ or $31$ observations except for February, and in the reference period we have $21$ such segments for each month.
For the data from the current time period, measurements are taken every $10\,$$\,$minutes, leading to $145$ observations  and a dense grid of observations for each day, while in the reference period measurements were only taken every hour ($25$ observations), resulting in a comparatively sparse grid.

The recent records are time-stamped in Coordinated Universal Time (UTC), the historical ones in Central European Time (UTC+1), and summer time was not used in Germany between 1950 and 1979. We shift the historical series by one hour, so that both periods refer to UTC, as do all times of day below. Gaps of up to one hour are interpolated linearly and gaps of one to three hours by splines, while days with longer gaps are removed; this concerns at most about 2\% of the recent and almost none of the historical values. In Berlin and Hamburg, the historical and recent series come from the same station, apart from small relocations and, in Berlin, a neighboring station in 1971 and 1972. In Frankfurt and Munich, different stations were used in the two periods, so that differences between the sites may contribute to the estimated changes there. Details, including checks at station changes, are given in Section~\ref{sec:preprocessing} in the supplement.

For a given month we estimate the daily average temperature curve $\mu^{\dense}$ in the current time period using a local quadratic estimator with bandwidth selected by cross-validation as described in Section~\ref{section: hv cross validation} of the supplement.
Then the difference function $\delta$, defined by
\begin{equation}\label{eq:diffappl}
\delta = \mu^{\dense}-\mu^{\sparse},
\end{equation}
where $\mu^{\sparse}$ is the daily average temperature curve in this month in the historical reference period, is estimated using residuals. This leads to improved rates if the difference $\delta$ is smooth compared to each average temperature curve, see Section \ref{sec:asympnormind}. The cross-validated bandwidths support this: in all combinations of city and season, for which bandwidths are selected, the bandwidth for the residual-based estimate of $\delta$ is larger than that for smoothing the historical sample on its own (supplement, Table~\ref{tab:bwcities}, Remark~\ref{rem:alphadelta}). $\mu^{\sparse}$ is then estimated as the difference of these two estimators. Details are provided in Section \ref{section:Real data}. For a given month, $\mu^{\dense}$ and $\mu^{\sparse}$ are mean curves averaged over the years of the respective period, so that $\delta$ compares period averages. %
Comparing the two halves of each period at level 5\%, we find no significant change within the historical period and, within the recent period, significant changes only in December in all four cities, in June in Berlin, in August and September in Frankfurt and in January in Hamburg (Section~\ref{sec:preprocessing}).

Figure \ref{fig: means of Berlin} displays estimates of  $\mu^{\dense}$ and $\mu^{\sparse}$, with $\mu^{\dense}$ being consistently above $\mu^{\sparse}$, as expected. In Figure \ref{fig:real data delta equals 0 test} we display the mean difference, $\int_0^{24}\hat{\delta}(y)\,\rd y / 24$, together with a confidence interval. Not surprisingly, the mean difference is positive in each month, and significantly so at level 5\% in every month except October (Table~\ref{tab: Main global tests}), with p-values computed as described in
Section~\ref{section:Real data} and Holm-adjusted over the 12 months, separately for each city and each hypothesis.

Further, we plot the estimate $\hat{\delta}$ of the difference function itself together with uniform confidence bands, the construction of which is detailed in Section \ref{section:Two sample problem}. The estimates show a clear pattern over the day, but the bands are wide and compatible with a constant difference in every month, so that Figure~\ref{fig:real data delta equals 0 test} alone does not show that the difference function is non-constant. The bands are wide mainly because the daily curves vary in their overall level from day to day, which affects all times of day alike. This variation cancels in the centered difference $\Delta = \delta - \int_0^{24}\delta(y)\,\rd y / 24$, which describes how the change varies over the day. In Figure~\ref{fig:real data constant test} we therefore construct confidence bands for $\Delta$ (dark shaded).

These are much tighter than those for $\delta$ itself (light-shaded area in Figure \ref{fig:real data constant test}). The hypothesis $\Delta = 0$ of a constant difference is rejected after Holm-adjustment for multiple testing at level 5\% in every month except December (Table~\ref{tab: Main global tests}). Hence we can conclude that, in eleven of the twelve months, the daily mean temperature pattern changed beyond the mere change in its daily average value; this is in line with \cite{Makowski2008}, who reported changes of the diurnal temperature range, but not for the whole daily cycle. In Berlin with its continental climate, mean temperatures from about 6:00 to 17:00 UTC show an above-average increase, peaking around 9:00 to 10:00 UTC, while evening and night mean temperatures saw a below-average increase; the change is most pronounced from April to August.

\begin{figure}[h!]
    \centering
    \begin{subfigure}[t]{0.85\linewidth}
        \includegraphics[width=\linewidth]{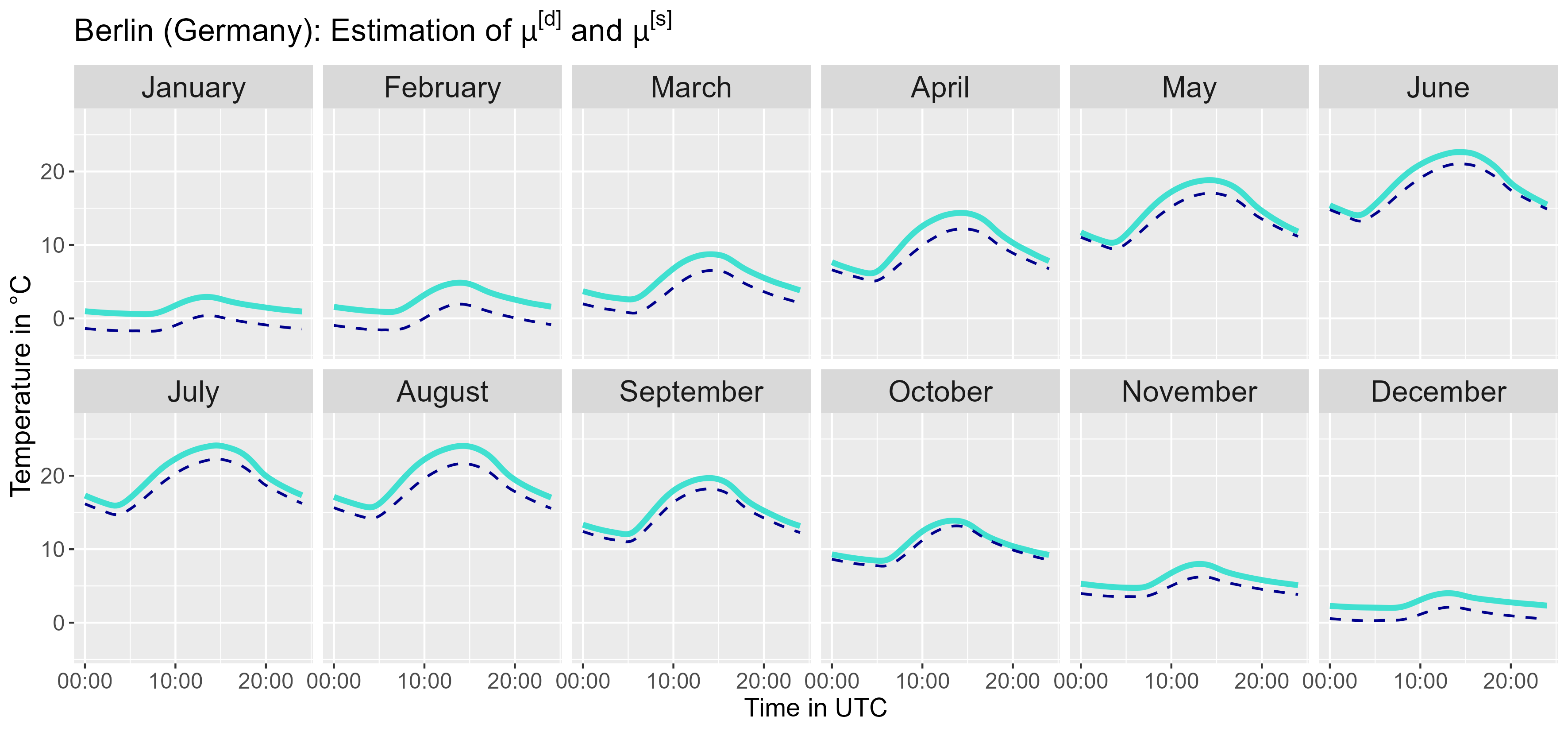}
    \end{subfigure}
    \caption{$\hat{\mu}^{\dense}$ (blue line) from $2000$ to $2025$ and  $\hat{\mu}^{\sparse}=\hat{\mu}^{\dense}-\hat{\delta}$ (dark blue dashed line) from $1952$ to $1972$.}
    \label{fig: means of Berlin}
\end{figure}

\begin{figure}[h!]
    \centering
    \phantom{xxx}
    \begin{subfigure}[t]{0.93\linewidth}
        \includegraphics[width=\linewidth]{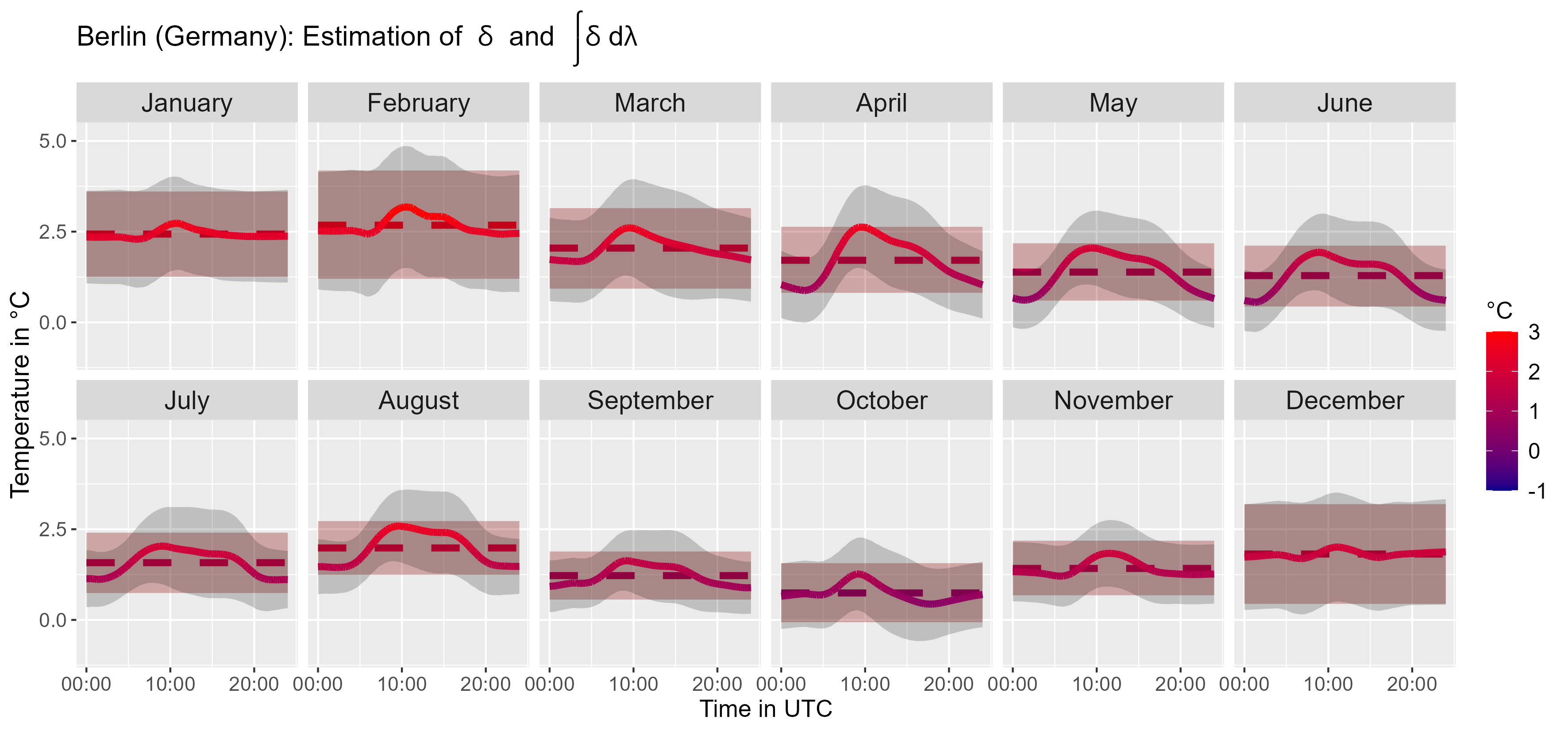}
    \end{subfigure}
    \caption{$\hat{\delta}$ (solid blue--red colored line) and $\int_0^{24}\hat{\delta}(y)\,\rd y / 24$ (dashed blue--red colored line), together with $95$\% uniform confidence bands and confidence intervals.}
    \label{fig:real data delta equals 0 test}
\end{figure}

\begin{figure}[h!]
    \centering
    \begin{subfigure}[t]{0.85\linewidth}
        \includegraphics[width=\linewidth]{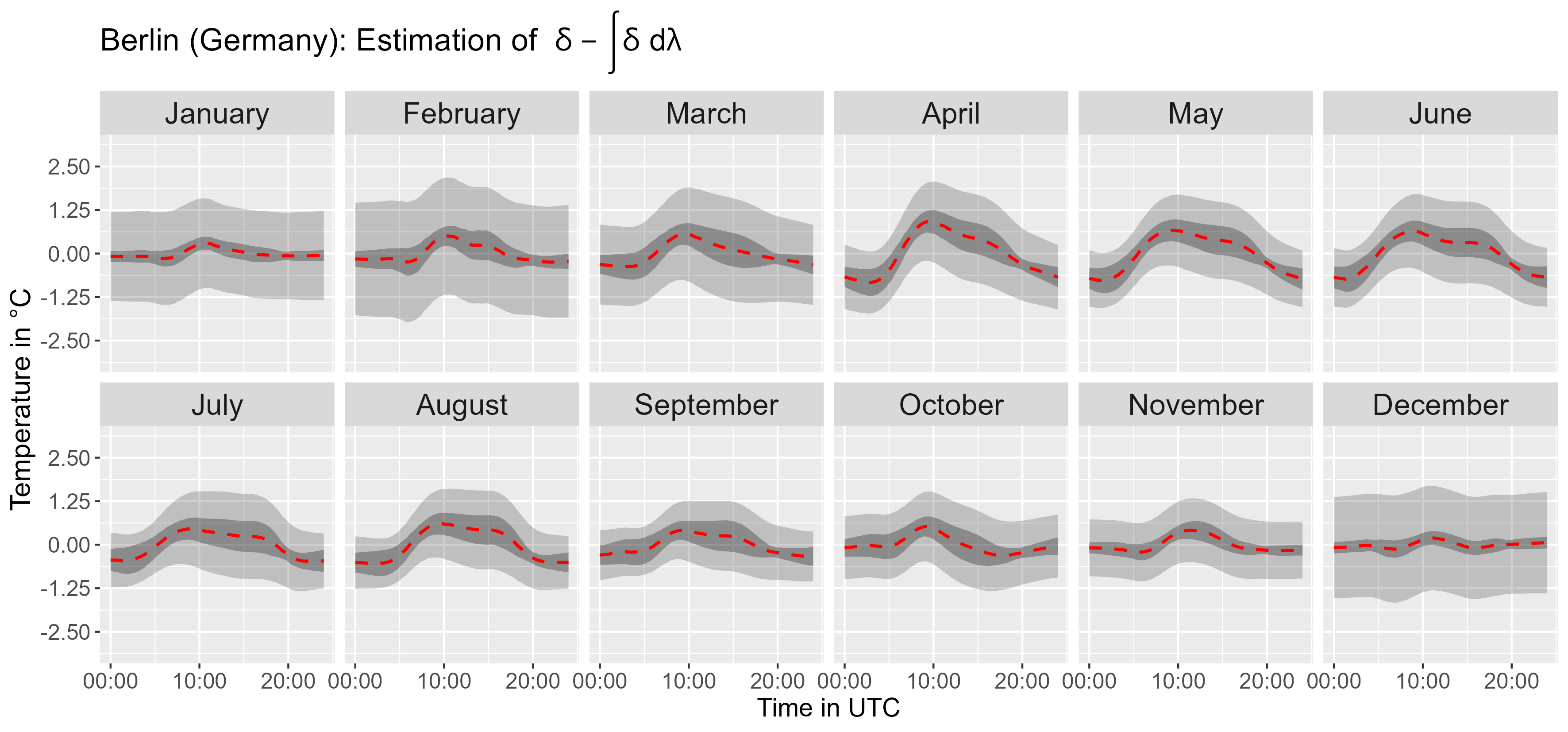}
    \end{subfigure}
    \caption{$\hat{\delta}-\int_0^{24}\hat{\delta}(y)\,\rd y / 24$ (red dashed line) with the $95$\% uniform confidence bands from Figure \ref{fig:real data delta equals 0 test} (light grey area) as well as using \eqref{eq:std statistic 1} for the centered difference (grey area).}
    \label{fig:real data constant test}
\end{figure}

\begin{table}[h!]
    \centering{
    \scriptsize
    \renewcommand{\arraystretch}{1.3}
    \begin{tabular}{c|c|c|c|c|c|c|c|c|c|c|c|c}
    \multicolumn{13}{c}{ } \\
    \multicolumn{13}{c}{Time periods $2000$--$2025$ and $1952$--$1972$} \\
    \hline
                         $H_0$                          & Jan       & Feb       & Mar       & Apr       & May       & Jun       & Jul      & Aug        & Sep       & Oct       & Nov      & Dec        \\
    \hline
        \rule{0pt}{2.5ex} $\int \delta = 0$             & $0{.}001$ & $0{.}002$ & $0{.}002$ & $0{.}001$ & $0{.}002$ & $0{.}011$ & $0{.}002$& $0{.}001$  & $0{.}002$ & $0{.}074$ & $0{.}002$& $0{.}022$   \\
        \rule{0pt}{2.5ex} $\Delta = 0$                  & $0{.}001$ & $0{.}001$ & $0{.}001$ & $0{.}001$ & $0{.}001$ & $0{.}001$ & $0{.}001$& $0{.}001$  & $0{.}002$ & $0{.}001$ & $0{.}001$& $0{.}068$
    \end{tabular}
    }
    \caption{Berlin (Germany): p-values of the global tests of $\int\delta = 0$ and $\Delta = 0$ for each month, Holm-adjusted over the 12 months, separately for each city and each hypothesis.}
    \label{tab: Main global tests}
\end{table}

\section{Methodology for unbalanced dense-sparse two-sample problems}\label{section:model,class,linear estimator}

We consider the two-sample problem for discretely-observed functional data, with a more densely-observed sample  $(Y^{\dense}_{k,l},t^{\dense}_l)$, $k=1,\hdots,{m}$, $l=1,\hdots,{p_d}$, and a potentially sparsely observed sample $(Y^{\sparse}_{i,j}, t^{\sparse}_{j}), $ for $i=1,\hdots,{n}$, $j=1,\hdots,{p}$.
We thus assume $p \lesssim p_d$, and in addition focus on the setting $n \lesssim {m}$. As in \citet{berger2024dense}, the terms dense and sparse refer to the asymptotic regime: %
in both samples the number of design points $p$ and $p_d$ tends to infinity, but in the sparse sample potentially too slowly for its mean function to be estimable at the parametric rate $n^{-1/2}$ from that sample alone, see Remark~\ref{rem:separate} in Appendix~\ref{app:rates}.
The two samples are assumed to be independent, while for each we consider both the independent sampling situation as well as a  functional time series framework.

Further, as motivated by our application, we focus on synchronously sampled observations, the typical design for machine recorded functional data. Thus, our model is formulated as
\begin{align}
\begin{split}
    Y^{\sparse}_{i,j}(t^{\sparse}_{j}) &= \mu^{\sparse}(t^{\sparse}_{j}) + Z^{\sparse}_i(t^{\sparse}_{j}) + \epsilon^{\sparse}_{i,j},\quad\text{ for }  i=1,\hdots,{n} \text{ and }j=1,\hdots,{p},\\
    Y_{k,l}^{\dense}(t^{\dense}_l) &= \mu^{\dense}(t^{\dense}_l) + Z^{\dense}_k(t^{\dense}_l) + \epsilon^{\dense}_{k,l},\quad\text{ for }  k=1,\hdots,{m} \text{ and }l=1,\hdots,{p_d}, \label{eq:main.4}
\end{split}
\end{align}
where the $t^{\sparse}_{1}<\hdots<t^{\sparse}_{{p}}$ and the $t^{\dense}_1<\hdots<t^{\dense}_{{p_d}}$ are the deterministic, synchronous design points covering the interval $T$, which we take as $T= [0,1]$ in our theoretical developments.
The observational errors $\epsilon^{\sparse}_{i,j}$ and $\epsilon^{\dense}_{k,l}$ are centered, independent of each other as well as independent of the centered and square-integrable random processes $(Z^{\sparse}_i\,:\,1\leq i \leq {n})$ and $(Z^{\dense}_k\,:\,1\leq k\leq {m})$. Finally, the mean functions are denoted by $\mu^{\sparse}$ and by $\mu^{\dense}$.

\smallskip

We are interested in comparing parameters or even the full distributions in the two samples in \eqref{eq:main.4}. Here we concentrate on the mean functions. In the following let us denote their difference by %
\begin{align*}
    \delta\defeq \mu^{\sparse} - \mu^{\dense} .
\end{align*}
Note that in our data application for the sake of interpretability we consider $\mu^{\dense} - \mu^{\sparse} $ in \eqref{eq:diffappl} and in Section \ref{section:Real data}.

\subsection{Mean-estimation in unbalanced dense-sparse two-sample problems}

In the model \eqref{eq:main.4}, smoothing methods are frequently used to estimate individual mean functions, see e.g.~\citet{caiyuan2011, cao2012simultaneous, degras, berger2024dense, zhang2016sparse, li2010uniform, xiao2020asymptotic}.
In a first step we estimate the more densely observed ${\mu}^{\dense}$ (as we assume $p \lesssim p_d$) by
\begin{align}
\label{eq:estimatormu}
    &\hat{\mu}^{\dense}_{m}(t; h_d) \defeq \sum_{l=1}^{{p_d}} w^{\dense}_{l}(t;{h_d},t^{\dense}_1,\hdots,t^{\dense}_{{p_d}})\bar{Y_l}^{\dense}, \quad \text{ where} \quad \bar{Y_l}^{\dense}=\frac{1}{m}\sum_{k=1}^{m} Y_{k,l}^{\dense}.
\end{align}
Here $w^{\dense}_{l}(t;{h_d},t^{\dense}_1,\hdots,t^{\dense}_{{p_d}}) = w_{l}^{\dense}(t;{h_d})$ are deterministic weights depending on the design points $t^{\dense}_l$ and on a bandwidth  parameter ${h_d} > 0$ for localization at $t^{\dense}_l$, and on which we impose the assumptions listed in Section~\ref{ass:genericweights}. These are checked for local polynomial weights in Lemma 1 in \citet{berger2024dense}.
Next we could proceed similarly as e.g.~\citet{cao2012simultaneous} and analogously to \eqref{eq:estimatormu} estimate ${\mu}^{\sparse}$, say by $ \bar{\mu}^{\sparse}_{{ n}}(\cdot, h)$, and then take the difference
\begin{equation}\label{eq:diffest}
\bar{\delta}_{n,m}(t; h, h_d) \defeq \bar{\mu}^{\sparse}_{{ n}}(t;h) - \hat{\mu}^{\dense}_{{m}}(t;h_d),
\end{equation}
where the subscripts $n$ and $m$ refer to the sizes of the two samples. However, as has been prominently observed in \citet{cai2024}, this estimator can be improved upon in the situation where first, $p \ll p_d$, and second, there is only moderately pronounced structure in the difference $\delta$ as compared to the individual mean functions, or more formally expressed higher-order smoothness of $\delta$ as compared to  $\mu^{\sparse}$ and $\mu^{\dense}$. Following \citet{cai2024} we therefore estimate $\delta$ directly based on residuals in the sparse sample, which are formed using $\hat{\mu}^{\dense}_{m}(\cdotp; h_d)$ from  \eqref{eq:estimatormu}. Thus, we set
\begin{align}
    &\hat{\delta}_{n,m}(t; h, h_d)\defeq \sum_{j=1}^{{p}} w^{\sparse}_{j}(t;{h},t^{\sparse}_{1},\hdots,t^{\sparse}_{{p}})\big(\bar{Y_j}^{\sparse} - \hat{\mu}^{\dense}_{{m}}(t^{\sparse}_{j}; h_d)\big), \label{eq:estimators}
\end{align}
and then
\begin{align}
    &\hat{\mu}^{\sparse}_{n,m}(t; h, h_d)\defeq \hat{\delta}_{n,m}(t; h, h_d) + \hat{\mu}^{\dense}_{{m}}(t; h_d). \label{eq:estimators1}
\end{align}

\subsection{Rates of convergence and asymptotic normality}\label{sec:asympnormind}

We characterize the amount of structure in the mean functions and their difference $\delta$ by H\"older smoothness assumptions.
Let $\Hclass$ denote the H\"older class of functions $f\colon T \rightarrow \R$ of smoothness $\alpha>0$ and H\"older constant $M>0$, defined in Appendix~\ref{app:holder}. The smaller $\alpha$ and the larger $M$, the richer the class, which then contains less regular functions with more pronounced features. Thus, we will assume
\begin{equation}\label{eq:smoothnessass}
 \mu^{\dense}, \mu^{\sparse} \in \Hclass \qquad \text{but}\qquad \delta \in \mathcal{H}(\alpha_\delta; M_\delta)\qquad \text{with}\quad \alpha_\delta \geq \alpha,
\end{equation}
and make use in particular of the situation in which $\alpha_\delta > \alpha$.
For the processes we assume H\"older continuous paths of a possibly lower order $\beta\in(0{,}1]$, see \eqref{eq:Z} in Appendix~\ref{app:holder}.

First we consider the setting of independent samples, as formulated in the following assumption.

\begin{assumption}[Independent processes and errors]\label{assumption:Independence}
      The processes $Z^{\sparse}_1,\hdots,Z^{\sparse}_{{n}}$, respectively $Z^{\dense}_1,\hdots,Z^{\dense}_{{m}}$, are independent copies of the centered random processes $Z^{\sparse}$, respectively  $Z^{\dense}$ with finite fourth moments, having covariance kernels $\gamma^{\sparse}_0$ and $\gamma^{\dense}_0$ respectively, and satisfying the H\"older condition \eqref{eq:Z}.
      The observational errors $\epsilon^{\sparse}_{i,j}$, $\epsilon^{\dense}_{k,l}$ are all independent, and independent of the processes $Z^{\sparse}_i$ and $Z^{\dense}_k$.
    Additionally, we assume the distributions of $\epsilon^{\sparse}_{i,j}$ and $\epsilon^{\dense}_{k,l}$ to be sub$\,$-$\,$Gaussian, with $\sigma^2>0$ an upper bound for all sub$\,$-$\,$Gaussian norms of $\epsilon^{\sparse}_{i,j}$ and $\epsilon^{\dense}_{k,l}$.
\end{assumption}

For a centered stationary sequence $(Z_i)_{i\in\Z}$ of processes we write $\gamma_b(t,s) = \expec[Z_1(t)\,Z_{1+b}(s)]$, $b\in\Z$, $t,s\in T$, for the lag-$b$ covariance kernel and $\Gamma = \sum_{b\in\Z}\gamma_b$ for the long-run covariance kernel, with superscripts $\gamma_b^{\sparse}$, $\Gamma^{\sparse}$ and $\gamma_b^{\dense}$, $\Gamma^{\dense}$ for the two samples. Under Assumption~\ref{assumption:Independence}, $\Gamma = \gamma_0$ is the covariance kernel.
Rates of convergence in the supremum norm for $\hat{\delta}_{n,m}$ in \eqref{eq:estimators}, uniformly over \eqref{eq:smoothnessass}, and for $\bar{\delta}_{n,m}$ in \eqref{eq:diffest} are given in Theorem~\ref{thm:rates} in Section~\ref{sec:proofthm} of the supplement and in Remark~\ref{rem:separate} in Appendix~\ref{app:ratesofconv}. Here we focus on asymptotic normality.

\begin{theorem}[CLT: Independent samples]\label{theorem:CLT_delta}
    In model \eqref{eq:main.4} consider the residual-based estimator $\hat{\delta}_{n,m}$ in \eqref{eq:estimators} for the difference $\delta = \mu^{\sparse} - \mu^{\dense}$ under the smoothness assumption \eqref{eq:smoothnessass}, and the  Assumption~\ref{assumption:weights} in Section~\ref{ass:genericweights} on the weights $w^{\sparse}_{j}(\cdot\,;{h})$ and $w^{\dense}_{l}(\cdot\,;{h_d})$  with degrees $\lceil\alpha_{\delta}\rceil - 1$ and $\lceil \alpha\rceil - 1$. Additionally, impose the H\"older conditions \eqref{eq:Z} on the processes, the independence Assumption \ref{assumption:Independence} and Assumption~\ref{assumption: design} on the design, and consider the setting $p \lesssim p_d$ and $n \lesssim m$.
    Suppose that $(\log p)^{1 + 3 \eta}\, n^{1/(2 \alpha_\delta)} \lesssim p$ for some $\eta>0$, and consider the asymptotically non-empty range of bandwidths $H = [c_1 (\log p)^{1+\eta}/p,$ $c_2 / (n^{1/(2 \alpha_\delta)}\, (\log p)^\eta)]$, $c_1,c_2>0$, together with a given sequence $h \in H$. Similarly, also assume that $(\log p_d)^{1 + 3 \eta}\, n^{1/(2 \alpha)} \lesssim p_d$, consider $H_d = [c_1 (\log p_d)^{1+\eta}/p_d,$ $c_2 / (n^{1/(2 \alpha)}\, (\log p_d)^{\eta})]$, together with a given sequence $h_d \in H_d$.
        Then if $n\llslim {m}$ we obtain
    \begin{equation}\label{eq:distinctorder}
        \sqrt{n}(\hat{\delta}_{n,m}(\cdotp; h, h_d)-\delta) \overset{D}{\longrightarrow} \mathcal{G}(0,\Gamma^{\sparse}),
    \end{equation}
    where weak convergence is in the space $C[0,1]$, and $\mathcal{G}$ is a continuous, real-valued Gaussian process on $[0{,}1]$ with covariance kernel $\Gamma^{\sparse}$ of $Z^{\sparse}$.

    If instead $n \simeq {m}$ applies with $ n/m \to \lambda$, then
    \begin{equation}\label{eq:sameorder}
        \sqrt{n}(\hat{\delta}_{n,m}(\cdotp; h, h_d)-\delta) \overset{D}{\longrightarrow} \mathcal{G}(0,\Gamma^{\sparse}+\lambda\,\Gamma^{\dense}).
    \end{equation}
\end{theorem}

To achieve the $\sqrt n$-rate, in case $\alpha_\delta > \alpha$ we only require $(\log n)\, n^{1/(2\, \alpha_\delta)}\ll p$ together with $(\log n)\, n^{1/(2\, \alpha)}\, \ll p_d$ instead of $(\log n)\, n^{1/(2\alpha)}\ll p$ as for the estimator in \eqref{eq:diffest}, see Remark~\ref{rem:separate}.

    The condition $(\log p)^{1 + 3 \eta}\, n^{1/(2 \alpha_\delta)} \lesssim p$, equivalently with $\log n$ in place of $\log p$, makes $(\log n)\, n^{1/(2 \alpha_\delta)} \ll p$ quantitative and allows bandwidth choices of order $h \eqsim (\log p)^{1+\eta}/p$, which are independent of the unknown $\alpha_\delta$ and only slightly above the interpolation level $h \eqsim 1/p$. As in \citet[Theorem 3]{berger2024dense}, the logarithmic factors in $H$ and $H_d$ involve $p$ and $p_d$, so that no upper bounds on the numbers of design points are needed. Further, the condition $(\log n)\, n^{1/(2 \alpha)} \ll p_d$ ensures that the remaining terms in (\ref{eq:convrateh}) which depend on $p_d$ and $m$ are negligible as compared to $n^{-1/2}$, and $(\log p_d)^{1 + 3 \eta}\, n^{1/(2 \alpha)} \lesssim p_d$ makes this condition quantitative and allows for a choice of $h_d$ not depending on $\alpha$.

\begin{remark}[The role of $\alpha_\delta > \alpha$]\label{rem:alphadelta}
Since $\mu^{\dense}$ and $\mu^{\sparse}$ lie in $\Hclass$, the difference $\delta$ always lies in $\mathcal{H}(\alpha; 2M)$, so \eqref{eq:smoothnessass} holds at least with $\alpha_\delta = \alpha$. The behavior of $\hat{\delta}_{n,m}$ does not break at $\alpha_\delta = \alpha$: its rate in Theorem~\ref{thm:rates} depends continuously on $\alpha_\delta$, its Gaussian limits in Theorem~\ref{theorem:CLT_delta} do not depend on $\alpha_\delta$, and only the requirement $(\log n)\, n^{1/(2\alpha_\delta)} \ll p$ on the sparse design weakens as $\alpha_\delta$ grows. For $\alpha_\delta = \alpha$, this requirement and the limit coincide with those for separate smoothing in Remark~\ref{rem:separate}, and in no case is the rate of $\hat{\delta}_{n,m}$ slower than that of $\bar{\delta}_{n,m}$. The gain from $\alpha_\delta > \alpha$ is the range $(\log n)\, n^{1/(2\alpha_\delta)} \ll p \lesssim (\log n)\, n^{1/(2\alpha)}$, in which, in general, only $\hat{\delta}_{n,m}$ permits asymptotic inference. %

In practice, $\alpha_\delta > \alpha$ or $M_\delta \ll M$ if $\alpha_\delta = \alpha$, which have similar finite sample effects, is plausible if both mean functions share their local features, as daily temperature curves do through the diurnal cycle in both periods.
A cross-validated bandwidth for $\hat{\delta}_{n,m}$ that is clearly larger than that for $\bar{\delta}_{n,m}$ therefore indicates $\alpha_\delta > \alpha$ or $M_\delta \ll M$. This is the case for the temperature data, see Section~\ref{sec:apllilust}.
\end{remark}

\begin{remark}[Choice of the bandwidths]\label{rem:bandwidth}
The bandwidth conditions in Theorem~\ref{theorem:CLT_delta} do not require undersmoothing relative to the bandwidths $h^*$ and $h_d^*$ that are rate-optimal for estimation in the supremum norm (Theorem~\ref{thm:rates} in the supplement). Under the conditions on $p$ and $p_d$ in Theorem~\ref{theorem:CLT_delta}, the rate in (\ref{eq:convrateh}) is the parametric $n^{-1/2}$, and at $h^*$ both the bias $(h^*)^{\alpha_\delta}$ and the contribution of the observational errors are of order $(\log(pn)/(pn))^{\alpha_\delta/(2\alpha_\delta+1)} = o(n^{-1/2})$, and similarly at $h_d^*$. The reason is that the leading stochastic term, the average of the $n$ centered processes, is of order $n^{-1/2}$ for every bandwidth, so that, in contrast to nonparametric regression, the bias need not be balanced against it. Correspondingly, the range $H$ is wide: its endpoints differ by a power of $n$ as soon as $p$ exceeds $n^{1/(2\alpha_\delta)}$ by a power of $n$, and then $H$ contains $h^*$; if $p$ exceeds $n^{1/(2\alpha_\delta)}$ only by logarithmic factors, those in the definition of $H$ may exclude $h^*$. Theorem~\ref{theorem:CLT_delta} concerns deterministic sequences of bandwidths and does not cover data-driven choices, for which we rely on the simulations in Section~\ref{subsection: Simulation}. In our numerical work the bandwidths are chosen by cross-validation in Section~\ref{section: hv cross validation}, which minimizes a mean squared prediction error and thus targets estimation accuracy  in the $L_2$-norm rather than the supremum norm.
By the supremum-norm criterion the chosen bandwidths vary strongly across sample replications, such that for a non-constant $\delta$ individual selections are too large for inference and therefore reduce the coverage, whereas the corresponding averaged bandwidth leads to coverage close to the nominal level.
A cross validation with the $L_2$ criterion reduces this spread and selects smaller bandwidths on average, which results in improvements of coverage.
Further, for each month a robustness check of the p-values of the global tests of $\int\delta = 0$ and $\Delta = 0$  across different bandwidths is included in supplementary Material~\ref{sec: add results} with Holm-adjustment applied over the 12 months for each city and hypothesis separately, indicating only one change in Holm decision, namely for $\Delta = 0$ in September in Hamburg.

\end{remark}

For the average difference $\int_0^1 \delta$, which we also discuss in  the application, $\sqrt{n}\,(\int_0^1\hat{\delta}_{n,m} - \int_0^1\delta)$ is asymptotically normal with variance $\int_{[0,1]^2}(\Gamma^{\sparse}+\lambda\,\Gamma^{\dense})$, $\lambda\geq 0$ as in \eqref{eq:sameorder}, under weaker conditions on $p$ and $p_d$, see Remark~\ref{rem:estint} in Section~\ref{sec:proofint} of the supplement.

The proof in Section~\ref{sec: proof main theorem} writes the linear smoothers as operators acting on the normalized averages of the processes and uses the processes only through their central limit theorem in $C[0{,}1]$, which simplifies the argument for a single sample in \citet{berger2024dense}. Hence, for functional time series, Theorem~\ref{theorem:CLT_delta} continues to hold with $\Gamma^{\sparse}$ and $\Gamma^{\dense}$ the long-run covariance kernels, provided the partial sums of the processes satisfy a central limit theorem in $C[0{,}1]$, see Section~\ref{sec:timeseries} in the supplement.

\subsection{Studentized confidence bands and the two-sample problem}\label{section:Two sample problem}

In our application we will investigate whether $\delta$ is equal to some constant value $c$, not necessarily $c=0$. Therefore, in this section we focus on the centered difference and its estimate
\begin{equation}\label{eq:centereddiff}
 \Delta = \delta - \int_0^1 \delta(y)\, \rd y,\qquad \hat{\Delta}_{n,m}(\cdotp;h,h_d) = \hat{\delta}_{n,m}(\cdotp;h,h_d) - \int_0^1 \hat{\delta}_{n,m}(y;h,h_d)\, \rd y.
\end{equation}
Analogous methods and results as developed in this section apply to the difference $\delta$ itself.

We studentize with estimates $\widehat{\Gamma}^{\sparse}$ and $\widehat{\Gamma}^{\dense}$ of $\Gamma^{\sparse}$ and $\Gamma^{\dense}$, given in Section~\ref{sec:covest}: under independence, the covariance kernel estimate $\hat{\gamma}_0$ of \citet{berger2025}, which makes use of higher-order smoothness of the covariance kernel off the diagonal, and for time series a Bartlett-weighted sum of lag-$b$ estimates $\hat{\gamma}_b$, $|b| \leq L$. For the results of this section we require consistency of these estimates; see Remark~\ref{rem:studnoncons} for estimates that are not consistent.

Let $\rP$ map the covariance kernel $f$ of a process $X$ to the covariance kernel of $X - \int_0^1 X$, so that $(\rP\circ f)(t,t) = f(t,t) - 2\int_0^1 f(t,u)\,\rd u + \int_0^1\!\int_0^1 f(u,v)\,\rd u\,\rd v$, see \eqref{eq:projection}. By Theorem~\ref{theorem:CLT_delta}, $\sqrt{n}\,(\hat{\Delta}_{n,m}-\Delta)(t)$ has asymptotic variance $\sigma^2_{\Delta}(t) = \big(\rP\circ(\Gamma^{\sparse} + \lambda\,\Gamma^{\dense})\big)(t,t)$, with $\lambda=0$ if $n \llslim m$ and $\lambda$ as in \eqref{eq:sameorder} if $n \simeq m$, which we estimate by $\hat{\sigma}^2_{\Delta}(t) = \big(\rP\circ(\widehat{\Gamma}^{\sparse} + (n/m)\,\widehat{\Gamma}^{\dense})\big)(t,t)$. If these estimates are consistent and $\sigma^2_{\Delta}>0$ on $[0{,}1]$, then by Corollary~\ref{theorem:CLT_dtandard} in Appendix~\ref{app:stud},
\begin{equation}\label{eq:std statistic 1}
   \bigg(\frac{\sqrt{n}\,\big(\hat{\Delta}_{n,m}(t;h,h_d) - \Delta(t)\big)}{\hat{\sigma}_{\Delta}(t)}\bigg)_{t\in[0,1]}\overset{D}{\longrightarrow} \mathcal{G} \qquad \text{in } C[0,1],
\end{equation}
where $\mathcal{G}$ is a centered Gaussian process with unit variance.

\begin{remark}[Studentization without consistency]\label{rem:studnoncons}
The convergence in \eqref{eq:std statistic 1} extends to estimates $\widehat{\Gamma}^{\sparse}$ and $\widehat{\Gamma}^{\dense}$ that converge uniformly in probability to other continuous covariance kernels, if the limit of $\hat{\sigma}^2_{\Delta}$ is positive on $[0{,}1]$: the limit process is then still centered Gaussian, but need not have unit variance. This applies to the Bartlett-weighted estimates with the small maximal lags $L$ used in our simulations and in the application: for fixed $L$ they estimate $\sum_{|b|\leq L}\big(1-|b|/(L+1)\big)\,\gamma_b$ instead of $\Gamma$. This is the covariance kernel of $(L+1)^{-1/2}\,(Z_1+\ldots+Z_{L+1})$, whereas the sum $\sum_{|b|\leq L}\gamma_b$ for weights equal to one need not be positive semidefinite. Since the bootstrap in Section~\ref{subsection:multiplier bootstrap} uses the same normalization, the lack of consistency affects the shape of the bootstrap bands, but is not essential for their level, see Section~\ref{sec:covest}.
\end{remark}

\begin{remark}[Construction of uniform confidence bands]\label{rem:bands}
    The convergence in \eqref{eq:std statistic 1} allows to construct uniform confidence bands for the centered difference $\Delta$ in \eqref{eq:centereddiff}.
    Given $\tau \in (0{,}1)$ let $q_{1-\tau}$ be the $(1-\tau)$-quantile of $\sup_{t\in[0,1]}|\mathcal{G}(t)|$, where $\mathcal{G}$ is the limit process in \eqref{eq:std statistic 1}.
    Then setting
\begin{align*}
        U^{\pm}_{n,m}(t;h,h_d) \defeq\hat{\Delta}_{n,m}(t;h,h_d)\, \pm \, \frac{q_{1-\tau}}{\sqrt{n}}\,\hat{\sigma}_{\Delta}(t),
    \end{align*}
    we obtain
    \begin{equation}\label{eq:asympconfband}
        \prob\big( \Delta(t) \in \big[U^{-}_{n,m}(t;h,h_d), U^{+}_{n,m}(t;h,h_d) \big]\text{ for all } t \in [0{,}1]\big) \longrightarrow 1-\tau.
    \end{equation}
\end{remark}

\begin{remark}[Two-sample problem]
    The convergence in \eqref{eq:std statistic 1} allows to construct a test for a constant difference $\delta$ of the mean functions in \eqref{eq:main.4},  that is for $\Delta=0$ where $\Delta$ is given in \eqref{eq:centereddiff}. This hypothesis can be tested at asymptotic level $\tau$ by checking whether the confidence band in \eqref{eq:asympconfband} contains the zero-function.
\end{remark}

\section{Numerical implementation and simulations}\label{section:quantiles}

In this section we discuss the numerical implementation of our methods and conduct a simulation study. We use local quadratic weights in the linear estimators in \eqref{eq:estimatormu} and \eqref{eq:estimators}, which were computed using the \texttt{R} package \texttt{locpol}. The weights in the estimates of the covariance kernel functions in (\ref{eq:Gamma_estimator}) are obtained by the restricted bivariate local polynomial weights of degree one from \citet{berger2025}. They are computed using the \texttt{biLocPol} \texttt{R} package\footnote{ contained in GitHub repository github.com/mbrgr/Optimal-Rates-Covariance-Kernel-Estimation-in-FDA \href{https://github.com/mbrgr/Optimal-Rates-Covariance-Kernel-Estimation-in-FDA}{[Link]}}. We extend this to calculate the weights in the estimates of the lagged covariance kernel functions in (\ref{eq:crosscov}). For bandwidth selection a $K$-fold cross validation is implemented for independent data. Furthermore considering time series a hv-block cross validation is adopted as suggested e.g.~in \citet{racine2000} and \citet{Roberts2017} which is the main method used in this section. The hv-block cross validation partition is shifted cyclically $R$ times to reduce the dependence on a particular partition. We refer to this as $R$-shifted hv-block cross-validation.
The \texttt{R}-code for the simulations, the data application as well as for computing lagged covariance kernels are available in a GitHub repository\footnote{github.com/KevinWilk/Beyond-average-warming \href{https://github.com/KevinWilk/Beyond-average-warming}{[github]}}.

In Section \ref{subsection:multiplier bootstrap} for independent samples as well as for functional time series we discuss how to estimate the asymptotic quantiles as required in construction of confidence bands using the multiplier bootstrap. Then in Section \ref{subsection: Simulation} we present the results of a simulation study.
The numerical experiments in this and the next section were performed using the Marburg Compute Cluster (MaRC3a) at the Philipps-Universit\"at Marburg.%

\subsection{Quantile estimation via the multiplier bootstrap}\label{subsection:multiplier bootstrap}

The multiplier bootstrap is a standard resampling tool, see e.g.~\citet{vaart1996} or \citet{Chernozhukov2013} for expositions and results. In the setting of functional data analysis, it has been used, e.g.~in \citet{chang2009bootstrapping}, \citet{Telschow2022} and \citet{dette2020}.

To obtain good finite sample approximations we work with the setting $n \simeq {m}$ with $ n/m \approx \lambda$ in \eqref{eq:sameorder}. With $\hat{\mu}^{\sparse}_{n,m}(\cdotp;h,h_d)\defeq \hat{\delta}_{n,m}(\cdotp;h,h_d)+\hat{\mu}^{\dense}_{m}(\cdotp;h_d)$ setting
\begin{align*}
    \hat{X}^{\sparse}_{n,i}(t)&\defeq \sum_{j=1}^p w^{\sparse}_{j}(t;{h})\big(Y^{\sparse}_{i,j}-\hat{\mu}^{\sparse}_{n,m}(t_j^{\sparse};h,h_d)\big) \quad \text{ and }\\
    \hat{X}^{\dense}_{{m},k}(t) &\defeq  \sum_{j=1}^p w^{\sparse}_{j}(t;{h}) \sum_{l=1}^{{p_d}}w^{\dense}_{l}(t^{\sparse}_{j};{h_d})\big(Y^{\dense}_{k,l}-\hat{\mu}^{\dense}_{{m}}(t_l^{\dense}; h_d) \big),
\end{align*}
the multiplier bootstrap generates samples $r=1, \ldots, N^*$ for the limit process $\mathcal{G}$ in \eqref{eq:std statistic 1} (situation $n \simeq {m}$) according to
\begin{align}
    B^{(r)}(t)\defeq & \frac{1}{\sqrt{n}\,\hat{\sigma}_{\Delta}(t)}\,\Big(\sum_{i=1}^n\big(\hat{X}^{\sparse}_{n,i}(t)-\int_0^1 \hat{X}^{\sparse}_{n,i}(y)\,dy\big)\,\xi_{i}^{(r)} \nonumber\\
    & \qquad +n/{m}\, \sum_{k=1}^{{m}}\big(\hat{X}^{\dense}_{{m},k}(t)-\int_0^1 \hat{X}^{\dense}_{{m},k}(y)\,dy\big)\,\eta_{k}^{(r)}\Big). \label{eq:multiboot}
\end{align}

Here in \eqref{eq:multiboot}, in case of independent samples the $\xi_{i}^{(r)}$ and the $\eta_{k}^{(r)}$ are chosen as independent, identically distributed standardized random variables e.g. normally distributed or Rademacher distributed. For time series the multipliers need to be dependent, and we use the tapered block multipliers recommended in \citet{Buecher2013} with block length $l(n) = \lfloor 2\,n^{1/3}\rfloor$, see Section~\ref{sec:multipliers}. We scale \eqref{eq:multiboot} with $\sqrt{n-1}$ instead of $\sqrt{n}$ as recommended in \citet{Telschow2022}, and estimate $q_{1-\tau}$ in Remark~\ref{rem:bands} by the empirical $(1-\tau)$-quantile $\hat{q}_{1-\tau}$ of $\sup_{t\in[0,1]}|B^{(r)}(t)|$, $r=1,\ldots,N^*$. For fully observed curves, the multiplier bootstrap is justified under independence by the conditional multiplier central limit theorem \citep[Theorem 2.9.6]{vaart1996}, and a closely related block multiplier bootstrap for $\varphi$-mixing functional time series by \citet[Theorem 3.3]{dette2020}. We do not prove its validity for the discretely observed, smoothed and studentized processes in \eqref{eq:multiboot}, but investigate its performance in the simulations below.

\subsection{Simulation results}\label{subsection: Simulation}

As data generating process we consider a family of stationary Ornstein–Uhlenbeck (OU) processes with AR(1)-type dependence of the Brownian motions across curves.
More precisely, for parameters $\sigma,\theta>0$ we set
\begin{align*}
    Z^{\sparse}_{i}(t) &= Z^{\sparse}_i(0) e^{-\theta t}  + \sigma \int_0^t e^{-\theta(t-s)} \rd B^{\sparse}_{i;s},\\
    Z^{\dense}_{k}(t) &= Z^{\dense}_k(0) e^{-\theta t}  + \sigma \int_0^t e^{-\theta(t-s)} \rd B^{\dense}_{k;s},
\end{align*}
with $Z^{\sparse}_i(0),Z^{\dense}_k(0) \sim \mathcal{N}(0,\sigma^2/(2\theta))$ and $(B^{\sparse}_{i;s})_{s\geq 0},(B^{\dense}_{k;s})_{s\geq 0}$ being standard Brownian motions with
\begin{align*}
    \cov(B^{\sparse}_{i;t},B^{\sparse}_{i+b;s}) = \cov(B^{\dense}_{k;t},B^{\dense}_{k+b;s}) = \rho^{|b|}_B\,\min(t,s),
\end{align*}
for $|\rho_B|< 1$ and $i,k \in \N$.
Their lagged covariance kernels are given by
\begin{align*}
    \gamma^{\sparse}_{|i-j|}(t,s) = \gamma^{\dense}_{|i-j|}(t,s) =   \frac{\sigma^2}{2 \theta} \Big(\rho_B^{|i-j|}\Big( e^{-\theta|t-s|}-e^{-\theta(t+s)} \Big) + \one_{i=j} \,e^{-\theta(t+s)} \Big). %
\end{align*}
Since $|\rho_B|< 1$ their long-run covariance kernels exist, i.~e.~$ ||\Gamma^\sparse||_{\infty}, ||\Gamma^\dense||_{\infty} < \infty.$
In the simulation we take $\rho_B = 0{.}5$, $\theta=1$, $\sigma = 4$ and let $\epsilon_{i,j}^{\sparse},\epsilon_{k,l}^{\dense} \sim \mathcal{N}(0,\sigma_{\epsilon}^2)$ with $\sigma_{\epsilon}=0{.}1$, and
\begin{align*}
    &\mu^{\sparse}(t)\defeq 3\sin\big(1{.}5\,\pi(2t-1)\big)e^{-2|2t-1|},\\
   & \mu^{\dense}(t) = \mu^{\sparse}(t)-\delta(t),
\end{align*}
with $\delta$ specified below. Furthermore, for the long run covariance kernels we choose maximal lags of $L, L_d = 3$. In both following subsections, we set $n/m = \lambda=5/6$.

\subsubsection{Size of the test for $\Delta=0$ and coverage of confidence bands if $\Delta=0$}\label{sec:simconstant}
In this setting, we choose
\begin{align*}
\delta \equiv 2\in \R,
\end{align*}
and consider equidistant grid points at $t^{\sparse}_{j} = (j - 1/2)/p$ for $j = 1,\hdots,p$, respectively $t^{\dense}_{l} = (l - 1/2)/p_d$ for $l = 1,\hdots,p_d.$

For $N=1000$ sample repetitions Table~\ref{tab:bandwidth mean delta constant} in Section~\ref{sec:additionalsim} in the supplement  provides averaged results of the $5$-shifted hv-block cross validation in a $5$-fold framework applied on $(Y^{\sparse}_{i,j}, t^{\sparse}_{j})$ and $(Y^{\dense}_{k,l},t^{\dense}_l)$ with $n\in\{25,50,100,300,600,1000\}$ and $m\in\{30,60,120,360,720,1200\}$ random curves as well as a design grid of $p\in \{25,50,75\}$ and $p_d = 100$. The bandwidth selection is discussed further in Section~\ref{section: hv cross validation} of the supplement. With an increasing number of curves, the selected bandwidth for estimating $\mu^{\dense}$ decreases, whereas the averaged bandwidth for estimating $\delta$ remains almost constant as the design points increase. The bandwidths selected for smoothing the sparse sample on its own, as $\bar{\delta}_{n,m}$ in \eqref{eq:diffest} requires, are much smaller (Table~\ref{tab:bwseparate}). Bandwidths selected for estimating the long run covariance kernel by $5$-shifted $5$-fold hv-block cross validation are displayed in Figure~\ref{fig:combined} in the supplement, Section~\ref{sec:additionalsim}.

With $n=25,p=25$ and $m = 30,p_d =100$ and bandwidths chosen from Table~\ref{tab:bandwidth mean delta constant}, Figure \ref{fig:Simulation Setup Delta constant} illustrates a data setup by displaying the simulated curves and corresponding local polynomial fits as well as the underlying mean functions, respectively difference function $\delta$.
\begin{figure}[h!]
    \centering
    \begin{subfigure}{0.32\linewidth}
        \includegraphics[width=\linewidth]{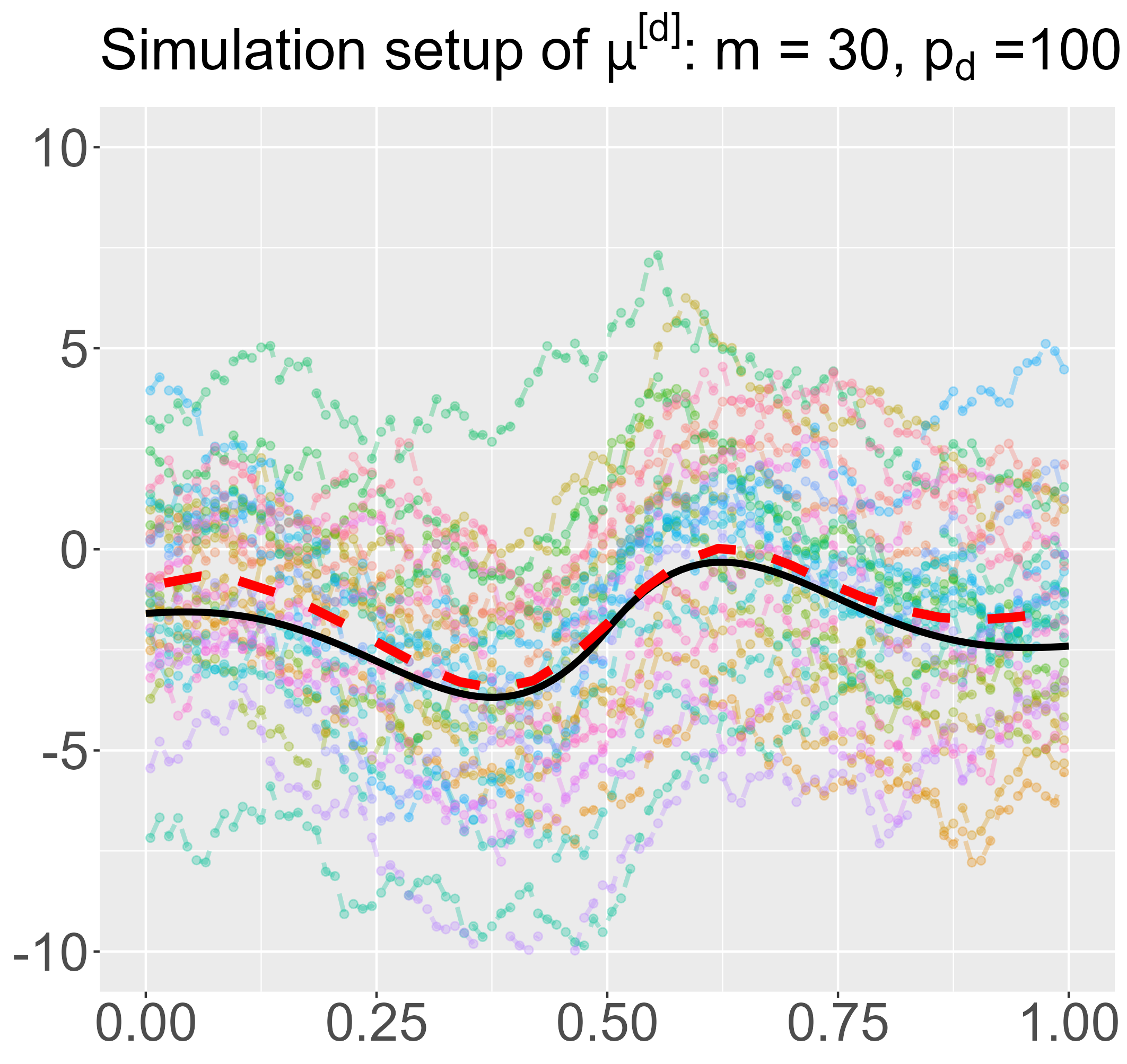}
        \caption{Dense regime}
    \end{subfigure}
    \hfill
     \begin{subfigure}{0.32\linewidth}
        \includegraphics[width=\linewidth]{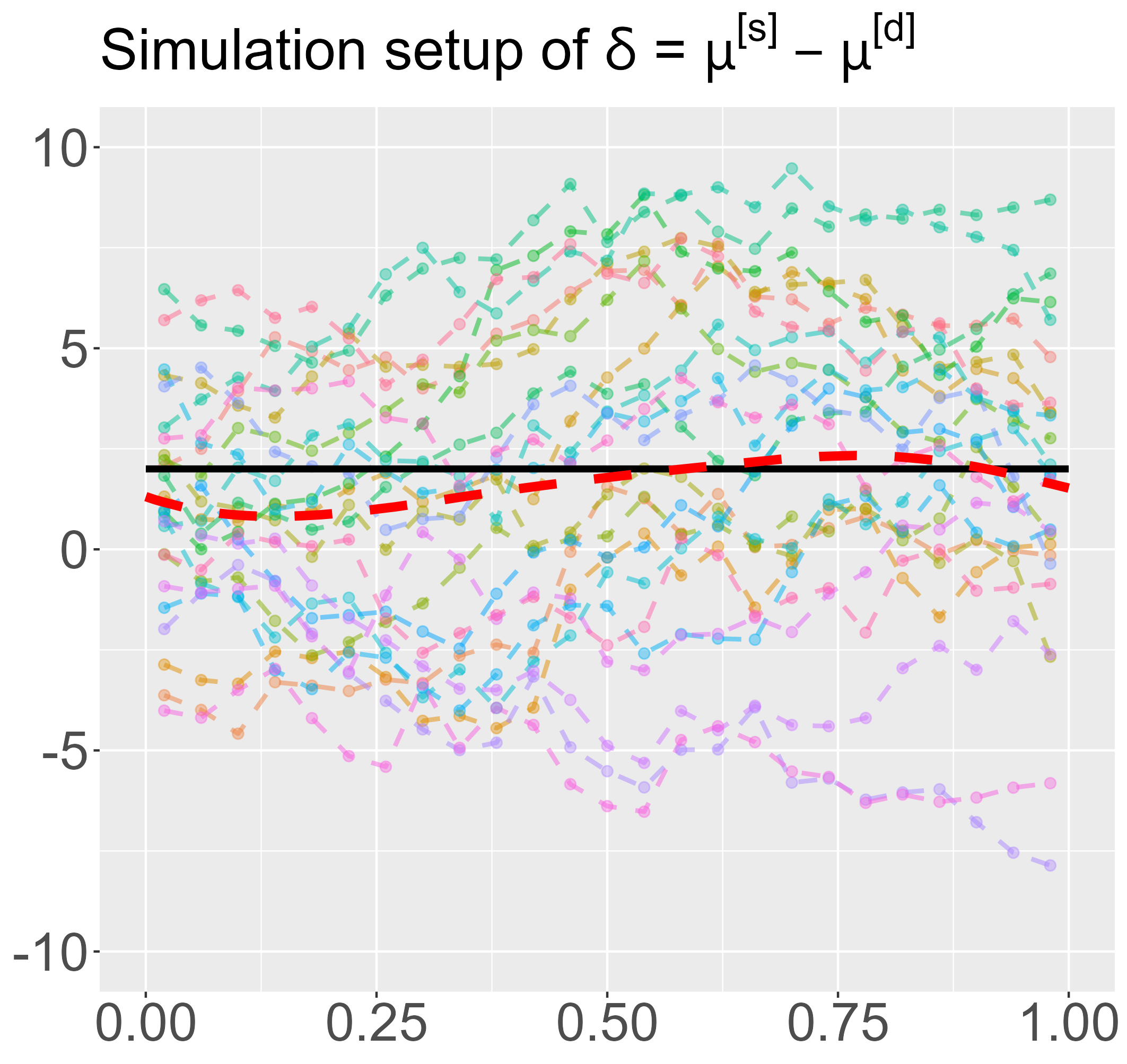}
        \caption{Sparse regime}
    \end{subfigure}
    \hfill
    \begin{subfigure}{0.32\linewidth}
        \includegraphics[width=\linewidth]{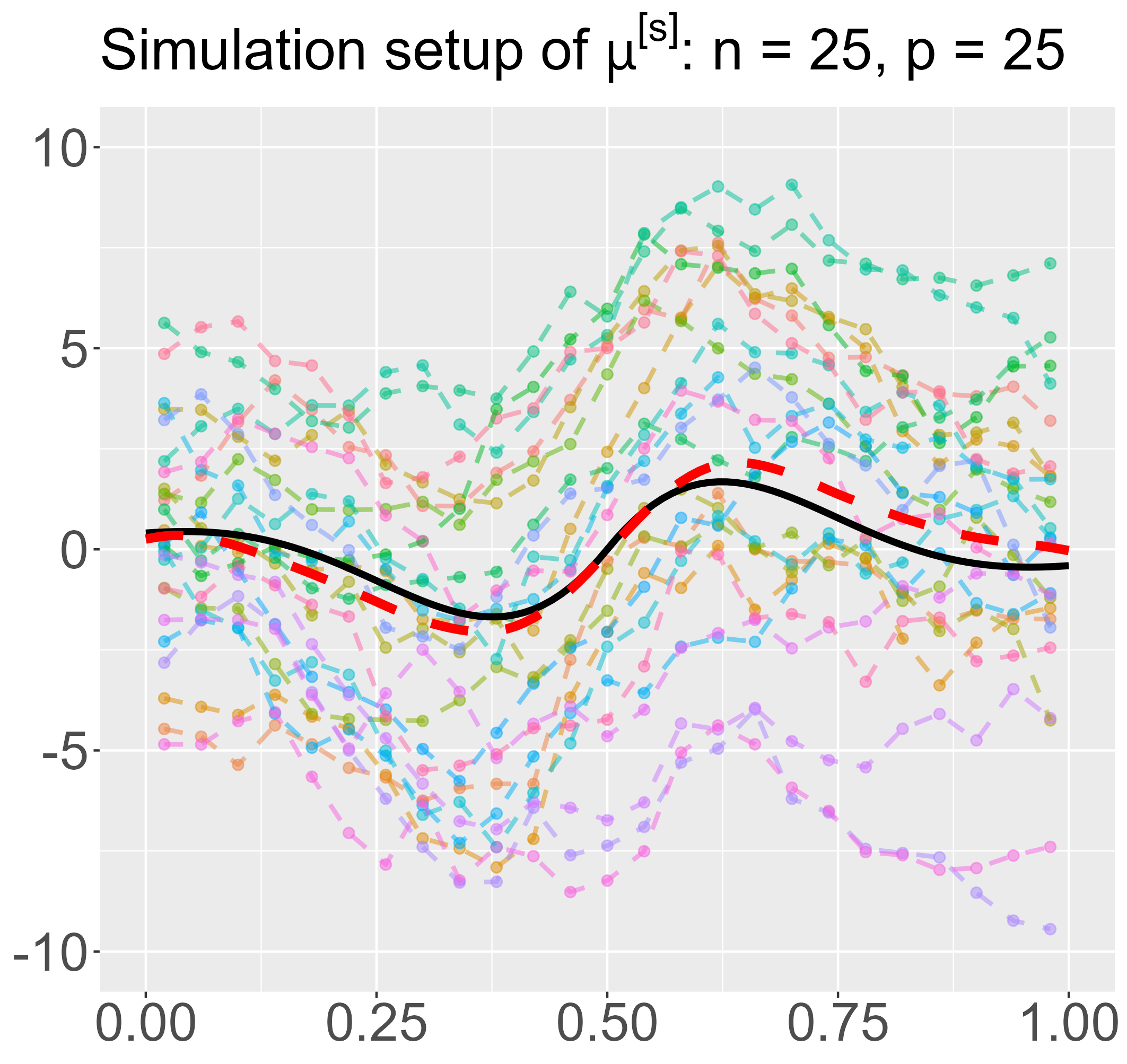}
        \caption{Sparse regime}
    \end{subfigure}
    \caption{Estimator: (a) $\hat{\mu}^{\dense}_{m}(\cdotp;0{.}22)$, (b) $\hat{\delta}_{n,m}(\cdotp; 0{.}67,0{.}22)$ and (c) $\hat{\mu}^{\sparse}_{n,m}(\cdotp;0{.}67,0{.}22)\defeq \hat{\delta}_{n,m}(\cdotp;0{.}67,0{.}22)+\hat{\mu}^{\dense}_{m}(\cdotp;0{.}22)$  (red dashed line), function: (a) $\mu^{\dense}$, (b) $\delta$ and (c) $\mu^{\sparse}$ (black line).} \label{fig:Simulation Setup Delta constant}
\end{figure}

Next, based on standardized asymptotic normality in \eqref{eq:sameorder} and similar in \eqref{eq:std statistic 1} in case of time series we use the dependent multiplier bootstrap with $N^*=1000$ to obtain estimates of those $(1-\tau)$-quantiles, denoted by $\hat{q}_{1-\tau}^{\text{\tiny DMB}}$ for $\tau \in (0{,}1)$. Figure \ref{fig:coverage rate under H0} shows the empirical coverages of $1000$ repeated simulation setups with each parameter combination of $n\in\{25,50,100,200,300,400,500,600,700,800,900,1000\}$, $p\in \{25,50,75\}$ as well as $m = 1{.}2 \, n$ and fixed $p_d = 100$, where the cross validation reruns for each replication. The nominal levels are kept reasonably well, in particular for the larger sample sizes.
\begin{figure}[h!]
    \centering
    \begin{subfigure}[b]{0.95\linewidth}
        \includegraphics[width=\linewidth]{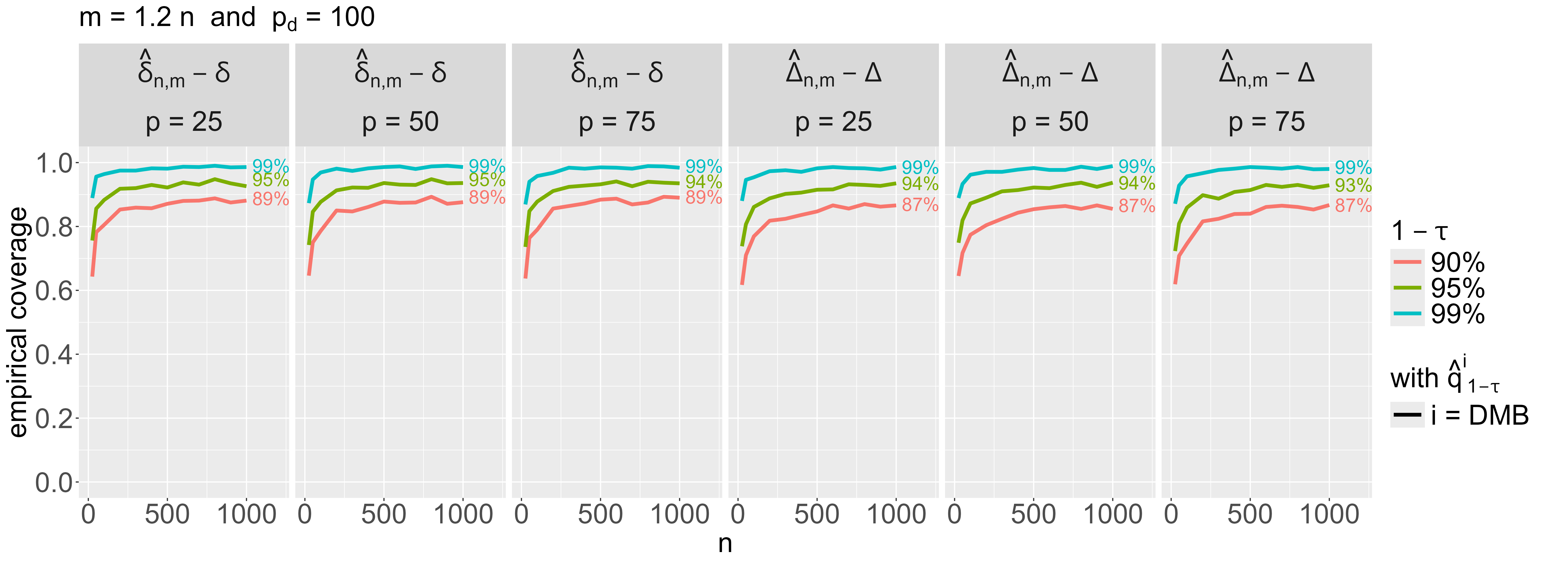}
    \end{subfigure}
    \caption{Constant difference function: Simulated coverage and level of the test using the  dependent multiplier bootstrap with $N^* = 1000$ on $1000$ repeated simulated samples based on the time-series version of \eqref{eq:sameorder} and \eqref{eq:std statistic 1}.}
    \label{fig:coverage rate under H0}
\end{figure}

\subsubsection{Power of the test and coverage of confidence bands for $\Delta \neq 0$}\label{sec:simnonconstant}

In the alternative scenario we consider the difference function
\begin{align*}
    \delta(t)\defeq 2-\sin\big(\pi(2t-1)\big)e^{-2|2t-1|}.
\end{align*}
Like the mean functions, this $\delta$ has a jump in its second derivative at $t = 1/2$, so that $\alpha_\delta = \alpha = 2$, but its H\"older constant $M_\delta$ is about five times smaller than $M$, the Hölder constant of $\mu^{\sparse}$. This setting thus corresponds to  Remark~\ref{rem:alphadelta}, with less pronounced structure in $\delta$ than in the mean functions due to $M_\delta \ll M$. Tables~\ref{tab:bandwidth mean delta not constant} and \ref{tab:bandwidth cov delta not constant} as well as Figure \ref{fig:Simulation Setup Delta not constant} are analogous to those in the previous subsection. Again the selected bandwidths are larger for the  difference function than for $\mu^{\dense}$. %

\begin{figure}[h!]
    \centering
    \begin{subfigure}{0.32\linewidth}
        \includegraphics[width=\linewidth]{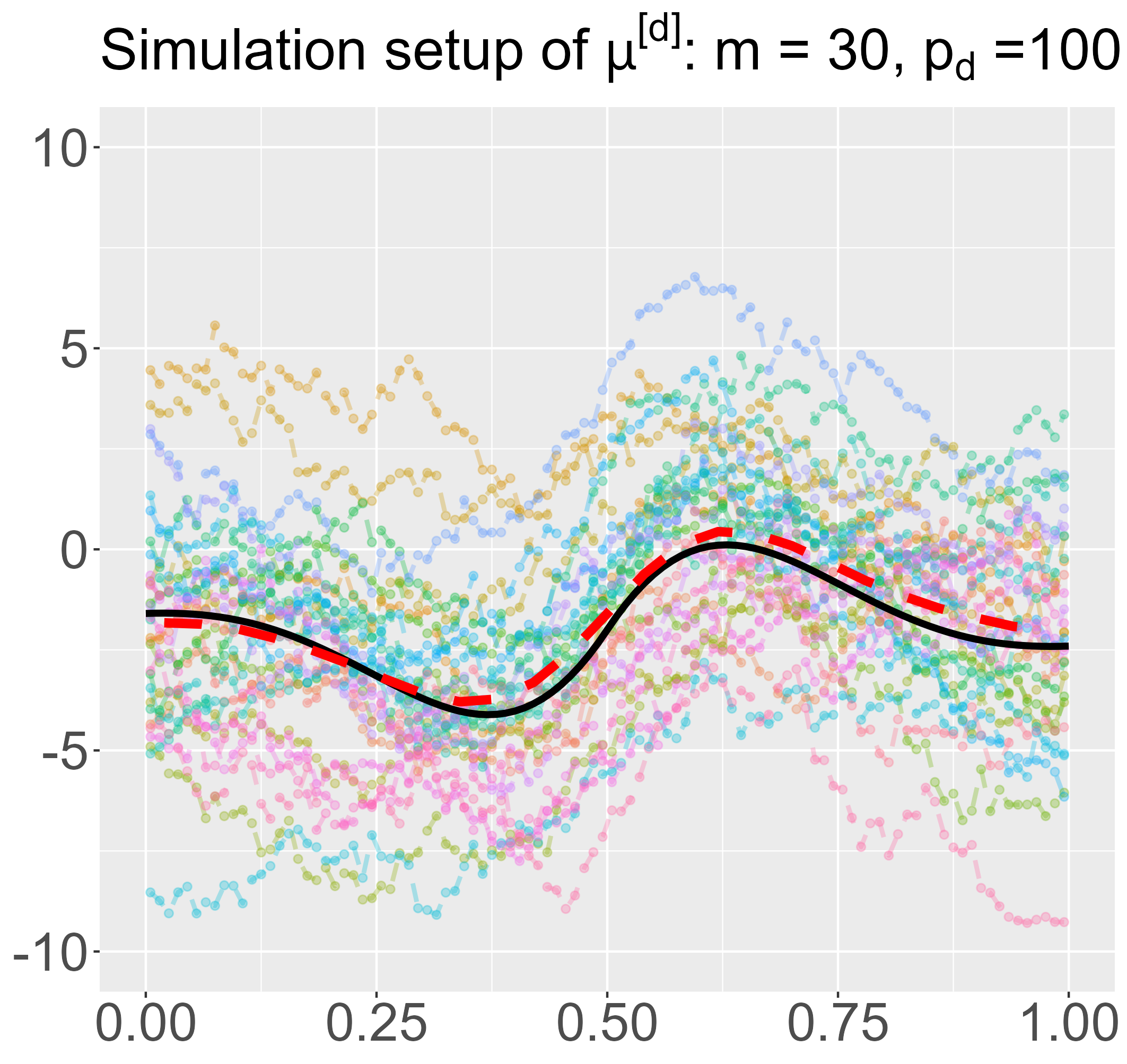}
        \caption{Dense regime}
    \end{subfigure}
    \hfill
     \begin{subfigure}{0.32\linewidth}
        \includegraphics[width=\linewidth]{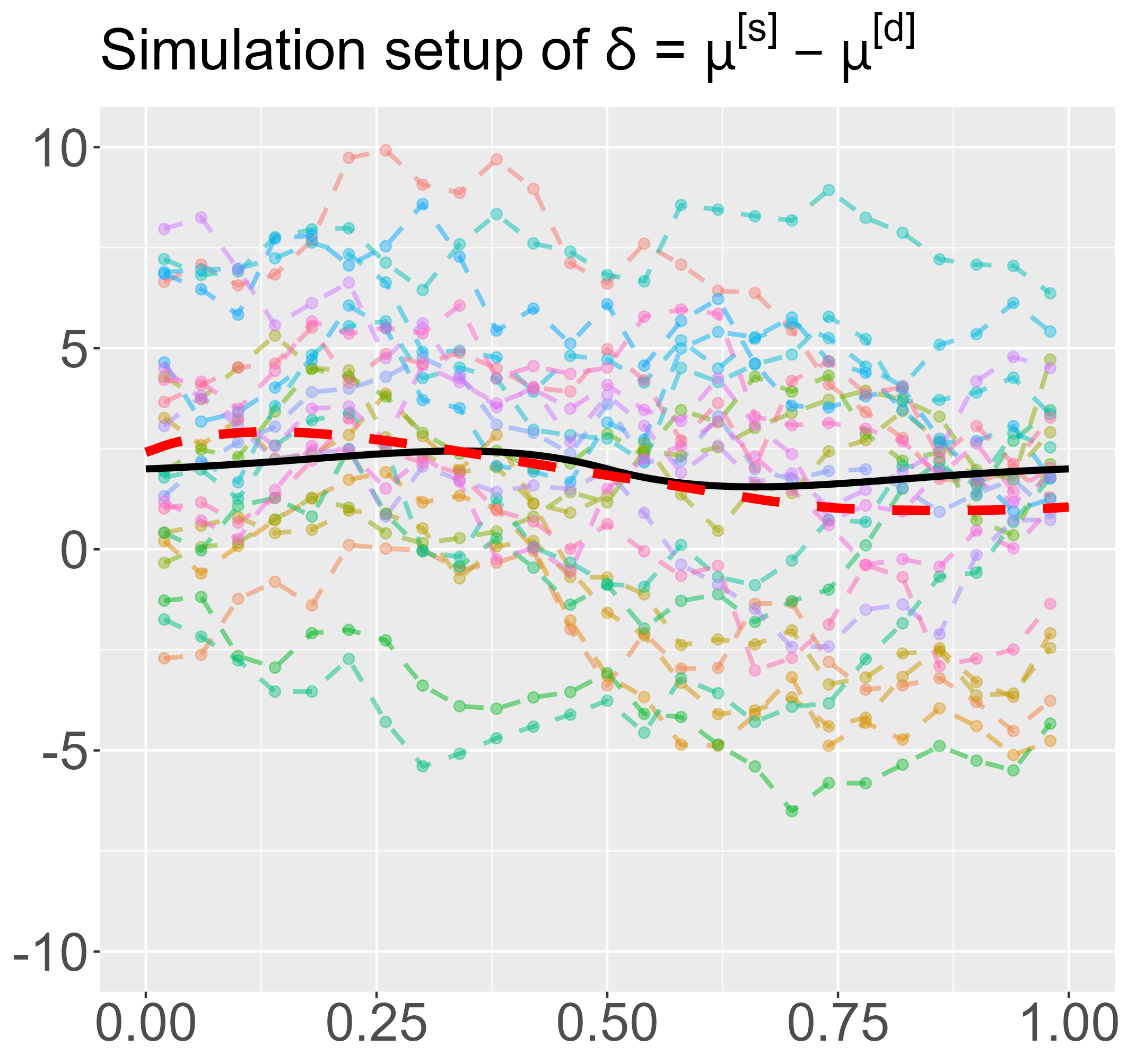}
        \caption{Sparse regime}
    \end{subfigure}
    \hfill
    \begin{subfigure}{0.32\linewidth}
        \includegraphics[width=\linewidth]{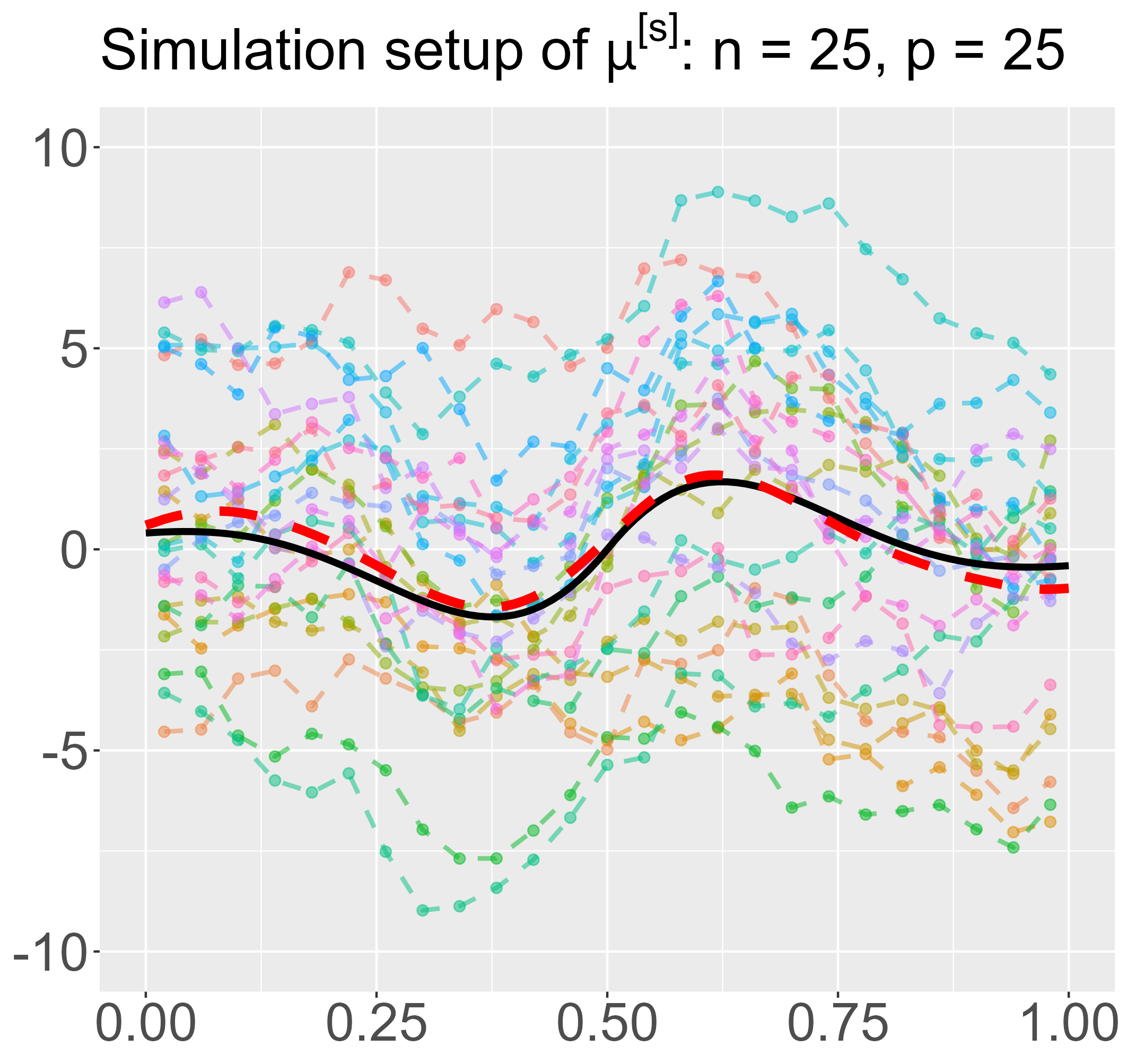}
        \caption{Sparse regime}
    \end{subfigure}
    \caption{Estimator: (a) $\hat{\mu}^{\dense}_{m}(\cdotp;0{.}21)$, (b) $\hat{\delta}_{n,m}(\cdotp;0{.}61,0{.}21)$ and (c) $\hat{\mu}^{\sparse}_{n,m}(\cdotp;0{.}61,0{.}21)\defeq \hat{\delta}_{n,m}(\cdotp;0{.}61,0{.}21)+\hat{\mu}^{\dense}_{m}(\cdotp;0{.}21)$  (red dashed line), function: (a) $\mu^{\dense}$, (b) $\delta$ and (c) $\mu^{\sparse}$ (black line).}  \label{fig:Simulation Setup Delta not constant}
\end{figure}

Similar to Figure \ref{fig:coverage rate under H0}, under a non constant difference function $\delta$ the empirical coverage rates and the empirical powers are plotted in Figure \ref{fig:coverage and power}. The coverage is slightly below that for a constant difference but still reasonably close to the nominal level.

\begin{figure}[h!]
    \centering
    \begin{subfigure}[b]{0.9\linewidth}
        \includegraphics[width=\linewidth]{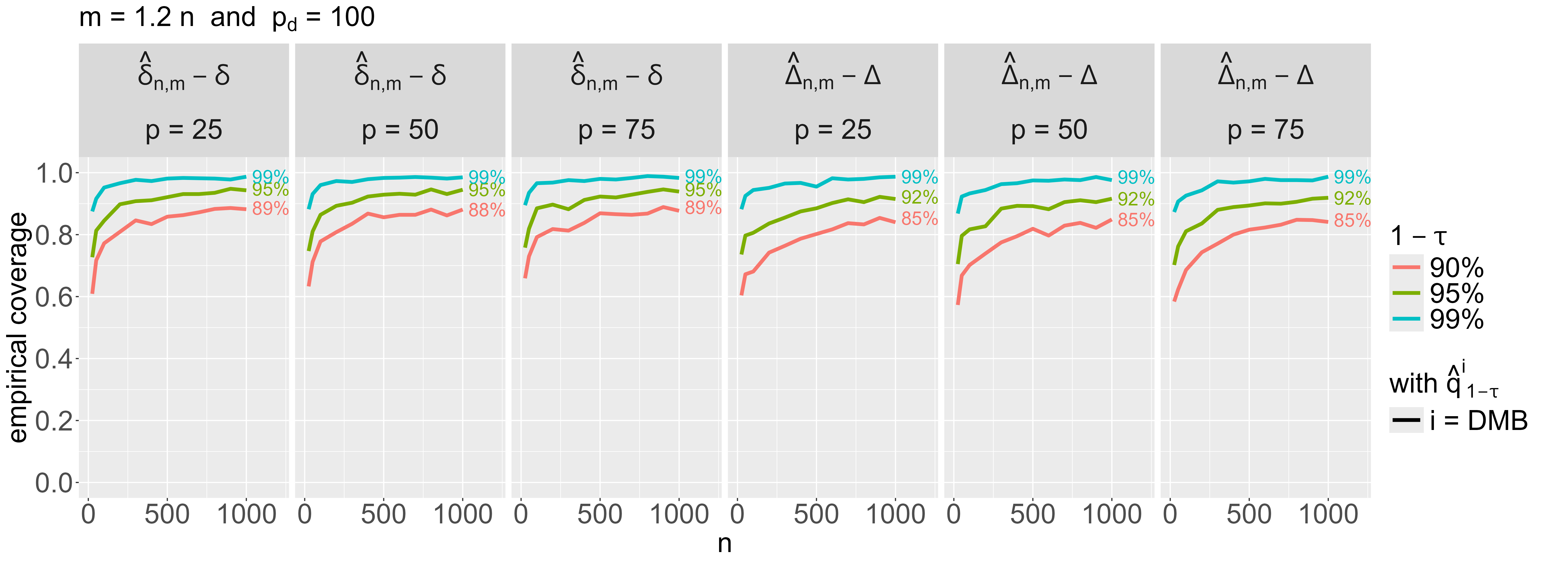}
    \end{subfigure}

    \hfill

    \begin{subfigure}[b]{0.9\linewidth}
        \includegraphics[width=\linewidth]{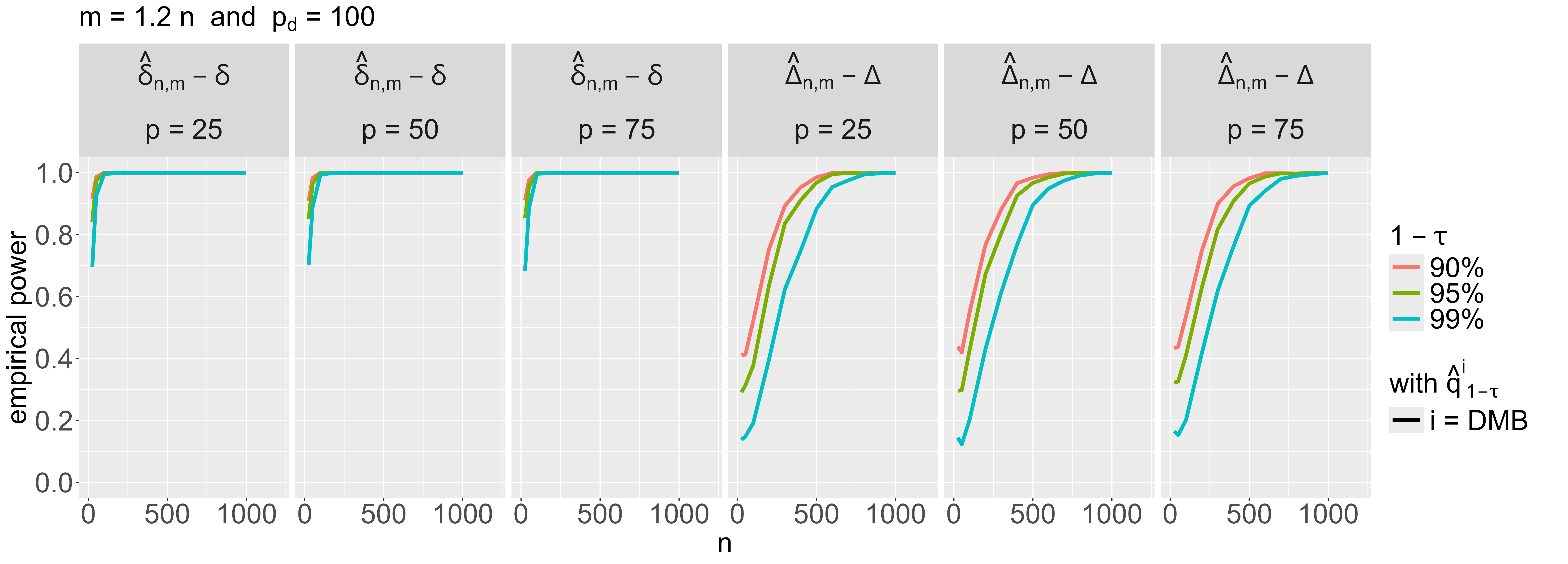}
    \end{subfigure}

    \caption{Non-constant difference function: Simulated coverage of the confidence band (upper panel), and power of test for constant difference (lower panel) using the dependent multiplier bootstrap with $N^* = 1000$ on $1000$ repeated simulated samples analogously to Figure \ref{fig:coverage rate under H0}.}
    \label{fig:coverage and power}
\end{figure}

Figure \ref{fig:comparison} highlights in one simulated sample why the implementation of dependent multipliers is necessary compared to independent ones for the multiplier bootstrap in \eqref{eq:multiboot}: In the setting $n = 600, p = 25$ and $m = 720, p_d = 100$ the $95\%$-quantiles are considerably underestimated using  independent multipliers, resulting in $(\hat{q}^{\text{\tiny IMB}}_{0{.}95}=2{.}09 \text{ (not centered difference: $\delta$)}$ and $\hat{q}^{\text{\tiny IMB}}_{0{.}95}=1{.}93 \text{ (centered difference: $\Delta$)})$. For dependent multipliers, the estimated quantiles $(\hat{q}^{\text{\tiny DMB}}_{0{.}95}=2{.}7$ (not centered difference: $\delta$) and $\hat{q}^{\text{\tiny DMB}}_{0{.}95}=3{.}11 \text{ (centered difference: $\Delta$)})$ are closer to the empirical quantiles $(\hat{q}_{0{.}95}=3{.}03 \text{ (not centered difference: $\delta$) and } \hat{q}_{0{.}95}=3{.}35 \text{ (centered difference: $\Delta$)})$ obtained from $1000$ simulation repetitions.

\begin{figure}[h!]
    \centering
    \begin{subfigure}[b]{0.49\linewidth}
        \includegraphics[width=\linewidth]{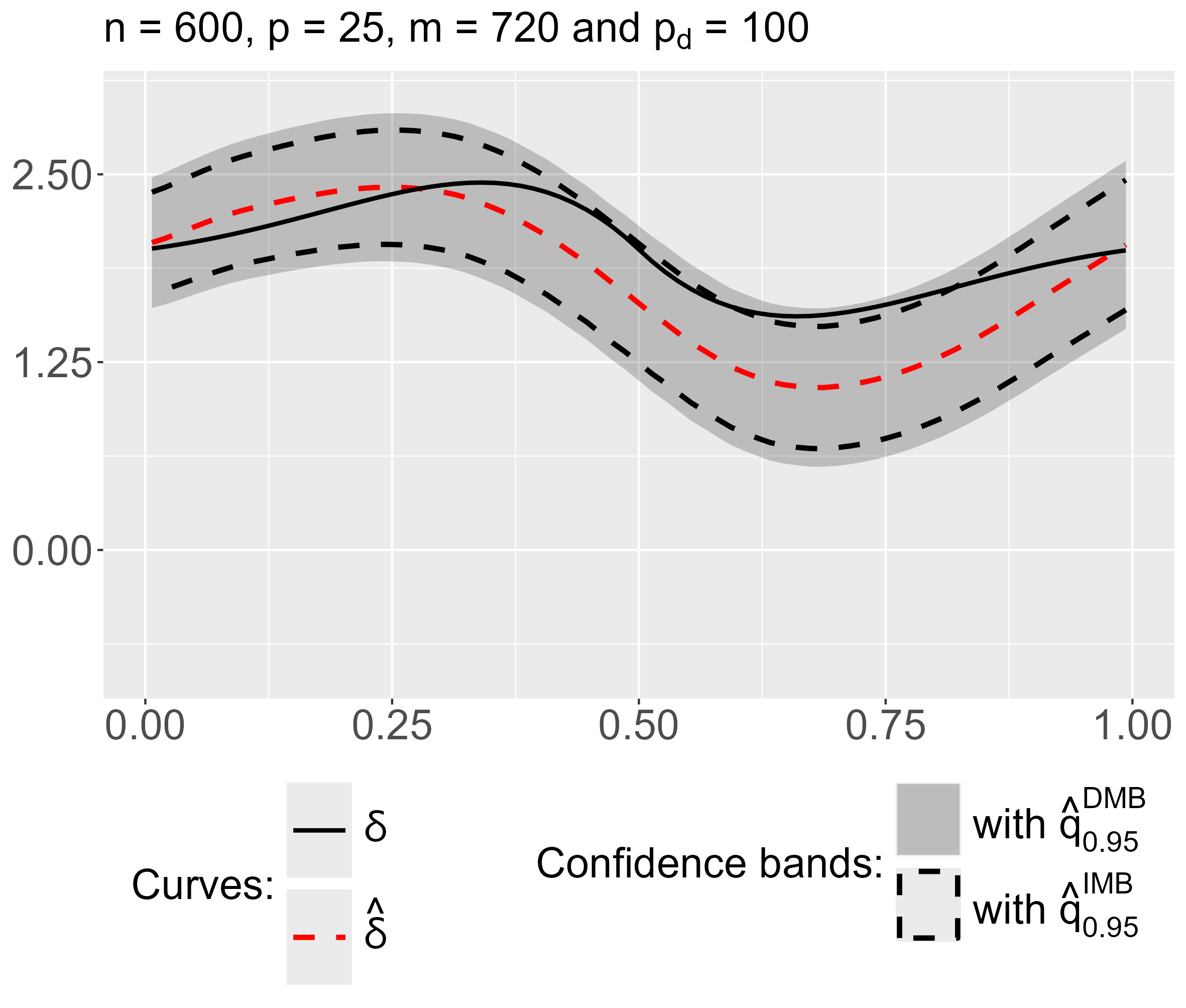}
    \end{subfigure}
    \hfill
    \begin{subfigure}[b]{0.49\linewidth}
        \includegraphics[width=\linewidth]{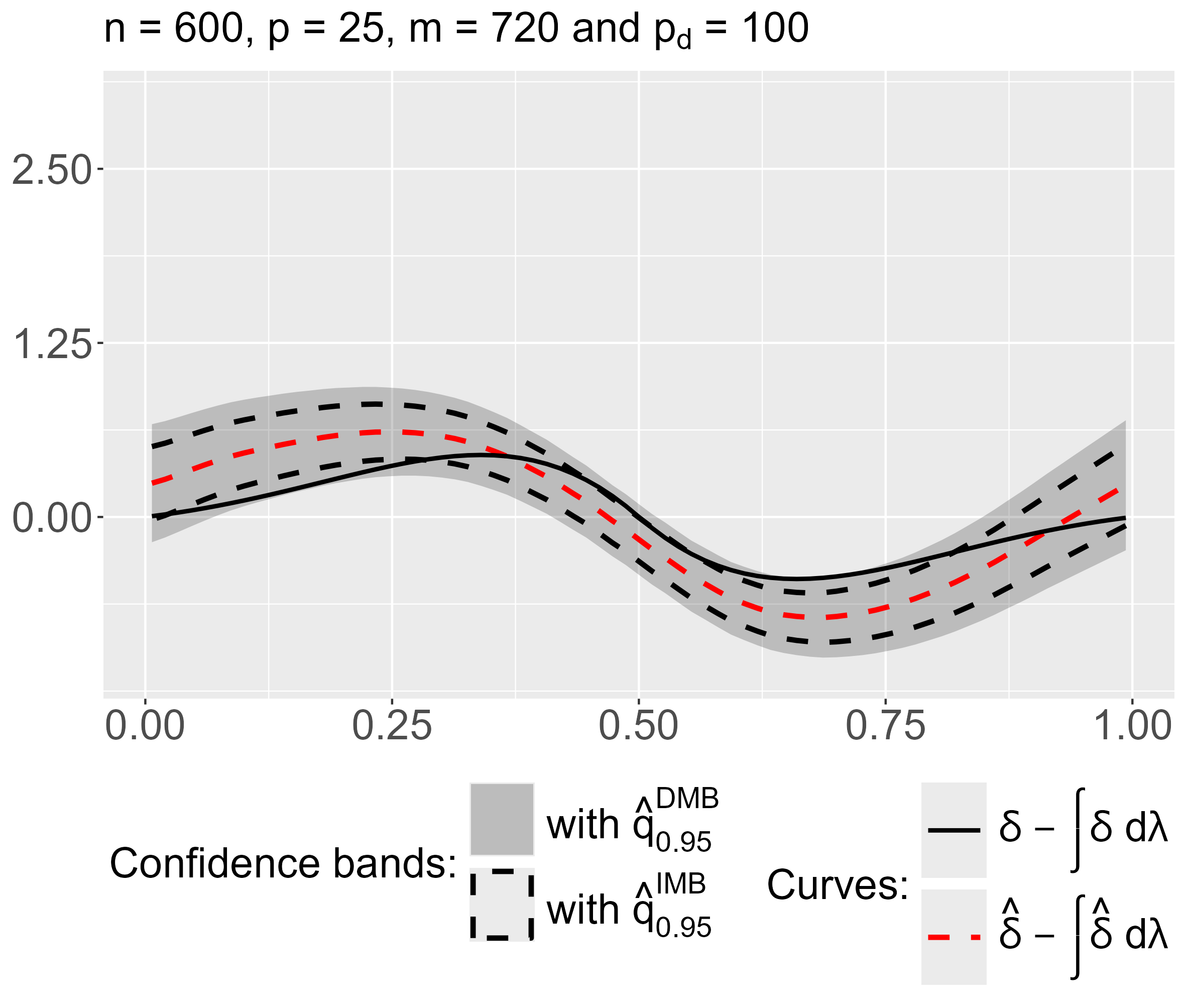}
    \end{subfigure}
\caption{Dependent and independent multiplier bootstrap: Comparison of confidence bands at level $95$\% with $N^*=1000$ based on one simulated sample.}
    \label{fig:comparison}
\end{figure}

\section{Comparing current and historical daily temperature time series}\label{section:Real data}

In this section we follow up on Section \ref{sec:apllilust} and present results for data for the three further major German cities Munich, Hamburg and Frankfurt am Main to highlight the distinct changes in the daily mean profile according to the climate region of the city.
For Frankfurt and Munich, part of the estimated difference may reflect the change of station, since the records of the two periods come from different stations, and for Frankfurt also the relocation of the station in 2014, see Section~\ref{sec:preprocessing}.
Again, for simplifying interpretation, as in \eqref{eq:diffappl} the difference function $\delta$ is defined by
$\delta = \mu^{\dense}-\mu^{\sparse}.$ Thus
$\delta>0$ indicates an average increase in daily weather temperature over time.

Figures \ref{fig:Muenchen} (Munich), \ref{fig:Hamburg} (Hamburg) as well as \ref{fig:Frankfurt} (Frankfurt) show plots analogous to those in Figures \ref{fig: means of Berlin} - \ref{fig:real data constant test} for Berlin.
\begin{figure}[h!]
    \centering
    \begin{subfigure}[b]{\linewidth}
        \centering
        \includegraphics[width=0.88\linewidth]{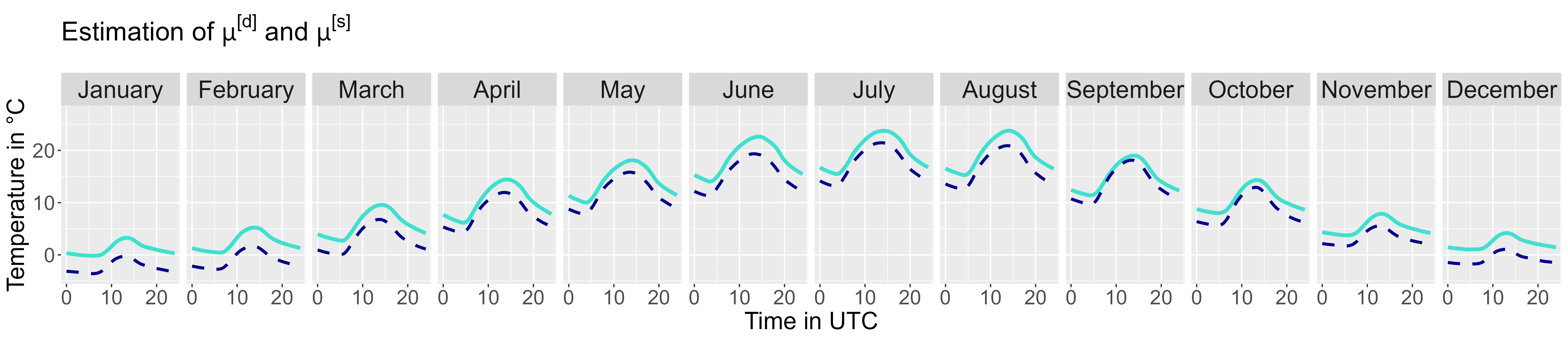}
    \end{subfigure}
    \begin{subfigure}[b]{\linewidth}
        \centering
        \includegraphics[width=0.88\linewidth]{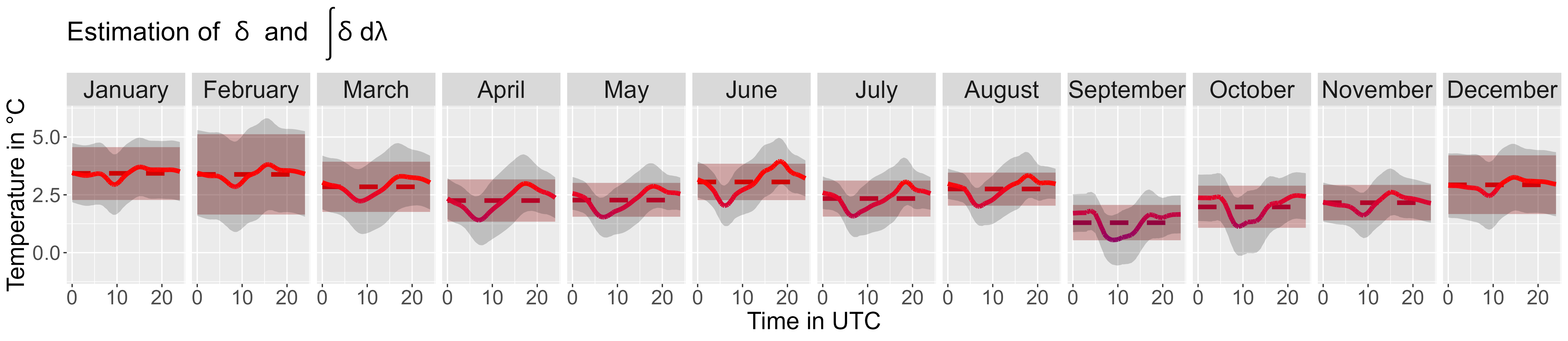}
    \end{subfigure}
    \hfill
    \begin{subfigure}[b]{\linewidth}
        \centering
        \includegraphics[width=0.88\linewidth]{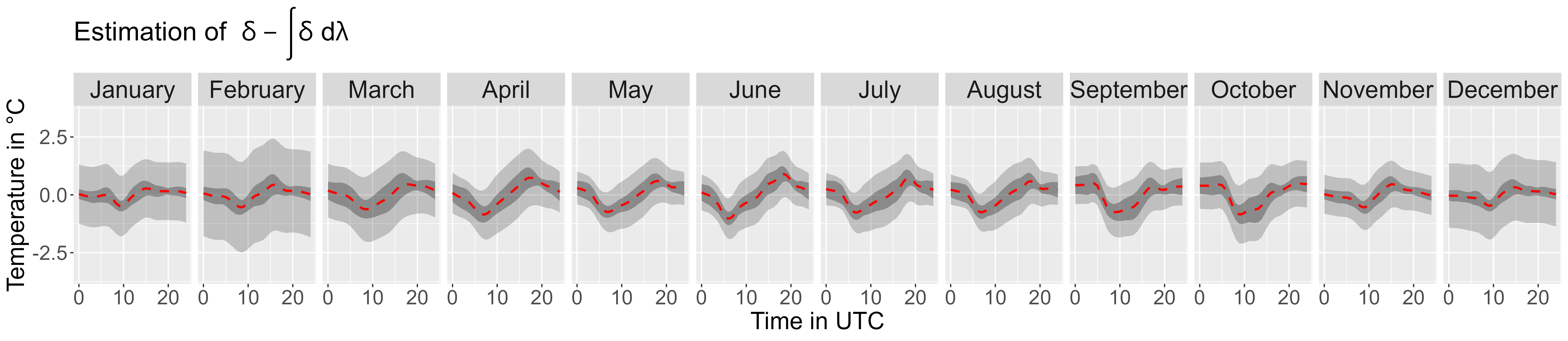}
    \end{subfigure}
\caption{Munich (Germany): Results analogously to Figures \ref{fig: means of Berlin} - \ref{fig:real data constant test}.}
    \label{fig:Muenchen}
\end{figure}
\begin{figure}[h!]
    \begin{subfigure}[b]{\linewidth}
        \centering
        \includegraphics[width=0.88\linewidth]{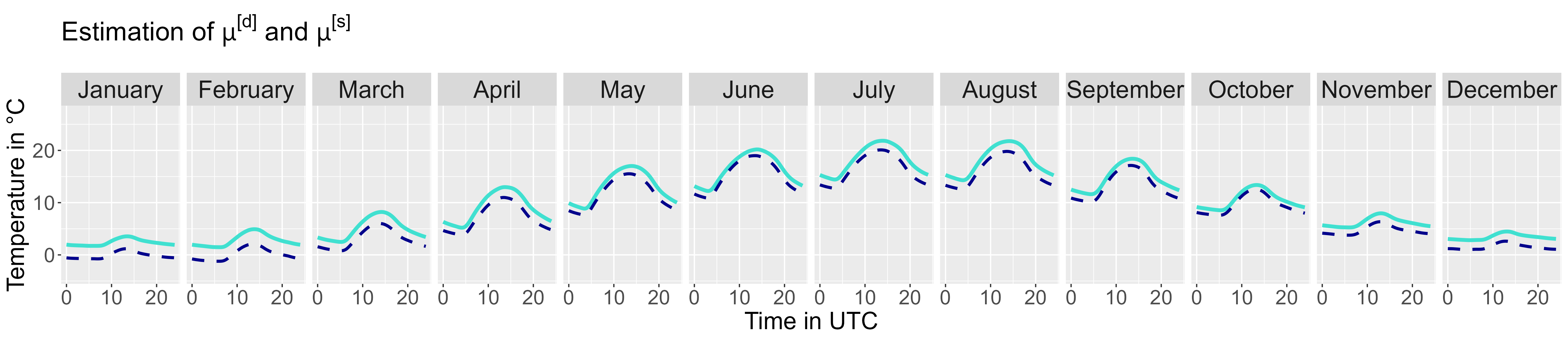}
    \end{subfigure}
    \begin{subfigure}[b]{\linewidth}
            \centering
        \includegraphics[width=0.88\linewidth]{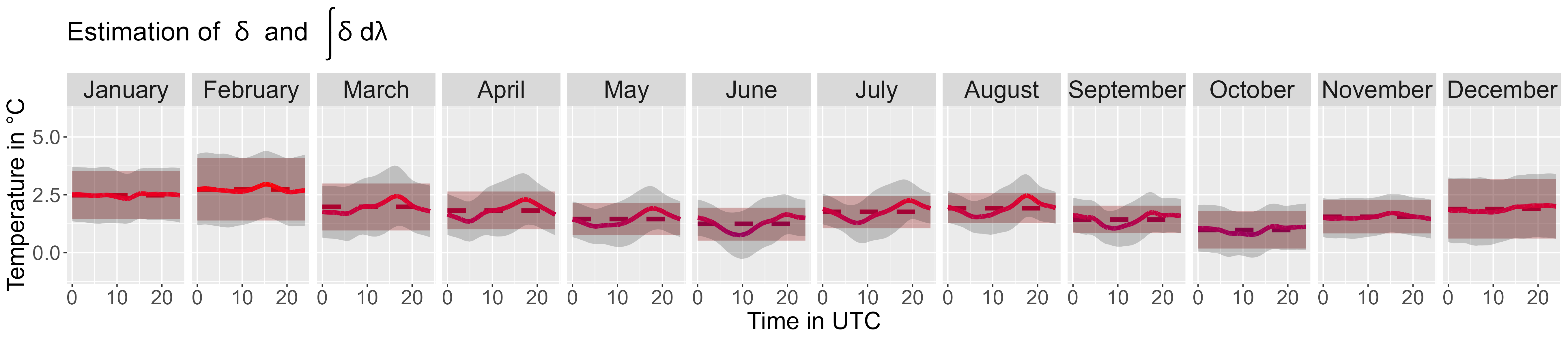}
    \end{subfigure}
    \begin{subfigure}[b]{\linewidth}
        \centering
        \includegraphics[width=0.88\linewidth]{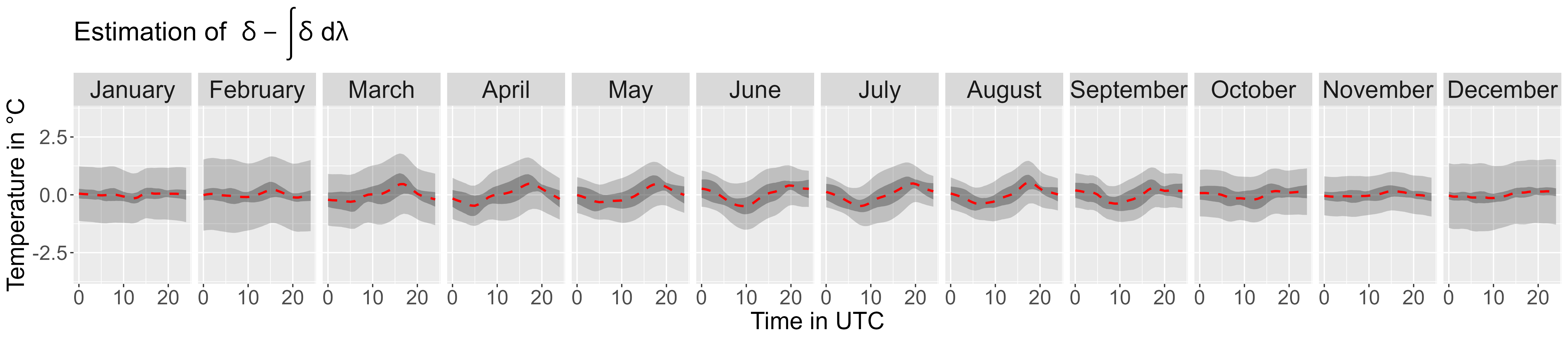}
    \end{subfigure}
\caption{Hamburg (Germany): Results analogously to Figures \ref{fig: means of Berlin} - \ref{fig:real data constant test}.}
    \label{fig:Hamburg}
\end{figure}
\begin{figure}[h!]
    \centering
    \begin{subfigure}[b]{\linewidth}
        \centering
        \includegraphics[width=0.88\linewidth]{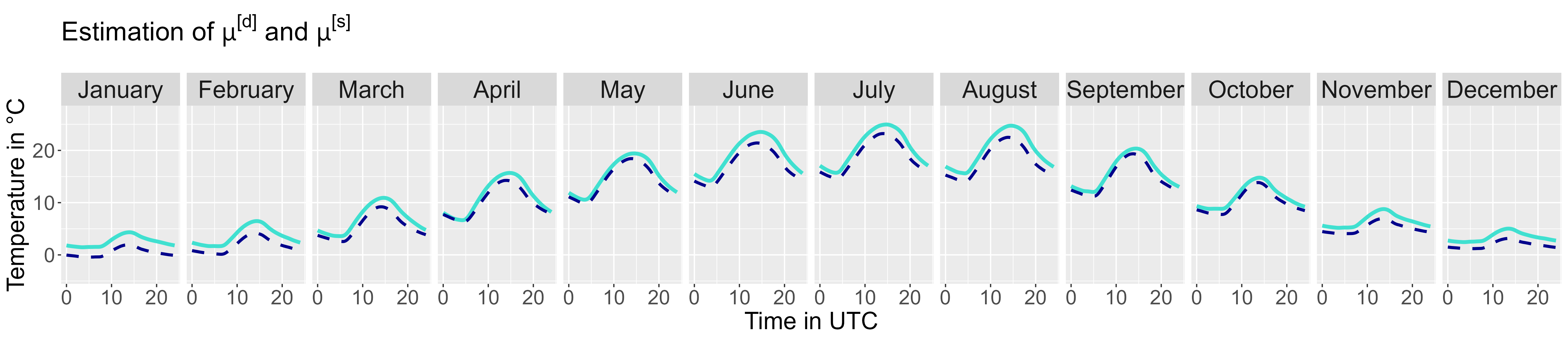}
    \end{subfigure}

    \begin{subfigure}[b]{\linewidth}
        \centering
        \includegraphics[width=0.88\linewidth]{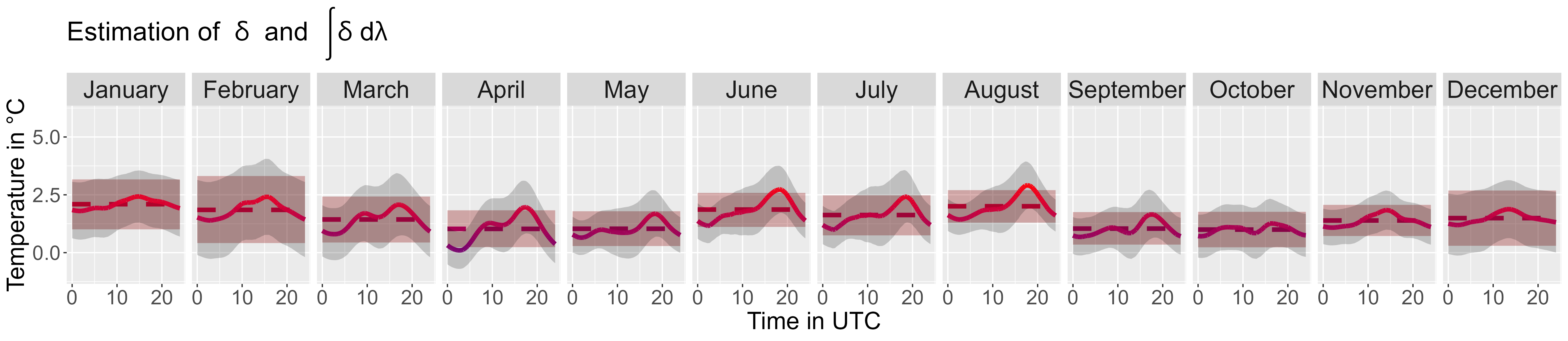}
    \end{subfigure}

    \begin{subfigure}[b]{\linewidth}
        \centering
        \includegraphics[width=0.88\linewidth]{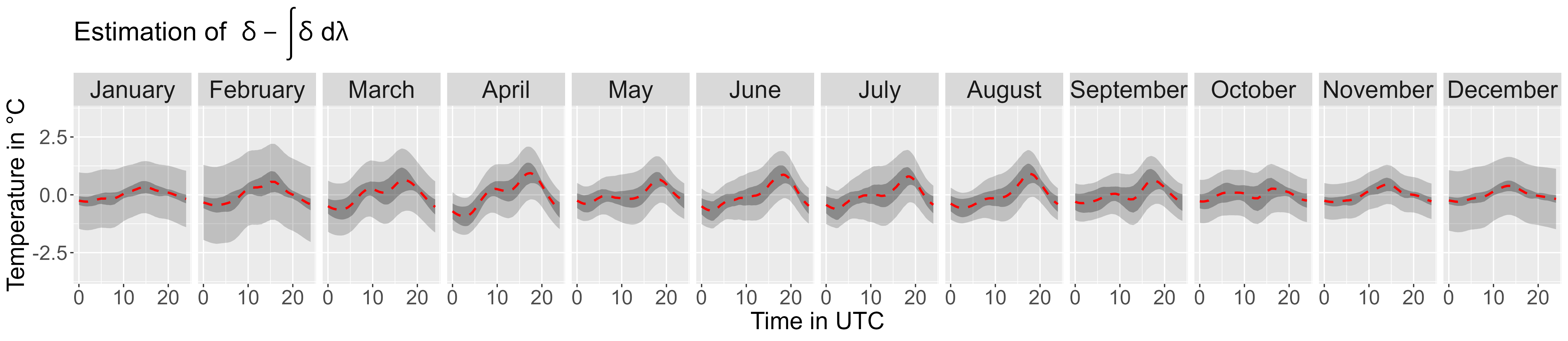}
    \end{subfigure}

\caption{Frankfurt am Main (Germany): Results analogously to Figures \ref{fig: means of Berlin} - \ref{fig:real data constant test} .}
    \label{fig:Frankfurt}
\end{figure}

\begin{table}[h!]
    \centering{
    \scriptsize
    \renewcommand{\arraystretch}{1.3}
    \begin{tabular}{c|c|c|c|c|c|c|c|c|c|c|c|c}
    \multicolumn{13}{c}{Frankfurt (Germany): Time periods $2000$--$2025$ and $1952$--$1972$} \\
    \hline
                         $H_0$              & Jan       & Feb       & Mar       & Apr       & May       & Jun       & Jul      & Aug        & Sep       & Oct       & Nov      & Dec         \\
    \hline
        \rule{0pt}{2.5ex} $\int \delta = 0$ & $0{.}002$ & $0{.}041$ & $0{.}031$ & $0{.}041$ & $0{.}035$ & $0{.}001$ & $0{.}001$& $0{.}001$  & $0{.}022$ & $0{.}041$ & $0{.}001$& $0{.}041$   \\
        \rule{0pt}{2.5ex} $\Delta = 0$      & $0{.}004$ & $0{.}001$ & $0{.}001$ & $0{.}001$ & $0{.}001$ & $0{.}001$ & $0{.}001$& $0{.}001$  & $0{.}001$ & $0{.}081$ & $0{.}001$& $0{.}001$   \\

    \multicolumn{13}{c}{ } \\
    \multicolumn{13}{c}{Hamburg (Germany): Time periods $2000$--$2025$ and $1952$--$1972$} \\
    \hline
                         $H_0$              & Jan       & Feb       & Mar       & Apr       & May       & Jun       & Jul      & Aug        & Sep       & Oct       & Nov      & Dec        \\
    \hline
        \rule{0pt}{2.5ex} $\int \delta = 0$ & $0{.}001$ & $0{.}001$ & $0{.}001$ & $0{.}001$ & $0{.}001$ & $0{.}004$ & $0{.}001$& $0{.}001$  & $0{.}001$ & $0{.}016$ & $0{.}001$& $0{.}01$   \\
        \rule{0pt}{2.5ex} $\Delta = 0$      & $0{.}966$ & $0{.}966$ & $0{.}128$ & $0{.}002$ & $0{.}001$ & $0{.}001$ & $0{.}001$& $0{.}001$  & $0{.}046$ & $0{.}966$ & $0{.}966$& $0{.}273$  \\

    \multicolumn{13}{c}{ } \\
    \multicolumn{13}{c}{Munich (Germany): Time periods $2000$--$2025$ and $1952$--$1972$} \\
    \hline
                         $H_0$              & Jan       & Feb       & Mar       & Apr       & May       & Jun       & Jul      & Aug        & Sep       & Oct       & Nov      & Dec        \\
    \hline
        \rule{0pt}{2.5ex} $\int \delta = 0$ & $0{.}001$ & $0{.}001$ & $0{.}001$ & $0{.}001$ & $0{.}001$ & $0{.}001$ & $0{.}001$& $0{.}001$  & $0{.}001$ & $0{.}001$ & $0{.}001$& $0{.}001$  \\
        \rule{0pt}{2.5ex} $\Delta = 0$      & $0{.}001$ & $0{.}001$ & $0{.}001$ & $0{.}001$ & $0{.}001$ & $0{.}001$ & $0{.}001$& $0{.}001$  & $0{.}001$ & $0{.}001$ & $0{.}001$& $0{.}001$  \\
    \end{tabular}
    }
    \caption{Frankfurt, Hamburg and Munich (Germany): p-values of the global tests of $\int\delta = 0$ and $\Delta = 0$ for each month, Holm-adjusted over the 12 months, separately for each city and each hypothesis.}
    \label{tab: other main global tests}
\end{table}

Table~\ref{tab: other main global tests} gives the p-values for the three cities, Holm-adjusted over the 12 months, separately for each city and each hypothesis; the raw p-values are estimated by the dependent multiplier bootstrap with $N^*=10 \,000$ and are bounded below by $1/N^*$.
At level 5\%, the average difference $\int\delta$ is significantly positive in every month in all three cities. The hypothesis $\Delta = 0$ of a constant difference is rejected in every month in Munich, in every month except October in Frankfurt, and from April to September in Hamburg (September is not significant for bandwidths $0{.}75\,h$; see Section \ref{sec:detailsmehtod}); see also the last row of plots in Figures \ref{fig:Muenchen} - \ref{fig:Frankfurt}. The timing differs from Berlin: in Frankfurt, Hamburg and Munich, the increase is largest in the late afternoon and early evening, around 17:00 to 19:00 UTC, and smallest in the early morning; in Hamburg with its maritime climate, the changes are less pronounced than at the other stations.

In Figure \ref{fig:tropical} we display for each city and each month the average daily change, together with the difference of the maxima of the mean curves as well as the difference of the minima to highlight changes in the extremes of the mean curves. In Berlin as well as in Frankfurt, the difference in maxima is mostly above the average difference, which in turn is above the difference in minima, in Berlin in particular from spring to summer. In Hamburg, the three differences are close to each other in most months, while in Munich the minima increased much more than the maxima in September and October.

\begin{figure}[h!]
    \centering
        \includegraphics[width=\linewidth]{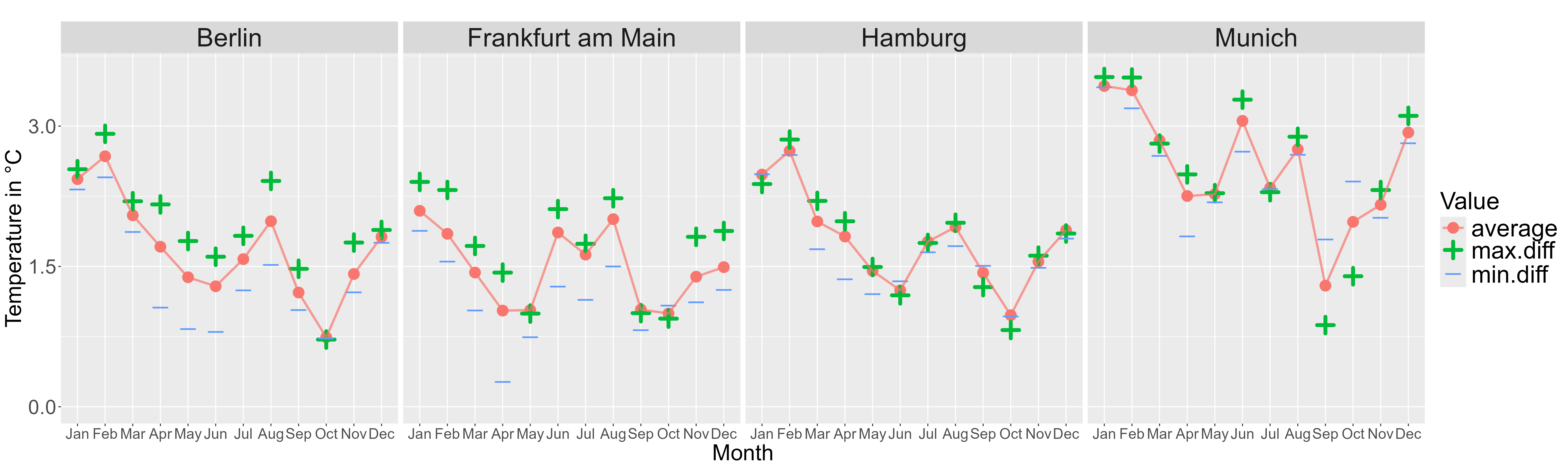}
    \caption{Average difference over the day (solid red) with difference of maxima of mean curves (green +) as well as difference of minima (blue -).}
    \label{fig:tropical}
\end{figure}

\section{Conclusion and discussion}\label{sec:discuss}

In view of global warming, the increase in average daily temperatures that we observe in the temperature time series is to be expected. However, our analysis brings to light  that the daily mean  temperature pattern in most months also changed beyond the positive offset. After Holm adjustment over the months, this holds in every month in Munich, in eleven of the twelve months in Berlin and Frankfurt, and from April to September in Hamburg. The change differs between the stations: in Berlin, daytime temperatures rose more than nighttime temperatures, in particular in spring and early summer, whereas in Frankfurt, Hamburg and Munich the increase was largest in the late afternoon and early evening. In Berlin and Frankfurt, the maxima of the mean curves increased more than their minima in most months; in Hamburg the two increased by similar amounts, and in Munich the minima increased more in September and October. For Frankfurt and Munich, part of the estimated change may be due to the different stations used in the two periods as well as for Frankfurt the relocation of the station in 2014. However, a disproportionately high increase of night temperatures is not consistently observed. An extension of our analysis to weather stations in southern parts of Europe, where tropical nights, with minimum temperatures of at least 20\,$^\circ$C, have become more frequent \citep{correa2024tropical} and are projected to increase further \citep{fischer2010consistent}, would be of some interest. Since tropical nights are defined by the daily minimum temperature rather than by the mean curve, their frequency also depends on changes in day-to-day variability \citep{schar2004role}.

On the methodological side, further parameters apart from the mean functions such as the covariance kernel or long-run covariance kernel, or associated basis functions could be the object of interest in two-sample statistics.
Also, the application of the dependent multiplier bootstrap procedure used in the present paper requires extensive computational resources. A computationally more efficient alternative could be based on the Kac–Rice formula as used e.g.~in \citet{liebl2023}.

\section*{Acknowledgements}

KW and HH gratefully acknowledge financial support from the DFG, grant HO 3260/9-1.

\appendix
\setcounter{equation}{0}\renewcommand{\theequation}{A.\arabic{equation}}\renewcommand{\theHequation}{A.\arabic{equation}}%
\setcounter{theorem}{0}\renewcommand{\thetheorem}{A.\arabic{theorem}}\providecommand{\theHtheorem}{}\renewcommand{\theHtheorem}{A.\arabic{theorem}}%
\section{Smoothness classes, sampling regimes and studentization}\label{app:rates}

\subsection{H\"older classes and H\"older continuous processes}\label{app:holder}
The H\"older class of functions $f\colon T \rightarrow \R$ of smoothness of $\alpha>0$ and
H\"older-constant $M>0$ is defined by
\begin{align*}
	\Hclass = & \big\{f\colon T \to \R \mid \forall \  k=0,\ldots,\lceil \alpha \rceil - 1, \quad |f^{(k)}(t)|\leq M \\
	& \quad\text{and}\quad |f^{(\lceil\alpha\rceil - 1)}(t)-f^{(\lceil\alpha\rceil - 1)}(s)|\leq M|t-s|^{\alpha - \lceil\alpha\rceil + 1}, \quad t,s \in T\big\}.
\end{align*}
Here $\lceil \alpha \rceil - 1 = \max\{ k \in \N_0 \mid k < \alpha\}$ is the largest integer strictly smaller than $\alpha$; for integer $\alpha$ it equals $\alpha-1$, so that, for instance, $\mathcal{H}(1;M)$ consists of the functions bounded by $M$ that are Lipschitz continuous with constant $M$.
For the processes we make H\"older-smoothness assumptions of potentially lower order.
Specifically, we assume that $\expec[Z_i^{\sparse}(0)^2],\expec[Z_k^{\dense}(0)^2]<\infty$ hold for $i=1,\hdots,{n}$ and $k=1,\hdots,{m}$ and there exist random variables $V^{\sparse}_i,V^{\dense}_k>0$ with second moments bounded uniformly in $i$ and $k$ such that for a constant \text{$0<{\beta} \leq 1$}  we have almost surely that
\begin{align}
    |Z^{\sparse}_i(t)-Z^{\sparse}_i(s)|&\leq \, V^{\sparse}_i\,  |t-s|^{{\beta}}  \, ,\qquad
    |Z_k^{\dense}(t)-Z_k^{\dense}(s)|\leq  \, V^{\dense}_k\,  |t-s|^{{\beta}},\qquad t,s\in T.\label{eq:Z}
\end{align}

\subsection{Dense and sparse regimes}\label{app:ratesofconv}
\begin{remark}[Convergence rates for $\bar{\delta}$ in \eqref{eq:diffest}]\label{rem:separate}
For the simple, difference-based estimate $\bar{\delta}_{n,m}$ in \eqref{eq:diffest}, in case $p \lesssim p_d$ and $n  \lesssim m$, for appropriate choices of $h^*$ and $h_d^*$ of the bandwidths we obtain from Theorem 1 in \citet{berger2024dense} that
    \begin{equation}\label{eq:ratesimplediff}
        \sup_{ \mu^{\dense}, \mu^{\sparse} \in \Hclass}\, \expec\big[\big\|\bar{\delta}_{n,m}(\cdotp; h^*, h_d^*) - \delta\big\|_\infty \big] = \, \mathcal O\Big(\max\Big(p^{-\alpha}, n^{-1/2}, \Big(\frac{\log(p\, n)}{p\, n} \Big)^{\alpha/(2\, \alpha + 1)} \Big) \Big),
    \end{equation}
    these rates are minimax-optimal. In particular, in contrast to random, asynchronous designs, for the regular fixed design in \eqref{eq:main.4} for consistency one requires $p \to \infty$ as $n \to \infty$, otherwise there remains a non-negligible discretization error.

    Now if $(\log n)\, n^{1/(2\alpha)}\ll p$ the $n^{-1/2}$-rate dominates and one additionally has process convergence of $\sqrt n\big( \bar{\delta}_{n,m}(\cdotp; h^*, h_d^*) - \delta\big)$ to a Gaussian process with continuous sample paths. This is the setting in which even the more sparsely-observed sample in \eqref{eq:main.4} is sufficiently dense and valid inference based on the simple difference-based estimator $\bar{\delta}$ is already feasible. Let us stress that a dense regime is determined by $p$ being large enough compared to $n$ given the smoothness $\alpha$ of the target function: for larger $\alpha$ and thus smoother target functions, the dense regime also covers smaller $p$ as is made precise in the condition $(\log n)\, n^{1/(2\alpha)}\ll p$.

    We are particularly interested in cases in which at least potentially $(\log n)\, n^{1/(2\alpha)}\gg p$ (sparse regime for estimating $\mu^{\sparse}$), so that the rate in \eqref{eq:ratesimplediff} is slower than $n^{-1/2}$, and even if it remains $n^{-1/2}$ for $(\log n)\, n^{1/(2\alpha)}\eqsim p$, asymptotic normality does not hold. We show in Theorem~\ref{theorem:CLT_delta} that for $\alpha_\delta > \alpha$ the estimator $\hat{\delta}$ in \eqref{eq:estimators} can still be asymptotically normally distributed at rate $\sqrt n$, so that valid asymptotic inference is possible.
\end{remark}

\subsection{Studentized process}\label{app:stud}
\begin{corollary}\label{theorem:CLT_dtandard}\emergencystretch=2em\relax
Under the assumptions of Theorem~\ref{theorem:CLT_delta}, let $\lambda=0$ if $n \llslim m$ and $\lambda$ as in \eqref{eq:sameorder} if $n \simeq m$, and in addition assume that $\var\{Z^{\sparse}(t) - \int_0^1 Z^{\sparse}\} + \lambda\,\var\{Z^{\dense}(t) - \int_0^1 Z^{\dense}\} > 0$ for each $t\in[0{,}1]$. If the covariance kernel estimates $\widehat{\Gamma}^{\sparse}$ and $\widehat{\Gamma}^{\dense}$, see Section~\ref{sec:covest}, are consistent in $C([0{,}1]^2)$, then with $\hat{\sigma}_{\Delta}$ from Section~\ref{section:Two sample problem}
\begin{equation}\label{eq:studentizedlimit}
   \bigg(\frac{\sqrt{n}\,\big(\hat{\Delta}_{n,m}(t;h,h_d) - \Delta(t)\big)}{\hat{\sigma}_{\Delta}(t)}\bigg)_{t\in[0,1]}\overset{D}{\longrightarrow} \mathcal{G}\big(0,\rho_{\rP}\circ(\Gamma^{\sparse}+\lambda\,\Gamma^{\dense})\big),
\end{equation}
with $\Delta$ and $\hat\Delta_{n,m}$ in \eqref{eq:centereddiff}, and where for a covariance kernel $K$
\begin{align*}
    \rho_{\rP}\circ K(t,s)\defeq\frac{(\rP \circ K)(t,s)}{\big((\rP \circ K)(t,t)\,(\rP \circ K)(s,s)\big)^{1/2}}
\end{align*}
denotes the correlation function associated to the covariance kernel $\rP \circ K$, and $\rP$ is the linear projection from $C([0{,}1]^2)$ to $C([0{,}1]^2)$ given by
\begin{equation}\label{eq:projection}
(\rP \circ f)(t,s) \defeq f(t,s) - \int_0^1 \big(f(t,y) + f(y,s)\big)\,\rd y + \int_0^1\!\int_0^1 f(x,y)\,\rd x\,\rd y.
\end{equation}
\end{corollary}
The proof is given in Section~\ref{sec:proofcor}. Under Assumption~\ref{assumption:Independence}, the variance in the condition equals $\sigma^2_{\Delta}(t)$ from Section~\ref{section:Two sample problem}, the limit of $\hat{\sigma}^2_{\Delta}(t)$, which the condition keeps away from zero. Further, a similar result holds for time series under the conditions of Section~\ref{sec:timeseries}, if the estimates of the long-run covariance kernels are consistent and $\sigma^2_{\Delta}(t)>0$ for all $t$.

\section*{Statement on AI use}
Before the first submission, we used ChatGPT (OpenAI; models of the Plus plan, January to May 2026) to discuss formulations for the title, the abstract and the introduction, and to discuss questions about the code, and to draft small
parts of the code, which the authors adapted, integrated into their own
implementation and checked. For the revision, we used Claude (Anthropic; models Opus 5, Opus 5.5 and Fable 5.2) as follows. \textit{Review:} to review drafts of the paper and the supplement, and to review and to suggest improvements to parts of the code. This led to the additional checks of the data preprocessing in Section~\ref{sec:preprocessing} of the supplement, to corrections of statements and notation, and to the detection of an inaccuracy in our original proof of Theorem~\ref{theorem:CLT_delta}. \textit{Proofs:} after reviewing the proof of Theorem~\ref{theorem:CLT_delta}, Claude suggested an alternative in terms of smoothing operators (Sections~\ref{sec:auxresults} and \ref{sec:proofthm}) and drafted a first version, which we checked, revised and adopted; it also pointed out its usefulness for extensions to functional time series (Section~\ref{sec:timeseries}). Further, Claude recommended including proofs for Remark~\ref{rem:estint} (Section~\ref{sec:proofint}), drafted them, and made explicit the conditions under which the remark holds. \textit{Text:} to assist in restructuring the paper, to simplify the notation, and to draft and edit passages of the revision, including Remarks~\ref{rem:alphadelta} and \ref{rem:bandwidth}. All numerical results, figures and tables were produced by the authors. The authors take full responsibility for all content.

\bibliographystyle{abbrvnat}
\bibliography{Literature_fda}

\clearpage

\begin{center}
\phantomsection\pdfbookmark[1]{Supplementary material}{supplement}%
{\Large\bfseries Supplementary material}
\end{center}

\setcounter{section}{0}\renewcommand{\thesection}{S\arabic{section}}\def\theHsection{S\arabic{section}}%
\setcounter{equation}{0}\renewcommand{\theequation}{S.\arabic{equation}}\def\theHequation{S.\arabic{equation}}%
\setcounter{figure}{0}\renewcommand{\thefigure}{S\arabic{figure}}\def\theHfigure{S\arabic{figure}}%
\setcounter{table}{0}\renewcommand{\thetable}{S\arabic{table}}\def\theHtable{S\arabic{table}}%
\setcounter{theorem}{0}\renewcommand{\thetheorem}{S\arabic{theorem}}\def\theHtheorem{S\arabic{theorem}}%
\setcounter{assumption}{0}\renewcommand{\theassumption}{S\arabic{assumption}}\def\theHassumption{S\arabic{assumption}}%
\let\lemma\Slemma\let\endlemma\endSlemma\let\remark\Sremark\let\endremark\endSremark%

\noindent This supplementary material provides rates of convergence in the supremum norm, the proofs of Theorem~\ref{theorem:CLT_delta} and Corollary~\ref{theorem:CLT_dtandard}, technical assumptions, numerical results and further details on the methodology used in the data analysis. Its structure is as follows. Section \ref{ass:genericweights} lists the necessary assumptions on the weights of a linear estimator for Theorem~\ref{theorem:CLT_delta}. Section \ref{sec: proof main theorem} states the result from \citet{berger2024dense} and two lemmas on linear smoothers used in the proofs, the rates of convergence and the results on the average difference, and contains the proofs. Further, Section \ref{section: hv cross validation} includes a $R$-shifted hv-block cross validation in a $K$-fold framework for time series. Section~\ref{Appendix} covers functional time series: the assumptions, the estimation of covariance and long-run covariance kernels and the dependent multipliers of the bootstrap. Finally, Section \ref{sec: add results} concludes with numerical results of the simulation for Section~\ref{section:quantiles} in \ref{sec:additionalsim} and with methodological details for Section~\ref{sec:apllilust} as well as a robustness check of Holm decisions for all four cities in \ref{sec:detailsmehtod}. The preprocessing of the temperature data and checks for trends within the two periods are documented in Section~\ref{sec:preprocessing}.

\section{Assumptions on the weights}\label{ass:genericweights}

\begin{assumption}[Weights of linear estimators]\label{assumption:weights}
    For sufficiently large $p_{0} \in \N$ and therefore sufficiently small $h_{0} > 0$, the following assumptions are valid for all $ {h} \in (c/{p}, h_{0}]$ with ${p} \geq p_{0}$ respectively $ {h_d} \in (c/{p_d}, h_{0}]$ with ${p_d} \geq p_{0}$, where $c > 0$ is a sufficiently large constant.
    In the following, $C_{1},C_{2} > 0$ are constants that do not depend on $n, m, p, p_d, h$ and $h_d$.
    \begin{enumerate}[label=(W\arabic*)]
        \item\label{item:W1} For $t \in T$, the weights reproduce polynomials up to the degree $\lceil\alpha_\delta\rceil - 1$, respectively $\lceil\alpha\rceil - 1$,
        \begin{align*}
            \sum_{j=1}^{{p}} w^{\sparse}_{j}(t;{h}) &= 1 \quad \text{ and } \quad \sum_{j=1}^{{p}} w^{\sparse}_{j}(t;{h})\,(t^{\sparse}_{j} - t)^{r} = 0 \quad \text{ for } r=1,\hdots,\lceil\alpha_\delta\rceil - 1.
        \end{align*}
    Analogously, the above applies to $w^{\dense}_{l}(t;{h_d})$ for $r=1,\hdots,\lceil\alpha\rceil - 1$.
    \item\label{item:W2}  For $t \in T$, the $j$-th weight disappears outside a small range around $t^{\sparse}_{j}$, respectively the $l$-th weight disappears outside a small range around $t^{\dense}_l$. For $j=1,\hdots,{p}$ and  $l=1,\hdots,{p_d}$ this means
    \begin{align*}
        w^{\sparse}_{j}(t;{h}) = 0 \quad \text{ if }  |t^{\sparse}_{j} - t| > {h} \quad \text{and} \quad w^{\dense}_{l}(t;{h_d}) = 0 \quad \text{ if }  |t^{\dense}_l - t| > {h_d}.
    \end{align*}
    \item\label{item:W3}For $t \in T$, the absolute value of the weights is limited by a constant $C_{1}>0$: $$\underset{1\leq j \leq p}{\sup} |w^{\sparse}_{j}(t;{h})| \leq \frac{C_{1}}{{p}{h}} \quad \text{and} \quad \underset{1\leq l \leq {p_d}}{\sup} |w^{\dense}_{l}(t;{h_d})| \leq \frac{C_1}{{p_d}{h_d}}.$$
    \item\label{item:W4}  For a Lipschitz constant $C_{2}>0$ and for  $t,s \in T$ it holds that
    \begin{align*}
        |w^{\sparse}_{j}(t;{h})-w^{\sparse}_{j}(s;{h})| &\leq \frac{C_{2}}{{p}{h}}\min\bigg(\frac{|t-s|}{{h}},1\bigg) %
    \end{align*}
    Analogously, the above applies to $w^{\dense}_{l}(t;{h_d})$ with the same constant $C_2$.
    \end{enumerate}
\end{assumption}

\begin{assumption}[Design Assumption]\label{assumption: design}
    For $t \in T$ and $C_{3}>0,$ the numbers of non-zero weights are bounded
    \begin{align*}
        \operatorname{card} \, \{j \,|\, w^{\sparse}_{j}(t;{h}) \neq 0 \text{ for }j=1,\hdots,{p}\} \leq C_{3}\, {p}{h} \quad \text{and}\\
        \operatorname{card} \,\{l \,|\, w^{\dense}_{l}(t;{h_d}) \neq 0 \text{ for }l=1,\hdots,{p_d}\}\leq C_3\, {p_d}{h_d}.
    \end{align*}
\end{assumption}

We notice that property \ref{item:W3} together with Assumption \ref{assumption: design} implies \ref{item:W5}:
\begin{enumerate}[label=(W5)]
\item\label{item:W5}  For $t \in[0{,}1]$, the sum of absolute values of the weights is limited by a constant $C_{4}>0$: $$\sum_{j= 1}^p| w^{\sparse}_j(t;{h})| \leq C_{4} \quad\text{and}\quad \sum_{l= 1}^{p_d}| w^{\dense}_l(t;{h_d})| \leq C_{4}.$$
\end{enumerate}

\section{Proofs}\label{sec: proof main theorem}

\subsection{Auxiliary results used in the proofs}\label{sec:auxresults}
To keep the paper self-contained, we state the result from \citet{berger2024dense} on which our proofs rely. The argument for asymptotic normality is simplified by combining approximation properties of the smoothing operators associated with the weights from Assumption~\ref{assumption:weights} with the asymptotic equicontinuity of weakly convergent sequences. To review the error bounds from \citealp[Lemma~2 and Theorem~1]{berger2024dense} let us first simplify the setting to a single sample
\begin{equation}\label{eq:onesample}
Y_{i,j} = \mu(t_j) + Z_i(t_j) + \epsilon_{i,j}, \qquad i=1,\ldots,n,\quad j=1,\ldots,p,
\end{equation}
with design points $t_1<\cdots<t_p$ in $T=[0,1]$. Let $\bar Y_j = n^{-1}\sum_{i=1}^n Y_{i,j}$. For the linear estimator $\hat\mu_n(t;h) = \sum_{j=1}^p w_j(t;h)\,\bar Y_j$ we have the error decomposition
\begin{align}
\hat\mu_n(t;h) - \mu(t) & = \sum_{j=1}^p w_j(t;h)\big(\mu(t_j)-\mu(t)\big) + \sum_{j=1}^p w_j(t;h)\,\bar\epsilon_{j,n} + \sum_{j=1}^p w_j(t;h)\,\bar Z_n(t_j) \nonumber\\
& \defeql I_1(t) + I_2(t) + I_3(t), \label{eq:onesampledecomp}
\end{align}
where $\bar\epsilon_{j,n} = n^{-1}\sum_{i=1}^n\epsilon_{i,j}$ and $\bar Z_n = n^{-1}\sum_{i=1}^n Z_i$. The $Z_i$ are i.i.d.~with the distribution of $Z$, where $Z$ is centered and satisfies $\expec[Z(0)^2]<\infty$ and $|Z(t)-Z(s)|\leq V|t-s|^\beta$ for some $\beta\in(0{,}1]$ and a random variable $V$ with $\expec[V^2]<\infty$. The errors $\epsilon_{i,j}$ are assumed to be independent, independent of the $Z_i$, and sub-Gaussian with uniformly bounded sub-Gaussian norms, and the weights are assumed to satisfy Assumptions~\ref{assumption:weights} and \ref{assumption: design} for a single sample, with degree $\lceil\alpha\rceil-1$ in \ref{item:W1}.

\begin{lemma}[{\citealp[Lemma~2 and Theorem~1]{berger2024dense}}]\label{lem:rates}
There is a constant $C>0$ such that for all $h\in(c/p,h_0]$ and $p\geq p_0$,
\[ \sup_{\mu\in\Hclass}\|I_1\|_\infty \leq C\,h^\alpha, \quad\ \expec\big[\|I_2\|_\infty\big] \leq C\,\Big(\frac{\log(1/h)}{n\,p\,h}\Big)^{1/2}, \quad\ \expec\big[\|I_3\|_\infty\big] \leq C\,n^{-1/2}. \]
\end{lemma}

The bounds on $I_1$ and $I_2$ in Lemma~\ref{lem:rates} do not involve the processes $Z_i$, and therefore remain valid for dependent processes.

Now we turn back to the dense-sparse two-sample situation. To simplify and improve the analysis of the CLT, for a function $f\colon[0{,}1]\to\R$ we consider the linear smoothers
\begin{equation}\label{eq:smoothop}
\big(\mathcal{S}^{\sparse}_{h} f\big)(t) = \sum_{j=1}^{p} w^{\sparse}_{j}(t;h)\,f(t^{\sparse}_{j}), \qquad \big(\mathcal{S}^{\dense}_{h_d} f\big)(t) = \sum_{l=1}^{p_d} w^{\dense}_{l}(t;h_d)\,f(t^{\dense}_{l}), \qquad t\in[0{,}1].
\end{equation}
Let $\omega(f,h) = \sup\{|f(t)-f(s)| : t,s\in[0{,}1],\ |t-s|\leq h\}$, $h>0$, denote the modulus of continuity of $f$. The bounds in Lemma~\ref{lem:smoothop} are elementary; the second is the bias bound of Lemma~\ref{lem:rates} for a function that is merely continuous.

\begin{lemma}[Smoothing operators and weak convergence]\label{lem:smoothop}
Impose the Assumptions~\ref{assumption:weights} and \ref{assumption: design}, and consider the smoothing operators $\mathcal{S}^{\sparse}_{h}$ and $\mathcal{S}^{\dense}_{h_d}$. Suppose that $p,p_d\geq p_0$, and that $h\in(c/p,h_0]$ and $h_d\in(c/p_d,h_0]$.
\begin{enumerate}
\item Then for all bounded $f\colon[0{,}1]\to\R$, the functions $\mathcal{S}^{\sparse}_{h}f$ and $\mathcal{S}^{\dense}_{h_d}f$ are continuous, and
\begin{align*}
\|\mathcal{S}^{\sparse}_{h} f\|_\infty &\leq C_4\,\|f\|_\infty, \qquad \|\mathcal{S}^{\sparse}_{h} f - f\|_\infty \leq C_4\,\omega(f,h), \\
\|\mathcal{S}^{\sparse}_{h}\mathcal{S}^{\dense}_{h_d} f - f\|_\infty &\leq C_4^2\,\omega(f,h_d) + C_4\,\omega(f,h).
\end{align*}
The first two bounds also hold for $\mathcal{S}^{\dense}_{h_d}$, with $h_d$ in place of $h$.
\item Let $G_n$, $n\in\N$, and $G$ be random elements of $C[0{,}1]$ with $G_n \overset{D}{\longrightarrow} G$ in $C[0{,}1]$. Then, if $h\to0$ and $h_d\to0$, we have that
\[ \|\mathcal{S}^{\sparse}_{h} G_n - G_n\|_\infty \overset{\prob}{\longrightarrow} 0 \qquad\text{and}\qquad \|\mathcal{S}^{\sparse}_{h}\mathcal{S}^{\dense}_{h_d} G_n - G_n\|_\infty \overset{\prob}{\longrightarrow} 0, \]
so that $\mathcal{S}^{\sparse}_{h} G_n$ and $\mathcal{S}^{\sparse}_{h}\mathcal{S}^{\dense}_{h_d} G_n$ converge in distribution to $G$ in $C[0{,}1]$ as well.
\end{enumerate}
\end{lemma}
\begin{proof}
\textbf{Ad i)} Continuity follows from \ref{item:W4}, and the first bound from \ref{item:W5}. By \ref{item:W1} the weights sum to one, and by \ref{item:W2} only design points with $|t^{\sparse}_j-t|\leq h$ are involved in the sum. Therefore by \ref{item:W5},
\[ \big|\big(\mathcal{S}^{\sparse}_{h} f\big)(t) - f(t)\big| = \Big|\sum_{j=1}^{p} w^{\sparse}_{j}(t;h)\,\big(f(t^{\sparse}_{j})-f(t)\big)\Big| \leq C_4\,\omega(f,h), \qquad t\in[0{,}1], \]
and similarly for $\mathcal{S}^{\dense}_{h_d}$. The third bound follows upon writing
$$\mathcal{S}^{\sparse}_{h}\mathcal{S}^{\dense}_{h_d} f - f = \mathcal{S}^{\sparse}_{h}\big(\mathcal{S}^{\dense}_{h_d} f - f\big) + \big(\mathcal{S}^{\sparse}_{h} f - f\big)$$
together with the first two bounds.

\smallskip

\textbf{Ad ii)} We show that $\omega(G_n,h)\to0$ and $\omega(G_n,h_d)\to0$ in probability. The convergence in distribution then follows from Slutsky's lemma upon using i). Note that $\omega(G_n,h)\to0$ and $\omega(G_n,h_d)\to0$ in probability amount to the asymptotic equicontinuity of a weakly convergent sequence, which is the necessity part of the CLT of \citet{andersen1987}, see also \citet[Theorem 1.5.7 and Addendum 1.5.8]{vaart1996}. For convenience we give a simple direct proof: fix $\epsilon>0$ and $r>0$. The map $f\mapsto\omega(f,r)$ is continuous on $C[0{,}1]$, since $|\omega(f,r)-\omega(g,r)|\leq 2\,\|f-g\|_\infty$. Since $\omega(f,\cdot)$ is nondecreasing and $h\leq r$ for large $n$, the continuous mapping theorem and the portmanteau theorem give
\[ \limsup_{n\to\infty}\,\prob\big(\omega(G_n,h)\geq\epsilon\big) \leq \limsup_{n\to\infty}\,\prob\big(\omega(G_n,r)\geq\epsilon\big) \leq \prob\big(\omega(G,r)\geq\epsilon\big). \]
The paths of $G$ are uniformly continuous on the compact interval $[0,1]$, so that $\omega(G,r)\to0$ almost surely as $r\to0$, and the right-hand side tends to zero. The same argument applies to $h_d$.

\end{proof}

Lemma~\ref{lem:smoothop}, ii) reduces a limit theorem for the smoothed, discretely observed processes to the limit theorem for the processes themselves. Following the line of arguments in \citet{degras}, in \citet[Theorem 3]{berger2024dense} the asymptotic normality of a linear smoother was instead derived from a functional central limit theorem for the triangular array of smoothed processes \citep[Theorem 10.6]{pollard}, for which manageability of this array and the convergence of its covariance kernels had to be verified. %

\subsection{Rates of convergence and proof of Theorem~\ref{theorem:CLT_delta}}\label{sec:proofthm}
\begin{theorem}[Rate of convergence]\label{thm:rates}\emergencystretch=2em\relax
 In model (\ref{eq:main.4}) consider the residual-based estimator $\hat{\delta}_{n,m}$ in (\ref{eq:estimators}) under the smoothness assumption (\ref{eq:smoothnessass}), Assumption~\ref{assumption:weights} on the weights with degrees $\lceil\alpha_{\delta}\rceil - 1$ and $\lceil \alpha\rceil - 1$, the H\"older conditions (\ref{eq:Z}), Assumption~\ref{assumption:Independence} and Assumption~\ref{assumption: design}, with $p \lesssim p_d$ and $n \lesssim m$. Then for $\alpha_\delta \geq \alpha$,
\begin{align}
    & \sup_{(h,h_d)}\, \sup_{ \mu^{\dense} \in \Hclass}\,\sup_{\delta \in \mathcal{H}(\alpha_\delta; M_\delta)} a_{n,m,h,h_d}^{-1}\, \expec\big[\big\|\hat{\delta}_{n,m}(\cdotp; h, h_d) - \delta\big\|_\infty \big] = \, \mathcal O\big(1\big),\label{eq:convrateh1}\\
    & a_{n,m,h,h_d} = \max\Big(n^{-1/2}, h^{\alpha_\delta},  \Big(\frac{\log(1/h)}{p\, n\, h} \Big)^{1/2}, h_d^{\alpha}, \Big(\frac{\log(1/ h_d)}{p_d\, m\, h_d} \Big)^{1/2} \Big),\nonumber
\end{align}
where the first $\sup$ is over $ (h,h_d) \in (c/p,h_0]\times (c/p_d,h_0]$ with $c$ and $h_0$ from Assumption~\ref{assumption:weights}. The suprema over $\mu^{\dense}$ and $\delta$ run over the parameter space $\mu^{\dense} \in \Hclass$, $\delta \in \mathcal{H}(\alpha_\delta; M_\delta)$. It contains all pairs satisfying (\ref{eq:smoothnessass}) and implies $\mu^{\sparse} = \mu^{\dense} + \delta \in \mathcal{H}(\alpha; M + M_\delta)$, since $\mathcal{H}(\alpha_\delta; M_\delta) \subset \mathcal{H}(\alpha; M_\delta)$ on $T$ for $\alpha_\delta \geq \alpha$. Therefore, choosing $h^* \sim \max\{c/p,\, (\log(p n)/(p n))^{1/(2 \alpha_\delta + 1)}\}$ and similarly for $h_d^*$ we obtain
\begin{align}
    &  \sup_{ \mu^{\dense} \in \Hclass}\,\sup_{\delta \in \mathcal{H}(\alpha_\delta; M_\delta)} \, \expec\big[\big\|\hat{\delta}_{n,m}(\cdotp; h^*, h_d^*) - \delta\big\|_\infty \big] = \, \mathcal O\big(a_{n,m, p, p_d}\big),\label{eq:convrateh}\\
    &  a_{n,m,p,p_d} = \max\Big(n^{-1/2}, p^{-\alpha_\delta},  \Big(\frac{\log(p\, n)}{p\, n} \Big)^{\alpha_\delta/(2\, \alpha_\delta + 1)}, p_d^{-\alpha}, \Big(\frac{\log(p_d\, m)}{p_d\, m} \Big)^{\alpha/(2\, \alpha + 1)} \Big).\nonumber
\end{align}
\end{theorem}
\begin{proof}[Proof of Theorem~\ref{theorem:CLT_delta} and Theorem~\ref{thm:rates}]\emergencystretch=2em\relax

By construction of the estimator, we have that
\begin{align}
    \hat{\delta}_{n,m}(t;h) - \delta(t) =&  \sum_{j=1}^{{p}} w^{\sparse}_{j}(t;{h})\,\big(\delta(t^{\sparse}_{j})-\delta(t) \big)+  \sum_{j=1}^{{p}} w^{\sparse}_{j}(t;{h})\,\big( \bar{\epsilon}^{\sparse}_{j,n} + \bar Z^{\sparse}_n(t^{\sparse}_{j}) \big) \label{eq:erroronesparse} \\
    & - \sum_{j=1}^{{p}} w^{\sparse}_{j}(t;{h})\,\big(\hat \mu^{\dense}_{m}(t^{\sparse}_{j}; h_d)-\mu^{\dense}(t^{\sparse}_{j})\big)\label{eq:errortwodense}
\end{align}
where
\begin{align*}
    \bar{\epsilon}^{\sparse}_{j,n}\defeq  \frac{1}{{n}}\sum_{i=1}^{{n}}\epsilon^{\sparse}_{i,j},\qquad \bar{Z}^{\sparse}_{n}\defeq  \frac{1}{{n}}\sum_{i=1}^{{n}}\, Z^{\sparse}_{i}.
\end{align*}

\bigskip

\textit{Proof of Theorem~\ref{thm:rates}.}\qquad Now the first part in \eqref{eq:erroronesparse} is the estimation error in the sparse observational model (\ref{eq:main.4}), but with $\mu^{\sparse}$ replaced by $\delta$. By Lemma~\ref{lem:rates}, applied with $\alpha_\delta$ and $M_\delta$ in place of $\alpha$ and $M$, the expected sup-norm of this term is of order $\max\big(h^{\alpha_\delta}, (\log(1/h)/(n\, p\, h))^{1/2}, n^{-1/2}\big)$. For \eqref{eq:errortwodense}, using \ref{item:W5} and Lemma~\ref{lem:rates}, now in the dense observational model in (\ref{eq:main.4}) we obtain
\begin{align*}
 & \sup_{t \in [0,1]} \Big| \sum_{j=1}^{{p}} w^{\sparse}_{j}(t;{h})\,\big(\hat \mu^{\dense}_{m}(t^{\sparse}_{j}; h_d)-\mu^{\dense}(t^{\sparse}_{j})\big)\Big| \\
 & \qquad \leq \big\|\hat \mu^{\dense}_{m}(\cdot ; h_d)-\mu^{\dense}\big\|_\infty\,  \sup_{t \in [0,1]} \sum_{j=1}^{{p}} \big|  w^{\sparse}_{j}(t;{h})\big| \leq C_4\, \big\|\hat \mu^{\dense}_{m}(\cdot ; h_d)-\mu^{\dense}\big\|_\infty,
\end{align*}
and by Lemma~\ref{lem:rates} the expectation of the right-hand side is of order
\begin{align*}
 \max\big\{h_d^{\alpha},\, (\log(1/h_d)/(m\, p_d\, h_d))^{1/2},\, m^{-1/2}\big\}.
\end{align*}
 \eqref{eq:convrateh1} follows since by assumption $n \lesssim m$, and the rate in \eqref{eq:convrateh} is obtained by plugging in the choices of the bandwidths $h$ and $h_d$.

\bigskip

\textit{Proof of Theorem~\ref{theorem:CLT_delta}.}\qquad Let $G^{\sparse}_n = \sqrt{n}\,\bar Z^{\sparse}_n$ and $G^{\dense}_m = \sqrt{m}\,\bar Z^{\dense}_m$, where $\bar Z^{\sparse}_n = n^{-1}\sum_{i=1}^n Z^{\sparse}_i$ and $\bar Z^{\dense}_m = m^{-1}\sum_{k=1}^m Z^{\dense}_k$. The proof uses the processes only through the central limit theorems
\begin{equation}\label{eq:cltprocesses}
G^{\sparse}_n \overset{D}{\longrightarrow} \mathcal{G}(0,\Gamma^{\sparse}) \qquad\text{and}\qquad G^{\dense}_m \overset{D}{\longrightarrow} \mathcal{G}(0,\Gamma^{\dense}) \qquad\text{in } C[0{,}1].
\end{equation}
Under Assumption~\ref{assumption:Independence} and the H\"older condition (\ref{eq:Z}), these hold with $\Gamma^{\sparse}=\gamma^{\sparse}_0$ and $\Gamma^{\dense}=\gamma^{\dense}_0$ by the Jain--Marcus central limit theorem, see \citet[Examples 2.11.13 and 2.11.14]{vaart1996}.

We use the smoothing operators in \eqref{eq:smoothop}. As in \eqref{eq:onesampledecomp}, let $I^{\sparse}_1$ and $I^{\sparse}_2$ denote the bias and the error term of the first part in \eqref{eq:erroronesparse}, that is, of the sparse sample with $\delta$ in place of $\mu^{\sparse}$, and let $I^{\dense}_1$ and $I^{\dense}_2$ denote those of $\hat\mu^{\dense}_m(\cdot\,;h_d)$, so that $\hat\mu^{\dense}_m(\cdot\,;h_d)-\mu^{\dense} = I^{\dense}_1+I^{\dense}_2+\mathcal{S}^{\dense}_{h_d}\bar Z^{\dense}_m$. Then \eqref{eq:erroronesparse} and \eqref{eq:errortwodense} read
\begin{align}
\sqrt{n}\,\big(\hat\delta_{n,m}(\cdot\,;h,h_d)-\delta\big) = {} & \mathcal{S}^{\sparse}_{h} G^{\sparse}_n - \sqrt{n/m}\;\mathcal{S}^{\sparse}_{h}\mathcal{S}^{\dense}_{h_d} G^{\dense}_m \label{eq:cltleading}\\
& + \sqrt{n}\,\big(I^{\sparse}_1 + I^{\sparse}_2\big) - \sqrt{n}\;\mathcal{S}^{\sparse}_{h}\big(I^{\dense}_1 + I^{\dense}_2\big). \label{eq:cltremainder}
\end{align}
For large $n$ we have $H\subset(c/p,h_0]$ and $H_d\subset(c/p_d,h_0]$, and $h,h_d\to0$.

\smallskip

\textit{Bounding \eqref{eq:cltremainder}.} By Lemma~\ref{lem:rates}, applied with $\alpha_\delta$ and $M_\delta$ in place of $\alpha$ and $M$, and since $h\in H$,
\begin{align*}
\sqrt{n}\,\|I^{\sparse}_1\|_\infty &\leq C\,\sqrt{n}\,h^{\alpha_\delta} \leq C\,c_2^{\alpha_\delta}\,(\log p)^{-\eta\,\alpha_\delta} \to 0, \\
\sqrt{n}\;\expec\big[\|I^{\sparse}_2\|_\infty\big] &\leq C\,\Big(\frac{\log(1/h)}{p\,h}\Big)^{1/2} \to 0,
\end{align*}
the latter since $p\,h\geq c_1(\log p)^{1+\eta}$ and hence, for large $n$, $h\geq 1/p$ and $\log(1/h)\leq\log p$. Similarly, by the first bound in Lemma~\ref{lem:smoothop}, i), and by Lemma~\ref{lem:rates} for the dense sample, since $n\lesssim m$ and $h_d\in H_d$, so that $p_d\,h_d\geq c_1(\log p_d)^{1+\eta}$ and $\log(1/h_d)\leq\log p_d$ for large $n$,
\begin{align*}
\sqrt{n}\,\big\|\mathcal{S}^{\sparse}_{h} I^{\dense}_1\big\|_\infty &\leq C\,\sqrt{n}\,h_d^{\alpha} \leq C\,c_2^{\alpha}\,(\log p_d)^{-\eta\,\alpha} \to 0, \\
\sqrt{n}\;\expec\big[\big\|\mathcal{S}^{\sparse}_{h} I^{\dense}_2\big\|_\infty\big] &\leq C\,\Big(\frac{n}{m}\,\frac{\log(1/h_d)}{p_d\,h_d}\Big)^{1/2} \to 0.
\end{align*}
Hence the supremum norm of the terms in \eqref{eq:cltremainder} tends to zero in probability.

\smallskip

\textit{Concluding for \eqref{eq:cltleading}.} By \eqref{eq:cltprocesses} and Lemma~\ref{lem:smoothop}, ii), in $C[0{,}1]$,
\[ \mathcal{S}^{\sparse}_{h} G^{\sparse}_n \overset{D}{\longrightarrow} \mathcal{G}(0,\Gamma^{\sparse}) \qquad\text{and}\qquad \mathcal{S}^{\sparse}_{h}\mathcal{S}^{\dense}_{h_d} G^{\dense}_m \overset{D}{\longrightarrow} \mathcal{G}(0,\Gamma^{\dense}), \]
and the two sequences are independent, since the two samples are. If $n\llslim m$, then $\sqrt{n/m}\to0$, the second term in \eqref{eq:cltleading} tends to zero in probability, and (\ref{eq:distinctorder}) follows from Slutsky's lemma. If $n/m\to \lambda$, the difference in \eqref{eq:cltleading} converges in distribution to $\mathcal{G}(0,\Gamma^{\sparse}+\lambda\,\Gamma^{\dense})$, which gives (\ref{eq:sameorder}).
\end{proof}

\subsection{Proof of Corollary~\ref{theorem:CLT_dtandard}}\label{sec:proofcor}
This follows from Theorem~\ref{theorem:CLT_delta} by first applying the continuous mapping theorem for the map $f \mapsto f - \int f$, $f \in C[0,1]$, and then applying a standard functional version of Slutsky's theorem in $C([0,1])$ (e.g.~from Slutsky’s lemma as in  \citet{vaart1996}, Example 1.4.7, together with the continuous mapping theorem (\citet{vaart1996}, Theorem 1.11.1) applied to $(f,g) \mapsto f/g^{1/2}$ in a neighborhood of a continuous, positive function $g>0$.)

\subsection{The average difference}\label{sec:proofint}
\begin{remark}[Estimating $\int \delta$]\label{rem:estint}
    In our application we are also interested in the average value of the difference between the mean functions in (\ref{eq:main.4}), that is, the integral $\int_0^1 \delta$. Here we do not have the additional error from the supremum of the observational errors: Under the assumptions of Theorem~\ref{thm:rates}, setting $h' \sim \max\Big(c/p, \big(\frac{1}{p\, n} \big)^{1/(2\, \alpha_\delta + 1)} \Big) $ and similarly for $h_d'$ we have that
    \begin{align}
        &  \sup_{ \mu^{\dense} \in \Hclass}\,\sup_{\delta \in \mathcal{H}(\alpha_\delta; M_\delta)} \, \expec\big[\big|\int_0^1 \hat{\delta}_{n,m}(t; h', h_d' ) \, \rd t - \int_0^1\, \delta(t)\, \, \rd t\big| \big] = \, \mathcal O\big(\bar a_{n,m, p, p_d}\big),\label{eq:convratehint}\\
        &  \bar a_{n,m,p,p_d} = \max\big(n^{-1/2}, p^{-\alpha_\delta},   p_d^{-\alpha} \big).\nonumber
    \end{align}
    Moreover, if $ n/m \to \lambda$, $\lambda \geq 0$, and the bandwidths $h \in (c/p,h_0]$ and $h_d \in (c/p_d,h_0]$ satisfy $\sqrt{n}\,h^{\alpha_\delta} \to 0$, $\sqrt{n}\,h_d^{\alpha} \to 0$ and $p_d\,h_d \to \infty$, which is possible if the $n^{-1/2}$-rate dominates in \eqref{eq:convratehint}, that is, if $n^{1/(2\alpha_\delta)} \ll p$ and $n^{1/(2\alpha)} \ll p_d$, then
    \begin{equation}\label{eq:sameorderint1}
        \sqrt{n}\Big(\int_0^1 \hat{\delta}_{n,m}(t; h, h_d) \, \rd t - \int_0^1\, \delta(t)\, \, \rd t\Big) \overset{D}{\longrightarrow} \mathcal N\Big(0,\int_{[0,1]^2} \, \big(\Gamma^{\sparse}+\lambda\,\Gamma^{\dense}\big)(s,t) \rd s\, \rd t\Big).
    \end{equation}
\end{remark}

\textit{Proof of \eqref{eq:convratehint}.}\qquad We use the notation of the proof of Theorem~\ref{theorem:CLT_delta} in Section~\ref{sec:proofthm}. Let
$$a_j = \int_0^1 w^{\sparse}_j(t;h)\,\rd t, \qquad j=1,\ldots,p.$$
By \ref{item:W2} and \ref{item:W3},
\begin{equation}\label{eq:intweights}
|a_j| \leq \frac{2\,C_1}{p} \qquad\text{so that}\qquad \sum_{j=1}^p |a_j| \leq 2\,C_1.
\end{equation}
We integrate the decomposition \eqref{eq:erroronesparse}, \eqref{eq:errortwodense}, or equivalently (after scaling with $\sqrt n$) \eqref{eq:cltleading} and \eqref{eq:cltremainder}. The integrals of the bias terms and of the process terms are bounded by their supremum norms, so that by Lemmas~\ref{lem:rates} and \ref{lem:smoothop}, i), their expected absolute values are of order $h^{\alpha_\delta} + n^{-1/2} + h_d^{\alpha} + m^{-1/2}$. For the error term of the sparse sample, $\int_0^1 I^{\sparse}_2(t)\,\rd t = \sum_{j=1}^p a_j\,\bar\epsilon^{\sparse}_{j,n}$ is centered, and since the errors are independent with variances bounded by a multiple of $\sigma^2$, \eqref{eq:intweights} gives
\[ \var\Big(\sum_{j=1}^p a_j\,\bar\epsilon^{\sparse}_{j,n}\Big) \leq \frac{C\,\sigma^2}{n}\,\sum_{j=1}^p a_j^2 \leq \frac{C\,\sigma^2}{n}\,\max_{1\leq j\leq p}|a_j|\,\sum_{j=1}^p |a_j| \leq \frac{C}{n\,p}. \]
In contrast to the supremum norm, after integrating neither a logarithmic factor nor the bandwidth enters. Similarly, with $\bar\epsilon^{\dense}_{l,m} = m^{-1}\sum_{k=1}^m\epsilon^{\dense}_{k,l}$,
\[ \int_0^1 \big(\mathcal{S}^{\sparse}_{h} I^{\dense}_2\big)(t)\,\rd t = \sum_{l=1}^{p_d} b_l\,\bar\epsilon^{\dense}_{l,m}, \qquad b_l = \sum_{j=1}^p a_j\,w^{\dense}_l(t^{\sparse}_j;h_d), \]
where $|b_l| \leq 2\,C_1^2/(p_d\,h_d)$ by \ref{item:W3} and \eqref{eq:intweights}, and $\sum_{l=1}^{p_d}|b_l| \leq 2\, C_1\, C_4$ by \ref{item:W5} and \eqref{eq:intweights} (second part), so that the variance is at most $C/(m\,p_d\,h_d)$. Altogether, uniformly over the parameter space of Theorem~\ref{thm:rates} and over $(h,h_d)\in(c/p,h_0]\times(c/p_d,h_0]$,
\[ \expec\Big[\Big|\int_0^1 \big(\hat\delta_{n,m}(t;h,h_d) - \delta(t)\big)\,\rd t\Big|\Big] \leq C\,\Big(n^{-1/2} + h^{\alpha_\delta} + h_d^{\alpha} + (m\,p_d\,h_d)^{-1/2}\Big), \]
where we used $n\lesssim m$ to remove the $m^{-1/2}$-term. For $h'$ as in Remark~\ref{rem:estint} we have $(h')^{\alpha_\delta} \lesssim \max\big(p^{-\alpha_\delta}, n^{-1/2}\big)$, since $(p\,n)^{-\alpha_\delta/(2\alpha_\delta+1)}$ is at most $n^{-1/2}$ if $p^{2\alpha_\delta}\geq n$, and at most $p^{-\alpha_\delta}$ otherwise. In the same way, $(h_d')^{\alpha}$ and $(m\,p_d\,h_d')^{-1/2} \lesssim (p_d\,m)^{-\alpha/(2\alpha+1)}$ are of order $\max\big(p_d^{-\alpha}, m^{-1/2}\big)$, and \eqref{eq:convratehint} follows.

\smallskip

\textit{Proof of \eqref{eq:sameorderint1}.}\qquad Under the conditions on the bandwidths, $\sqrt{n}\,|\int_0^1 I^{\sparse}_1| \leq C\,\sqrt{n}\,h^{\alpha_\delta}\to0$ and $\sqrt{n}\,|\int_0^1 \mathcal{S}^{\sparse}_{h} I^{\dense}_1| \leq C\,\sqrt{n}\,h_d^{\alpha}\to0$, and by the variance bounds above the integrated error terms, multiplied by $\sqrt{n}$, are of order $\mathcal{O}_{\prob}(p^{-1/2})$ and $\mathcal{O}_{\prob}\big((n/m)^{1/2}\,(p_d\,h_d)^{-1/2}\big)$. Both tend to zero, since $p\to\infty$, $n\lesssim m$ and $p_d\,h_d\to\infty$. Further, $h,h_d\to0$, so that by \eqref{eq:cltprocesses}, Lemma~\ref{lem:smoothop}, ii) and the continuous mapping theorem
\[ \int_0^1 \big(\mathcal{S}^{\sparse}_{h} G^{\sparse}_n\big)(t)\,\rd t = \int_0^1 G^{\sparse}_n(t)\,\rd t + o_{\prob}(1) \overset{D}{\longrightarrow} \mathcal{N}\Big(0,\int_{[0,1]^2}\Gamma^{\sparse}\Big), \]
and in the same way $\int_0^1 \big(\mathcal{S}^{\sparse}_{h}\mathcal{S}^{\dense}_{h_d} G^{\dense}_m\big)(t)\,\rd t$ is asymptotically normal with variance $\int_{[0,1]^2}\Gamma^{\dense}$. Since the two samples are independent and $n/m\to \lambda$, \eqref{eq:sameorderint1} follows. For $h\in H$ and $h_d\in H_d$ as in Theorem~\ref{theorem:CLT_delta} the conditions on the bandwidths are satisfied, and then \eqref{eq:sameorderint1} also follows directly from Theorem~\ref{theorem:CLT_delta} and the continuous mapping theorem.

\section{hv cross validation}\label{section: hv cross validation}

We consider hv-block cross validation in a $K$-fold framework, repeated over $R$ cyclically shifted partitions, and proceed as follows: Observing curves of a time series of length $n$ we split them into test and train data according to a specific procedure, since \citet{Roberts2017} shows cross validation on random data splits tend to underestimate in structured data. Furthermore the partition is shifted cyclically $R$ times to reduce the dependence on a particular partition.
A test block of the length $\lfloor n/K \rfloor$ is given by a selection of indices for each $k$-th fold, $k=1,\hdots,K$ and for a $r$ cyclically shift, $r=1,\hdots,R$:
\begin{align*}
    I^{\text{test}}_{k,r} = \big\{\, i \in \{1,\hdots,n\} :\ (k-1) L_K \le (i-1+s_r) \bmod (K L_K ) < k L_K \,\big\},
\end{align*}
where $L_K = \lfloor n/K \rfloor$ and $s_r =\lfloor ((r-1)L_K)/R\rfloor.$
Blocks of contiguous time ensure greater independence between cross-validation folds, see \citet{racine2000} and \citet{Roberts2017}. Further, to decrease the dependency between the test block and a training block we reduce the indices training set $I^{\text{train}}_{k,r}$ by a gap length $g$ in each direction

\begin{align*}
    I^{\text{train}}_{k,r} = \{1,\hdots,n\} \setminus G_{k,r},
\end{align*}
where $G_{k,r} = \big\{\, i \in \{1,\hdots,n\} :\ |i-j| \le g \ \text{for some } j \in I^{\text{test}}_{k,r} \,\big\}.$
Evaluating the univariate local polynomial estimators and bivariate estimators $K$-times requires a given grid of bandwidths $h_w$, $w=1,\hdots,W.$ Finally we take the mean of these $K$ $L_2$-norm errors to obtain a bandwidth $h^{\text{\footnotesize cv}}$ corresponding to the minimal average $L_2$-norm error
\begin{align}
      h^{\text{\footnotesize cv}}=\text{arg}\!\min \big\{\text{CV}(h_w)\mid w=1,\hdots,W\big\}, \label{bandwidth CV}
\end{align}
where $\text{CV}(\cdotp)$ is defined by
\begin{align*}
     \text{CV}(h_{d,w}) = \frac{1}{KR} \,\sum_{r=1}^R\sum_{k=1}^K \Big(\frac{1}{p_d}\sum_{l = 1}^{p_d} \big( \bar{Y}^{\dense,\text{test},k,r}_{l}-\hat{\mu}^{\dense}(t^{\dense}_l;h_{d,w},Y^{\dense,\text{train},-k,r}_{i,l})\big)^2\Big)^{1/2}
\end{align*}
in case of bandwidth selection for $\hat{\mu}^{\dense}_{m}$ and
\begin{align*}
  \text{CV}(h_w) = \frac{1}{KR} \,\sum_{r=1}^R\sum_{k=1}^K \Big(\frac{1}{p}\sum_{j = 1}^{p}\big(\bar{Y}^{\sparse,\text{test},k,r}_{j}-\hat{\mu}^{\dense}_{m}(t_j^{\sparse};h_d^{\text{\footnotesize cv}}) -\hat{\delta}_{n,m}(t^{\sparse}_j;h_w, h_d^{\text{\footnotesize cv}},Y^{\sparse,\text{train},-k,r}_{i,j})\big)^2\Big)^{1/2}
\end{align*}
in case of $\hat{\delta}_{n,m}$.

By considering bivariate lagged kernel estimators for $|b|\geq 0$ we use $\{Z_i^{\sparse}(t^{\sparse}_j)\mid i \in I^{\text{test}}_{k,r}, j =1,\hdots,p\}$ as the empirical lagged covariance matrix denoted by $(Z^{\sparse,\text{test},k,r}_{j,l})_{j,l=1,\hdots,p}$ and compute the empirical covariance matrix by taking the remaining train data set $\{Z_i^{\sparse}(t^{\sparse}_j) \mid i \in I^{\text{train}}_{k,r}, j =1,\hdots,p\}$ as input to the estimator \eqref{eq:Gamma_estimator} or \eqref{eq:crosscov}.
For each index group $I^{\text{test}}_{k,r}$ as the test set evaluating the bivariate local polynomial estimators \eqref{eq:Gamma_estimator} and \eqref{eq:crosscov} $KR$-times we choose a bandwidth $h^{\text{\footnotesize cv}}$ in \eqref{bandwidth CV} corresponding to the minimal average sup-norm error,
that is
\begin{align*}
    \text{CV}(h_w) = \frac{1}{KR} \,\sum_{r=1}^R\sum_{k=1}^K \max_{1\leq j < l \leq p} \big|\hat{\gamma}^{\sparse}_{b}(t^{\sparse}_j,t^{\sparse}_l;h_w,Z_{j,l}^{\sparse,\text{train},-k,r})-Z^{\sparse,\text{test},k,r}_{j,l}\big|
\end{align*} for $b=0$ and for $b\geq 1$
\begin{align*}
    \text{CV}(h_w) = \frac{1}{KR} \,\sum_{r=1}^R\sum_{k=1}^K \max_{1\leq j \leq l \leq p} & \big|\hat{\gamma}^{\sparse}_{b}(t^{\sparse}_j,t^{\sparse}_l;h_w,Z_{j,l}^{\sparse,\text{train},-k,r})+\hat{\gamma}^{\sparse}_{b}(t^{\sparse}_j,t^{\sparse}_l;h_w,Z_{j,l}^{\sparse,\text{train},-k,r})^\top\\
    &-Z^{\sparse,\text{test},k,r}_{j,l}-(Z^{\sparse,\text{test},k,r}_{j,l})^\top\big|.
\end{align*}
Due to symmetry for $b$ and $-b,$ $h^{\text{\footnotesize cv}}$ is chosen equally. Unless otherwise stated, we set $R =5,$ $K = 5$ and $g=5$ for further analysis.

\section{Functional time series}\label{Appendix}

\subsection{Setting and long-run covariance kernels}\label{sec:timeseries}
\begin{assumption}[Time series of processes]\label{assumption:timeseries}
      The processes $(Z^{\sparse}_i)_{i \in \N}$ and $(Z^{\dense}_k)_{k \in \N}$ both are stationary and ergodic sequences, independent of each other,  of stochastic processes on $[0,1]$ with finite fourth moments which satisfy the H\"older condition (\ref{eq:Z}).
      For the observational errors $\epsilon^{\sparse}_{i,j}$, $\epsilon^{\dense}_{k,l}$ we keep the independence and distributional assumptions from Assumption~\ref{assumption:Independence}.
\end{assumption}

\begin{remark}\label{assumption:tskernels}
    For the stationary time series setting of Assumption \ref{assumption:timeseries} we introduce the cross-covariance kernel defined for $b \in \Z$ and $t,s \in [0,1]$ by
    \begin{align*}
        \gamma^{\sparse}_b(t,s) = \expec[Z_{1}^{\sparse}(t)\,Z_{1+b}^{\sparse}(s)], \qquad \gamma^{\dense}_b(t,s) = \expec[Z_{1}^{\dense}(t)\,Z_{1+b}^{\dense}(s)].
    \end{align*}
    By stationarity,
    \begin{equation}\label{eq:simcrosscov}
        \gamma^{\sparse}_{-b}(s,t) = \gamma^{\sparse}_b(t,s).
    \end{equation}
    We will need to impose additional assumptions which guarantee dependent Jain-Marcus central limit theorems for the sequences $(Z^{\sparse}_i)_{i \in \N}$ and $(Z^{\dense}_k)_{k \in \N}$, under which the long-run covariance kernels
    $$ \Gamma^{\sparse}(t,s) = \gamma^{\sparse}_0(t,s) + \sum_{b=1}^\infty \Big(\gamma^{\sparse}_b(t,s) + \gamma^{\sparse}_{-b}(t,s)\Big)\,, \qquad t,s \in [0,1],$$
    similarly for  $\Gamma^{\dense}(t,s) $, exist and are continuous, and under which
    \begin{align*}
        \frac1{n^{1/2}}\, \big(Z_{1}^{\sparse} + \ldots + Z_{n}^{\sparse} \big) \overset{D}{\longrightarrow} \mathcal{G}(0,\Gamma^{\sparse}),
    \end{align*}
    where weak convergence is in $C[0,1]$. \citet[Theorem 2.1]{dette2020} present such a result for $\varphi$-mixing functional time series with mixing coefficients which decrease sufficiently fast. More generally, this follows from a Donsker theorem for dependent sequences when taking the evaluation function $F_s:C[0,1] \to \R$, $f \mapsto f(s)$ as indexing function class, see e.g.~\citet[Example 2.11.13]{vaart1996}, which allows to deduce Jain Marcus CLTs for time series under the mixing conditions stated e.g.~in \citet[Theorem 5.2]{dedecker2002maximal}.

    These central limit theorems for the two sequences are condition \eqref{eq:cltprocesses}, with $\Gamma^{\sparse}$ and $\Gamma^{\dense}$ the long-run covariance kernels. Under this condition the asymptotic normality in (\ref{eq:distinctorder}) and (\ref{eq:sameorder}) continues to hold with the long-run covariance kernels: the proof of Theorem~\ref{theorem:CLT_delta} in Section~\ref{sec:proofthm} uses the processes only through \eqref{eq:cltprocesses}, and the bounds on the bias and error terms do not involve the processes. In particular, no central limit theorem for triangular arrays of smoothed processes is required. The same applies to the limit \eqref{eq:sameorderint1} for the average difference, see Section~\ref{sec:proofint}.
\end{remark}

\subsection{Estimating covariance and long-run covariance kernels}\label{sec:covest}
\begin{remark}[Studentization and covariance kernel estimates]
We work with studentization, and therefore require estimates of the covariance kernels $\gamma^{\sparse}_0$ and $\gamma^{\dense}_0$, and in the time-series setting also of the cross-covariance kernels and of the long-run covariance kernel.

For the covariance kernels we use the estimator introduced in \citet{berger2025}, which for $\gamma^{\sparse}_0$ is defined as
\begin{align}
    \hat{\gamma}^{\sparse}_{0}(t,s)\defeq \frac{1}{ n-1}\sum_{i=1}^{ n}\sum_{1\leq j<l \leq  p} w^{\sparse}_{j,l}(t,s;{h})\big(Y^{\sparse}_{i,j}Y^{\sparse}_{i,l}-\bar{Y_{j}}^{\sparse}\bar{Y_{l}}^{\sparse}\big), \qquad t \leq s, \label{eq:Gamma_estimator}
\end{align}
where the weights $w^{\sparse}_{j,l}(t,s;{h})$ only build on design point pairs $(t^{\sparse}_j, t^{\sparse}_l)$, $j < l$ above the diagonal. The idea is to make use of higher-order smoothness of $\gamma^{\sparse}_0$ if there is a non-differentiable kink on the diagonal. See \citet{berger2025} for further details.
The estimator would then be extended by symmetry, $w^{\sparse}_{j,l}(t,s;{h}) \defeq w^{\sparse}_{j,l}(s,t;{h})$, $t > s$, resulting in a continuous function over all $(t,s)$. %
$\hat{\gamma}^{\dense}_{0}$ is defined analogously.

In the time-series setting we also need to estimate cross-covariance kernels and the long-run covariance kernel. For a lag $b \geq 1$, if we use the weights from \eqref{eq:Gamma_estimator} we can set-up the estimator
\begin{align}\label{eq:crosscov}
    \hat{\gamma}^{\sparse}_{b}(t,s)\defeq \frac{1}{n-1}\sum_{i=1}^{n-b}\sum_{1\leq j \leq l\leq p} w^{\sparse}_{j,l}(t,s;{h})\big(Y^{\sparse}_{i,j}Y^{\sparse}_{i+b,l}-\bar{Y_{j}}^{\sparse}\bar{Y_{l}}^{\sparse}\big), \qquad t \leq s,
\end{align}
for $b \leq -1$ the estimate is set-up similarly.
From the symmetry relation \eqref{eq:simcrosscov} the estimates \eqref{eq:crosscov}, $t \leq s$, together with those for negative $b$ again suffice to estimate full cross-covariance kernels, setting $\hat{\gamma}^{\sparse}_{b}(t,s) \defeq \hat{\gamma}^{\sparse}_{-b}(s,t)$ for $t > s$. Note that we get two distinct estimates on the diagonal $t=s$ resulting from $b$ and $-b$. We use the average of these estimates in the following. In \eqref{eq:crosscov} we normalize by $1/(n-1)$ for every lag, as in \eqref{eq:Gamma_estimator}, although the sum contains only $n-b$ lagged pairs. This differs from averaging over the available pairs by the factor $(n-b)/(n-1)$, which tends to one uniformly in $0 \leq b \leq L$ if the maximal lag satisfies $L/n \to 0$; under this condition the two normalizations are asymptotically equivalent also when $L$ grows with $n$. In our application, $L \leq 4$. Finally, for the long-run covariance kernel we take
\begin{align*}
\widehat{\Gamma}^{\sparse}(t,s) = \hat{\gamma}^{\sparse}_{0}(t,s) + \sum_{b=1}^L w_{b;L}\, \Big(\hat{\gamma}^{\sparse}_{b}(t,s) + \hat{\gamma}^{\sparse}_{-b}(t,s)\Big)\,, \qquad t\leq s,
\end{align*}
where each $\hat{\gamma}^{\sparse}_{b}$ uses its own bandwidth (Section~\ref{sec:detailsmehtod}), $L \in \N$ is a tuning parameter and $w_{b;L}$ are weights. We use the Bartlett weights $w_{b;L} \defeq 1-\frac{b}{L+1}$ from \citet{NeweyWest1994}, and similarly define the estimator $\widehat{\Gamma}^{\dense}$ of the long run covariance kernel $\Gamma^{\dense}$.
Now for studentization we only need consistency for the covariance kernel estimates, not the $\sqrt n$-rate, which will be satisfied also for the sparse sample, or sparsely observed time series.
Details on rates of convergence for the covariance kernel under independence are presented in  \citet{berger2025}.
Their analysis can be extended to cross-covariance kernels under suitable assumptions on the time series, however, proving consistency for the estimate of the long-run covariance kernel is outside the scope of present paper.
Such consistency requires a maximal lag $L$ that increases with $n$, in which case the Bartlett weights tend to one at each fixed lag. These weights are a standard choice: for empirical autocovariances without smoothing they give a positive semidefinite estimate \citep{NeweyWest1987}, and by downweighting the higher lags they reduce the variability of the estimate. In the simulations and in the application we use small maximal lags, $L \leq 4$. As discussed in Remark~\ref{rem:studnoncons}, $\widehat{\Gamma}^{\sparse}$ then estimates a Bartlett-weighted sum of the lagged covariance kernels rather than $\Gamma^{\sparse}$, even if the dependence ends at lag $L$; with weights equal to one the estimate can be consistent only in that case. For the level of the confidence bands and tests this is not essential, since the same estimate $\hat{\sigma}_{\Delta}$ normalizes the statistic in (\ref{eq:std statistic 1}) and the bootstrap process in (\ref{eq:multiboot}), while the dependence between the curves enters the bootstrap through the block multipliers of Section~\ref{sec:multipliers}, whose block length increases with $n$: if $\hat{\sigma}_{\Delta}$ converges uniformly in probability to a positive continuous function, the convergence in (\ref{eq:std statistic 1}) still holds by the proof in Section~\ref{sec:proofcor}, and the band keeps its asymptotic level, provided the bootstrap approximates the distribution of the numerator. The coverage and the power in Section~\ref{subsection: Simulation} are simulated for this procedure, with $L = L_d = 3$.
In our application, curves from different years are treated as independent, so that the curves of a month are dependent only within a yearly segment of at most 31 days, and the lag test selects at most four lags; the estimator for these segmented series is given in Section~\ref{sec:detailsmehtod}. Replacing the Bartlett weights by weights equal to one changes no Holm decision in Tables~\ref{tab: Main global tests} and \ref{tab: other main global tests} for any of the four cities; the Holm-adjusted p-values differ at most in the third decimal place, and in one case in the second.
\end{remark}

\subsection{Dependent multipliers for the bootstrap}\label{sec:multipliers}
\begingroup
\emergencystretch=2em\relax
For time series, the multipliers in (\ref{eq:multiboot}) need to have certain dependencies. In our implementation we follow the recommendations in  \citet{Buecher2013}: Setting
\begin{align*}
    \kappa_1(b;n) \defeq \begin{cases}
        \frac{1}{2\,l(n) -1}, &\text{ if } |b| < l(n),\\
        0,&\text{ otherwise,}
    \end{cases}  \quad \text{ and } \quad \kappa_2(b;n) \defeq \max\Big(0,\frac{1-|b|/l(n)}{l(n)}\Big)
\end{align*}
we simulate
\begin{align}
    \xi_{i,n}^{(r)}(j) \defeq \sum_{b = -\infty}^{\infty} \kappa_j(b;n)\, W^{\sparse,(r)}_{i+b} \quad \text{and} \quad \eta^{(r)}_{k,m}(j) \defeq \sum_{b = -\infty}^{\infty} \kappa_j(b;m)\, W^{\dense,(r)}_{k+b},\label{eq:taperedmultipliers}
\end{align}
where $j=1$ and $W^{\sparse,(r)}_i$ and $W^{\dense,(r)}_k$ are independently normally distributed as $W^{\sparse,(r)}_i \overset{i.i.d.}{\sim} \mathcal{N}(0,1/q(n))$ and $W^{\dense,(r)}_k \overset{i.i.d.}{\sim} \mathcal{N}(0,1/q(m))$ with $q(n) \defeq 1/(2\,l(n)-1).$ Qualitatively similar to \citet{Buecher2013} we set $l(n) = \lfloor 2\,n^{1/3}\rfloor$.
\endgroup

\section{Additional numerical results}\label{sec: add results}

\subsection{Results for the simulations in Section~\ref{section:quantiles}}\label{sec:additionalsim}

Tables~\ref{tab:bandwidth mean delta constant} and \ref{tab:bandwidth mean delta not constant} contain the averaged cross-validated bandwidths for the mean functions in the settings with constant and with non-constant difference function in Section~\ref{subsection: Simulation}.

\begin{table}[h!]
\begingroup
    \centering{
    \scriptsize
    \renewcommand{\arraystretch}{1.3}
    \begin{tabular}{c|c|c|c|c|c|c}
         $\bar{h}_{d,N}^{\text{\footnotesize cv}}$ & $\,m = 30\,$  & $\,m = 60\,$ & $m = 120$ & $m = 360$ & $m = 720$ & $m = 1200$ \\
        \hline

        \rule{0pt}{2.5ex}$p_d = 100$ & $0{.}22$ & $0{.}18$ & $0{.}16$ & $0{.}12$ & $0{.}11$ & $0{.}1$ \\
        \hline
        & & & & &\\
        $\bar{h}_{N}^{\text{\footnotesize cv}}$ & $\,n = 25\,$ & $\,n = 50\,$ & $n = 100$ & $n = 300$ & $n = 600$ & $n = 1000$ \\
        \hline

        \rule{0pt}{2.5ex} $p = 75$ & $0{.}65$ & $0{.}66$ & $0{.}65$ & $0{.}66$ & $0{.}65$ & $0{.}65$   \\

        \rule{0pt}{2.5ex} $p = 50$ & $0{.}65$  & $0{.}64$ & $0{.}65$ & $0{.}63$ & $0{.}63$ & $0{.}64$ \\

        \rule{0pt}{2.5ex} $p = 25$ & $0{.}67$ & $0{.}65$ & $0{.}67$ & $0{.}67$ & $0{.}64$ & $0{.}66$  \\

    \end{tabular}
    }
    \caption{Bandwidth selection for $\hat{\mu}^{\dense}_{m}(\cdotp;\bar{h}_{d,N}^{\text{\footnotesize cv}})$ and $\hat{\delta}_{n,m}(\cdotp;\bar{h}_{N}^{\text{\footnotesize cv}},\bar{h}_{d,N}^{\text{\footnotesize cv}})$, mean of $N=1000$ repetitions.}
    \label{tab:bandwidth mean delta constant}
\endgroup
\end{table}

\begin{table}[h!]
\begingroup
    \centering{
    \scriptsize
    \renewcommand{\arraystretch}{1.3}
    \begin{tabular}{c|c|c|c|c|c|c}
         $\bar{h}_{d,N}^{\text{\footnotesize cv}}$ & $\,m = 30\,$  & $\,m = 60\,$ & $m = 120$ & $m = 360$ & $m = 720$ & $m = 1200$ \\
        \hline

        \rule{0pt}{2.5ex}$p_d = 100$ & $0{.}21$ & $0{.}17$ & $0{.}15$ & $0{.}12$ & $0{.}1$ & $0{.}09$ \\
        \hline
        & & & & &\\
        $\bar{h}_{N}^{\text{\footnotesize cv}}$ & $\,n = 25\,$ & $\,n = 50\,$ & $n = 100$ & $n = 300$ & $n = 600$ & $n = 1000$ \\
        \hline

        \rule{0pt}{2.5ex} $p = 75$ & $0{.}61$ & $0{.}54$ & $0{.}43$ & $0{.}28$ & $0{.}22$ & $0{.}19$ \\

        \rule{0pt}{2.5ex} $p = 50$ & $0{.}61$ & $0{.}55$ & $0{.}45$ & $0{.}28$ & $0{.}23$ & $0{.}2$ \\

        \rule{0pt}{2.5ex} $p = 25$ & $0{.}61$ & $0{.}54$ & $0{.}45$ & $0{.}29$ & $0{.}23$ & $0{.}2$ \\

    \end{tabular}
    }
    \caption{Bandwidth selection for $\hat{\mu}^{\dense}_{m}(\cdotp;\bar{h}_{d,N}^{\text{\footnotesize cv}})$ and $\hat{\delta}_{n,m}(\cdotp;\bar{h}_{N}^{\text{\footnotesize cv}},\bar{h}_{d,N}^{\text{\footnotesize cv}})$, mean of $N=1000$ repetitions.}
    \label{tab:bandwidth mean delta not constant}
\endgroup
\end{table}

Table~\ref{tab:bwseparate} contains the corresponding bandwidths for smoothing the sparse sample on its own, as the estimator $\bar{\delta}_{n,m}$ in (\ref{eq:diffest}) requires. They are much smaller than those for $\hat{\delta}_{n,m}$, in line with Remark~\ref{rem:alphadelta}.

\begin{table}[h!]
    \centering{
    \scriptsize
    \renewcommand{\arraystretch}{1.3}
    \begin{tabular}{c|c|c|c|c|c|c}
    $\bar{h}_{N}^{\text{\footnotesize cv}}$ & $\,n = 25\,$ & $\,n = 50\,$ & $n = 100$ & $n = 300$ & $n = 600$ & $n = 1000$ \\
    \hline
    \multicolumn{7}{c}{\rule{0pt}{2.5ex}Constant difference function} \\
    \hline
    \rule{0pt}{2.5ex}$p = 75$ & $0{.}23$ & $0{.}19$ & $0{.}16$ & $0{.}13$ & $0{.}11$ & $0{.}1$ \\
    \rule{0pt}{2.5ex}$p = 50$ & $0{.}23$ & $0{.}19$ & $0{.}16$ & $0{.}13$ & $0{.}11$ & $0{.}1$ \\
    \rule{0pt}{2.5ex}$p = 25$ & $0{.}24$ & $0{.}19$ & $0{.}17$ & $0{.}13$ & $0{.}11$ & $0{.}11$ \\
    \hline
    \multicolumn{7}{c}{\rule{0pt}{2.5ex}Non-constant difference function} \\
    \hline
    \rule{0pt}{2.5ex}$p = 75$ & $0{.}24$ & $0{.}19$ & $0{.}16$ & $0{.}13$ & $0{.}11$ & $0{.}1$ \\
    \rule{0pt}{2.5ex}$p = 50$ & $0{.}23$ & $0{.}19$ & $0{.}16$ & $0{.}13$ & $0{.}11$ & $0{.}1$ \\
    \rule{0pt}{2.5ex}$p = 25$ & $0{.}24$ & $0{.}19$ & $0{.}17$ & $0{.}13$ & $0{.}12$ & $0{.}11$ \\
    \end{tabular}
    }
    \caption{Bandwidth selection for the separate smoother $\bar{\mu}^{\sparse}_{n}$ of the sparse sample in $\bar{\delta}_{n,m}$, with $m = 1{.}2\,n$ and $p_d = 100$, mean of $N=1000$ repetitions. The bandwidths for $\hat{\mu}^{\dense}_{m}$ and $\hat{\delta}_{n,m}$ are in Tables~\ref{tab:bandwidth mean delta constant} and \ref{tab:bandwidth mean delta not constant}.}
    \label{tab:bwseparate}
\end{table}

Similar to Table~\ref{tab:bandwidth mean delta constant} by applying the $5$-shifted hv-block cross validation in a $5$-fold framework and averaging over $N=1000$ sample repetitions, Figure \ref{fig:combined}(a) contains the selected bandwidths of the long run covariance kernel estimator $\widehat{\Gamma}^{\sparse}$ with $L=3$ for $p\in \{25,50,75\}$ and $n = 1000$, respectively the bandwidths of $\widehat{\Gamma}^{\dense}$ with $L_d=3$ for $p_d = 100$ and $m = 1200$, by lag $b$. In addition Figure \ref{fig:combined}(b) shows the corresponding boxplots of this repeated method, whose observed increase in lag can be explained by a smoother structure as the lag increases. The variability in each boxplot may be due to additional variation resulting from the estimation in the test data.

\begin{figure}[h!]
    \centering
    \scriptsize

    \begin{subfigure}[b]{0.48\linewidth}
        \centering
        \renewcommand{\arraystretch}{1.3}
        \begin{tabular}{c|c|ccc}
             & $m = 1200$ & \multicolumn{3}{c}{$n = 1000$} \\
             lag & $p_d = 100$ & $p = 75$ & $p = 50$ & $p = 25$ \\
            \hline
            \rule{0pt}{2.8ex}$b=0$ & $0{.}34$ & $0{.}38$ & $0{.}39$ & $0{.}41$ \\
            \rule{0pt}{2.8ex}$b=1$ & $0{.}42$ & $0{.}46$ & $0{.}45$ & $0{.}48$ \\
            \rule{0pt}{2.8ex}$b=2$ & $0{.}6$ & $0{.}63$ & $0{.}61$ & $0{.}65$ \\
            \rule{0pt}{2.8ex}$b=3$ & $0{.}67$ & $0{.}7$ & $0{.}7$ & $0{.}72$ \\
        \end{tabular}\\
        \vspace{5mm}
        \caption{Mean.}
    \end{subfigure}
    \hfill
    \begin{subfigure}[b]{0.48\linewidth}
        \centering
        \includegraphics[width=\linewidth]{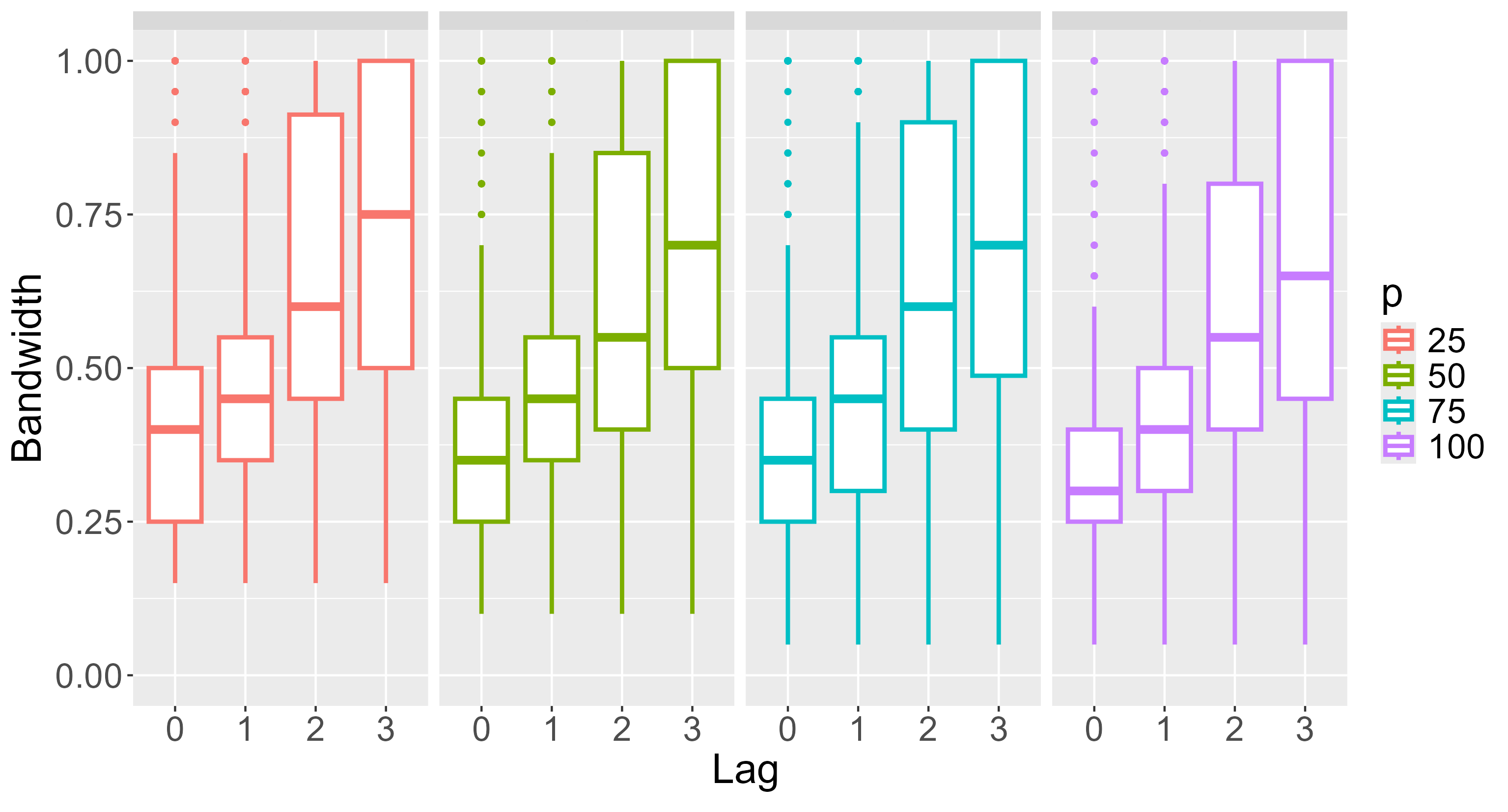}
        \caption{Boxplot for $n=1000$.}
    \end{subfigure}

    \caption{Bandwidth selection for the lag-$b$ estimates in $\widehat{\Gamma}^{\dense}$ and $\widehat{\Gamma}^{\sparse}$, $L=L_d=3$, with $N=1000$ repetitions.}
    \label{fig:combined}
\end{figure}

\begin{table}[h!]
    \centering{
    \scriptsize
    \renewcommand{\arraystretch}{1.3}
    \begin{tabular}{c|c|c|c|c}
          lag & $b=0$ & $b=1$ & $b=2$ & $b=3$\\
        \hline
        \rule{0pt}{2.8ex} $\,\,m = 1200, p_d = 100$ & $0{.}35$ & $0{.}42$ & $0{.}6$ & $0{.}69$ \\
    \end{tabular}
    }
    \caption{Bandwidth selection for the lag-$b$ estimates in $\widehat{\Gamma}^{\dense}$, $L_d=3$, mean of $N=1000$ repetitions, in case of non-constant difference function.}
    \label{tab:bandwidth cov delta not constant}
\end{table}

For illustration, Figure \ref{fig:scaled coverage} provides the coverage similar to Figures~\ref{fig:coverage rate under H0} and ~\ref{fig:coverage and power} for a constant difference and non-constant difference with in $L_2$-norm cross-validated corresponding bandwidths scaled by $\log(n)^{-0.2}$, which corresponds to the logarithmic factor in the upper end of the bandwidth range of Theorem~\ref{theorem:CLT_delta} with $\eta = 0{.}2$, since $\log p$ and $\log n$ are of the same order here.

\begin{figure}[h!]
    \centering
    \begin{subfigure}[b]{0.9\linewidth}
        \includegraphics[width=\linewidth]{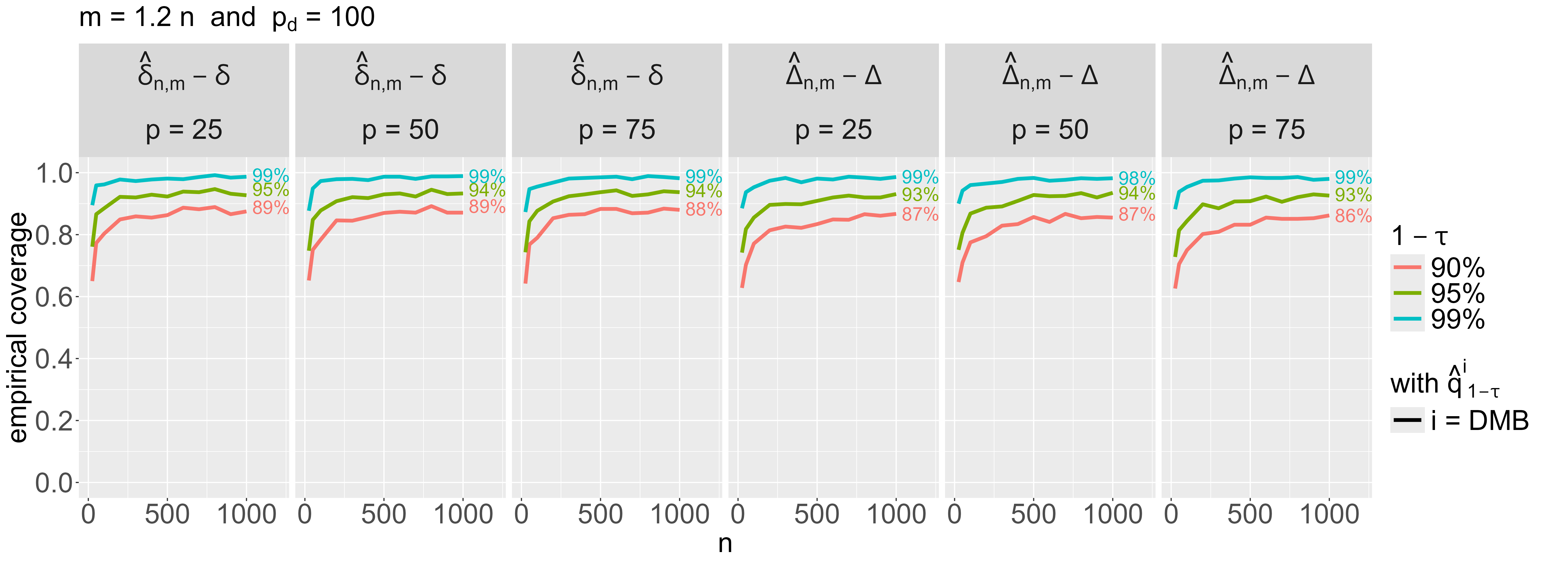}
    \end{subfigure}

    \hfill

    \begin{subfigure}[b]{0.9\linewidth}
        \includegraphics[width=\linewidth]{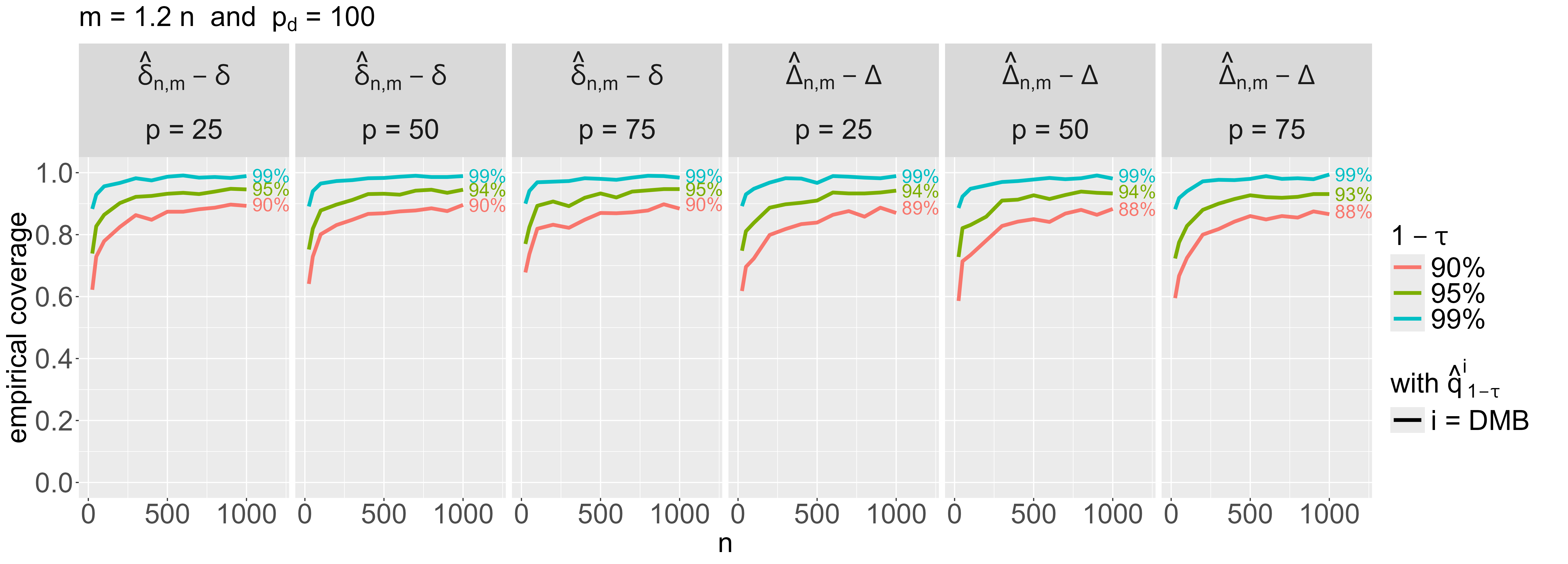}
    \end{subfigure}

    \caption{With cross-validated bandwidths scaled by $\log(n)^{-0.2}$ simulated coverage of the confidence band of test for constant difference (upper panel) and non-constant difference (lower panel) using the dependent multiplier bootstrap with $N^* = 1000$ on $1000$ repeated simulated samples analogously to Figures \ref{fig:coverage rate under H0} and ~\ref{fig:coverage and power}.}
    \label{fig:scaled coverage}
\end{figure}

\subsection{Methodological details for Section~\ref{sec:apllilust}}\label{sec:detailsmehtod}

First we discuss the implementation of cross-validation for choosing the bandwidths $h$ and $h_d$ in the estimator of the difference function, as well as the bandwidths $h_0$ and $h_{d,0}$ used for estimating covariance kernels.
To maximize the amount of data used for model fitting in each cross-validation run, each month of the year can be treated as a separate fold since the fulfillment of the required independence assumptions is reasonable which is in accordance with  \citet{Roberts2017}. Due to the natural partition into blocks, no shifting is required, i.e., we set $R = 1$. To stabilize the bandwidth selection and reduce the computational effort, we sum the months respectively to their season and run the cross-validation method on them, which is reasonable since these months are similar in terms of their patterns.

\begin{table}[h!]
    \centering{
    \scriptsize
    \renewcommand{\arraystretch}{1.3}
    \begin{tabular}{c|c|c|c|c|c|c|c|c|c|c|c|c}
                         Month                                        & Jan     & Feb     & Mar     & Apr     & May     & Jun     & Jul     & Aug     & Sep     & Oct     & Nov     & Dec     \\
    \hline
        \multicolumn{13}{c}{\rule{0pt}{2.5ex}2000--2025 (dense sample)} \\
        \rule{0pt}{2.5ex}$m$                                   & $806$   & $734$   & $804$   & $777$   & $803$   & $770$   & $806$   & $803$   & $773$   & $803$   & $777$   & $806$   \\
        \hline
        \rule{0pt}{2.5ex} $h_d^{\text{\footnotesize cv}}$   (h) & $1{.}7$ & $1{.}7$ & $1{.}5$ & $1{.}5$ & $1{.}5$ & $1{.}4$ & $1{.}4$ & $1{.}4$ & $1{.}4$ & $1{.}4$ & $1{.}4$ & $1{.}7$ \\
        \rule{0pt}{2.5ex} $h_{d,0}^{\text{\footnotesize cv}}$ (h) & $1{.}8$ & $1{.}8$ & $1{.}8$ & $1{.}8$ & $1{.}8$ & $1{.}8$ & $1{.}8$ & $1{.}8$ & $1{.}8$ & $1{.}8$ & $1{.}8$ & $1{.}8$ \\
        \rule{0pt}{2.5ex} $L_d$                                         & $2$     & $2$     & $2$     & $2$     & $3$     & $4$     & $2$     & $2$     & $2$     & $2$     & $2$     & $1$ \\
        \hline
        \multicolumn{13}{c}{\rule{0pt}{2.5ex}1952--1972 (sparse sample)} \\
        \rule{0pt}{2.5ex}$n$                                          & $651$   & $594$   & $651$   & $630$   & $651$   & $630$   & $651$   & $651$   & $630$   & $651$   & $630$   & $651$   \\
        \hline
        \rule{0pt}{2.5ex} $h^{\text{\footnotesize cv}}$           (h) & $2{.}7$ & $2{.}7$ & $4{.}7$ & $4{.}7$ & $4{.}7$ & $4{.}1$ & $4{.}1$ & $4{.}1$ & $4{.}0$ & $4{.}0$ & $4{.}0$ & $2{.}7$ \\
        \rule{0pt}{2.5ex} $h_0^{\text{\footnotesize cv}}$         (h) & $3{.}4$ & $3{.}4$ & $3{.}4$ & $3{.}4$ & $3{.}4$ & $3{.}4$ & $3{.}4$ & $3{.}4$ & $3{.}4$ & $3{.}4$ & $3{.}4$ & $3{.}4$ \\
        \rule{0pt}{2.5ex} $L$                                         & $2$     & $2$     & $2$     & $2$     & $3$     & $4$     & $2$     & $2$     & $2$     & $2$     & $2$     & $1$

    \end{tabular}
    }
    \caption{Bandwidth selection for $\hat{\mu}^{\dense}_{m}(\cdotp;h_d^{\text{\footnotesize cv}})$, $\hat{\delta}_{n,m}(\cdotp;h^{\text{\footnotesize cv}},h_d^{\text{\footnotesize cv}})$, $\hat{\gamma}^{\sparse}_0$ (bandwidth $h_0^{\text{\footnotesize cv}}$) and $\hat{\gamma}^{\dense}_0$ (bandwidth $h_{d,0}^{\text{\footnotesize cv}}$).}
    \label{tab:Berlin bandwidth application}
\end{table}

Table \ref{tab:Berlin bandwidth application} shows the bandwidth selection in Berlin for the estimator of the averaged daily weather curve from $2000$ to $2025$ ($\mu^{\dense}$), of the difference function ($\delta$) and of the lag $0$ covariance kernels from $2000$ to $2025$ ($\Gamma^{\dense}$) as well as from $1952$ to $1972$ ($\Gamma^{\sparse}$).
The larger bandwidths chosen in the second step for estimating $\delta$ compared to those for $\mu^{\dense}$ from the first step are consistent with the assumption that $\alpha_{\delta} \geq \alpha$. Further, based on the insights from the simulation in Section~\ref{subsection: Simulation}, which show that the bandwidth  does not substantially influence the structure of the covariance kernel estimate, we decided to choose fixed values in the application, and for the cross covariance kernels take $h_b^{\text{\footnotesize cv}}=1{.}1\cdot h_{b-1}^{\text{\footnotesize cv}}$ for $b = 1,\hdots,L$ and $h_{d,b}^{\text{\footnotesize cv}}=1{.}1\cdot h_{d,b-1}^{\text{\footnotesize cv}}$ for $b = 1,\hdots,L_d$.

Table~\ref{tab:bwcities} compares the cross-validated bandwidths of the residual-based estimator $\hat{\delta}_{n,m}$ with those of the estimator $\bar{\delta}_{n,m}$ in (\ref{eq:diffest}), which smooths the historical sample on its own, for all four cities. The bandwidths are selected for each season. In all combinations of city and season, the bandwidth for $\hat{\delta}_{n,m}$ is larger, by a factor of 1.2 to 2.4. By Remark~\ref{rem:alphadelta}, this indicates that the difference function is smoother than the historical mean function, $\alpha_\delta > \alpha$.

\begin{table}[h!]
    \centering{
    \scriptsize
    \renewcommand{\arraystretch}{1.3}
    \begin{tabular}{l|c|c|c|c}
    Bandwidth (h) for & Dec--Feb & Mar--May & Jun--Aug & Sep--Nov \\
    \hline
    \multicolumn{5}{c}{\rule{0pt}{2.5ex}Berlin} \\
    \hline
    $\hat{\mu}^{\dense}_{m}$, 2000--2025 & $1{.}7$ & $1{.}5$ & $1{.}4$ & $1{.}4$ \\
    $\hat{\delta}_{n,m}$, 1952--1972 & $2{.}7$ & $4{.}7$ & $4{.}1$ & $4{.}0$ \\
    $\bar{\mu}^{\sparse}_{n}$ in $\bar{\delta}_{n,m}$, 1952--1972 & $2{.}0$ & $2{.}0$ & $3{.}4$ & $2{.}0$ \\
    \hline
    \multicolumn{5}{c}{\rule{0pt}{2.5ex}Frankfurt} \\
    \hline
    $\hat{\mu}^{\dense}_{m}$, 2000--2025 & $1{.}4$ & $1{.}9$ & $1{.}4$ & $1{.}4$ \\
    $\hat{\delta}_{n,m}$, 1952--1972 & $4{.}0$ & $4{.}4$ & $3{.}2$ & $4{.}0$ \\
    $\bar{\mu}^{\sparse}_{n}$ in $\bar{\delta}_{n,m}$, 1952--1972 & $3{.}1$ & $2{.}0$ & $2{.}0$ & $3{.}4$ \\
    \hline
    \multicolumn{5}{c}{\rule{0pt}{2.5ex}Hamburg} \\
    \hline
    $\hat{\mu}^{\dense}_{m}$, 2000--2025 & $1{.}5$ & $1{.}4$ & $1{.}7$ & $2{.}5$ \\
    $\hat{\delta}_{n,m}$, 1952--1972 & $2{.}7$ & $4{.}0$ & $3{.}4$ & $3{.}0$ \\
    $\bar{\mu}^{\sparse}_{n}$ in $\bar{\delta}_{n,m}$, 1952--1972 & $2{.}0$ & $2{.}0$ & $2{.}0$ & $2{.}0$ \\
    \hline
    \multicolumn{5}{c}{\rule{0pt}{2.5ex}Munich} \\
    \hline
    $\hat{\mu}^{\dense}_{m}$, 2000--2025 & $1{.}4$ & $1{.}4$ & $1{.}4$ & $3{.}1$ \\
    $\hat{\delta}_{n,m}$, 1952--1972 & $3{.}1$ & $4{.}0$ & $3{.}1$ & $3{.}0$ \\
    $\bar{\mu}^{\sparse}_{n}$ in $\bar{\delta}_{n,m}$, 1952--1972 & $2{.}1$ & $2{.}0$ & $2{.}0$ & $2{.}0$ \\
    \end{tabular}
    }
    \caption{Cross-validated bandwidths in hours, selected for each season: for the mean $\hat{\mu}^{\dense}_{m}$ of 2000--2025, for the residual-based estimator $\hat{\delta}_{n,m}$, and for the separate smoother $\bar{\mu}^{\sparse}_{n}$ of 1952--1972 in $\bar{\delta}_{n,m}$. Both estimators of $\delta$ use the same $\hat{\mu}^{\dense}_{m}$.}
    \label{tab:bwcities}
\end{table}

Furthermore, Table \ref{tab:Berlin bandwidth application} contains the maximum lags of the estimator of the long run covariance kernel from $1952$ to $1972$ ($L$) and from $2000$ to $2025$ ($L_d$)  for each month. These were determined by using the test from  \citet{Kokoszka2017} to assess the cumulative significance of empirical lagged covariance kernels at level $95$\%. The results of these tests are contained in Figure \ref{fig:test of lagged kernels}.

\begin{figure}[h!]
    \centering
    \begin{subfigure}[b]{0.95\linewidth}
        \includegraphics[width=\linewidth]{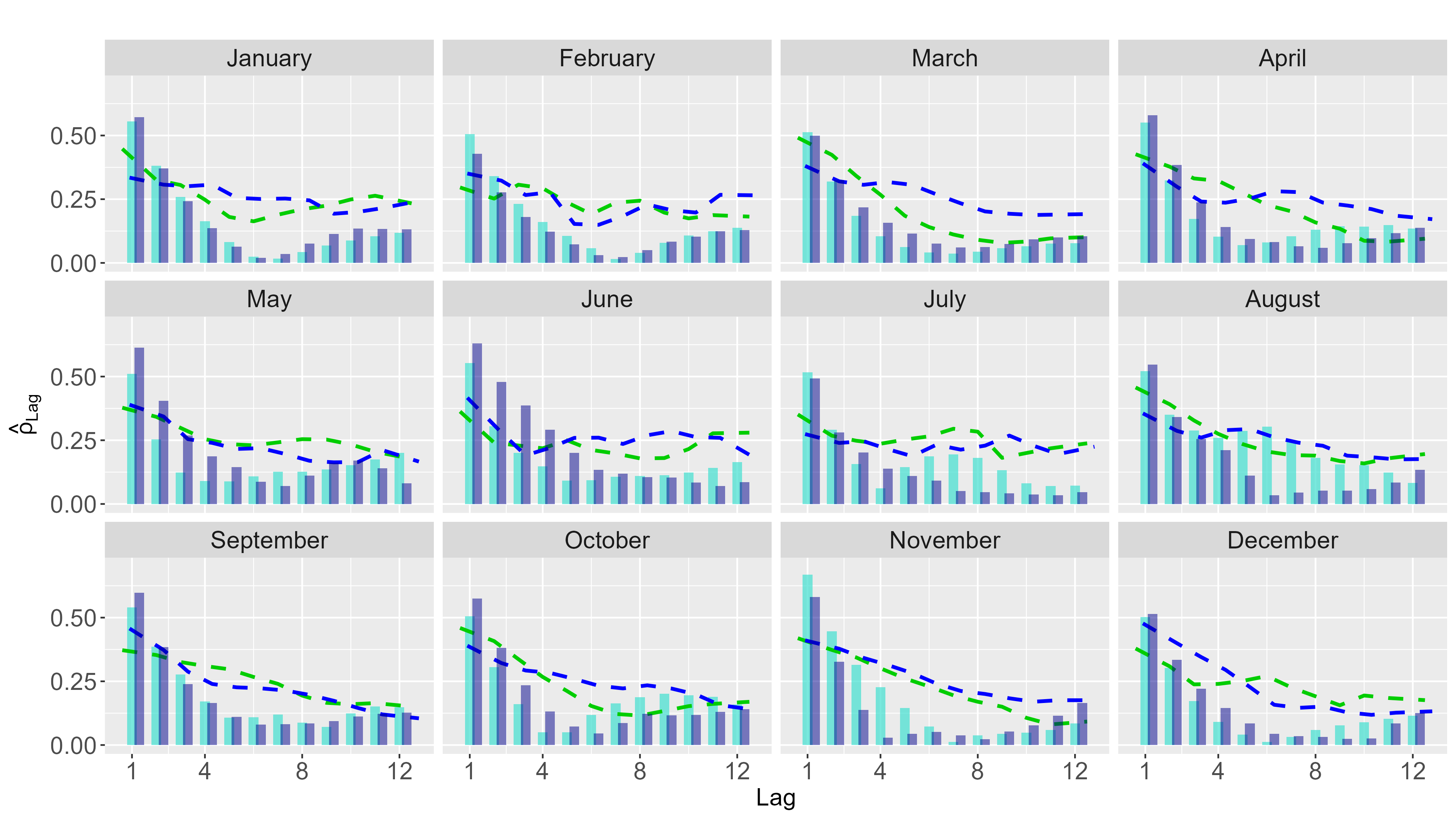}
    \end{subfigure}
    \caption{Test to assess the cumulative significance of empirical lagged covariance kernels at level $95$\%.}
    \label{fig:test of lagged kernels}
\end{figure}

\textbf{Segments and long-run covariance.} For a given month, the $n$ historical curves form $J = 21$ yearly segments of $28$ to $31$ consecutive days, separated by eleven months, and the $m$ recent curves form $J_d = 26$ such segments. We treat curves from different years as independent. Then the curves are dependent only within segments, hence only up to a lag of $30$ days, and lagged pairs are formed only within segments. For $b \geq 1$, let $\hat{\gamma}^{\sparse}_{b}(t,s;\mathcal{I})$ denote the estimator \eqref{eq:crosscov} with the sum restricted to pairs of curves within a segment $\mathcal{I}$ and normalized by $|\mathcal{I}| - 1$ instead of $n - 1$, and similarly for $-b$. We then use
\begin{align*}
\widehat{\Gamma}^{\sparse}(t,s) = \hat{\gamma}^{\sparse}_{0}(t,s) + \sum_{b=1}^{L} w_{b;L}\, \frac{1}{J}\sum_{\mathcal{I}} \Big(\hat{\gamma}^{\sparse}_{b}(t,s;\mathcal{I}) + \hat{\gamma}^{\sparse}_{-b}(t,s;\mathcal{I})\Big), \qquad t \leq s,
\end{align*}
where the inner sum runs over the $J$ segments and $\hat{\gamma}^{\sparse}_{0}$ in \eqref{eq:Gamma_estimator} uses all $n$ curves; $\widehat{\Gamma}^{\dense}$ is defined analogously, with $J_d$ segments and maximal lag $L_d$. A segment $\mathcal{I}$ contains only $|\mathcal{I}| - b$ pairs at lag $b$, so this normalization gives the lag-$b$ terms the weight $(|\mathcal{I}| - b)/(|\mathcal{I}| - 1) \approx 1 - b/|\mathcal{I}|$. This is the weight with which lag-$b$ covariances enter the variance of the average over all curves when no pair crosses a segment boundary. The test above selects $L, L_d \leq 4$ (Table~\ref{tab:Berlin bandwidth application} for Berlin), which indicates short memory but does not prove that the dependence ends at lag $L$.

Figure \ref{fig:sd function application} shows the monthly estimated standard deviation of the covariance kernels along the diagonal, $\hat{\gamma}^{\sparse}_0(t,t)^{1/2}$ and $\hat{\gamma}^{\dense}_0(t,t)^{1/2}$ as well as Figure \ref{fig:sd long run function application} illustrates the monthly estimators $\widehat{\Gamma}^{\sparse}(t,t)^{1/2}$ and $\widehat{\Gamma}^{\dense}(t,t)^{1/2}$, where $L$ and $L_d$ are the maximum of significant lag sequence from the dense and sparse data set of each month from Table \ref{tab:Berlin bandwidth application}. While both curves are similar for each time period, the long-run standard deviation function is larger than that of the lag-zero covariance kernel. Taking into account serial dependence results in larger bands, which is essential to properly keep the nominal level; in our procedure this is achieved by the dependent multipliers of the bootstrap, while $\widehat{\Gamma}^{\sparse}$ and $\widehat{\Gamma}^{\dense}$ determine the shape of the bands, see Section~\ref{sec:covest}.

\begin{figure}[t!]
    \centering
    \begin{subfigure}[b]{0.49\linewidth}
        \includegraphics[width=\linewidth]{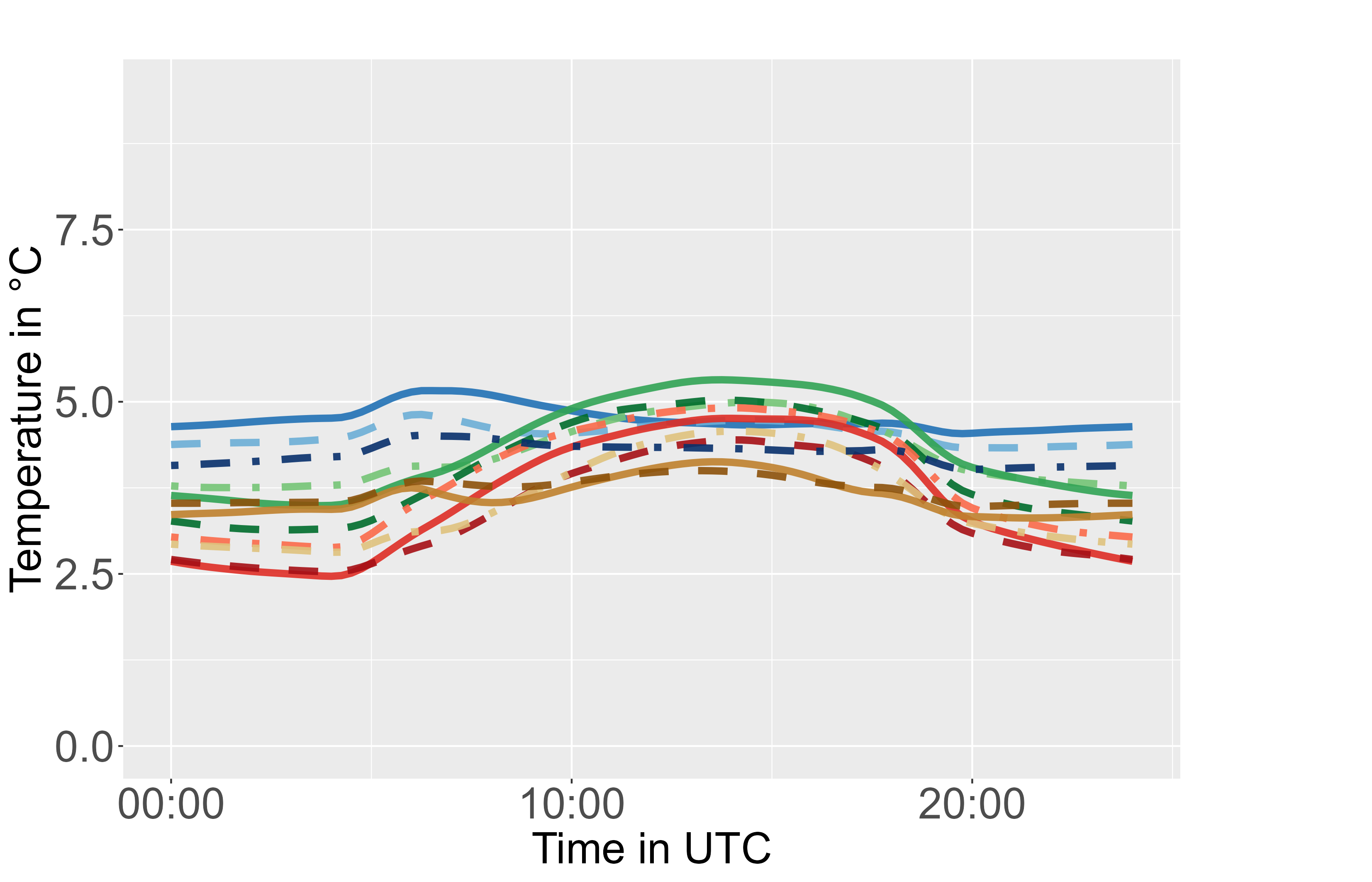}
        \caption{$2000$--$2025$.}
    \end{subfigure}
    \hfill
    \begin{subfigure}[b]{0.49\linewidth}
        \includegraphics[width=\linewidth]{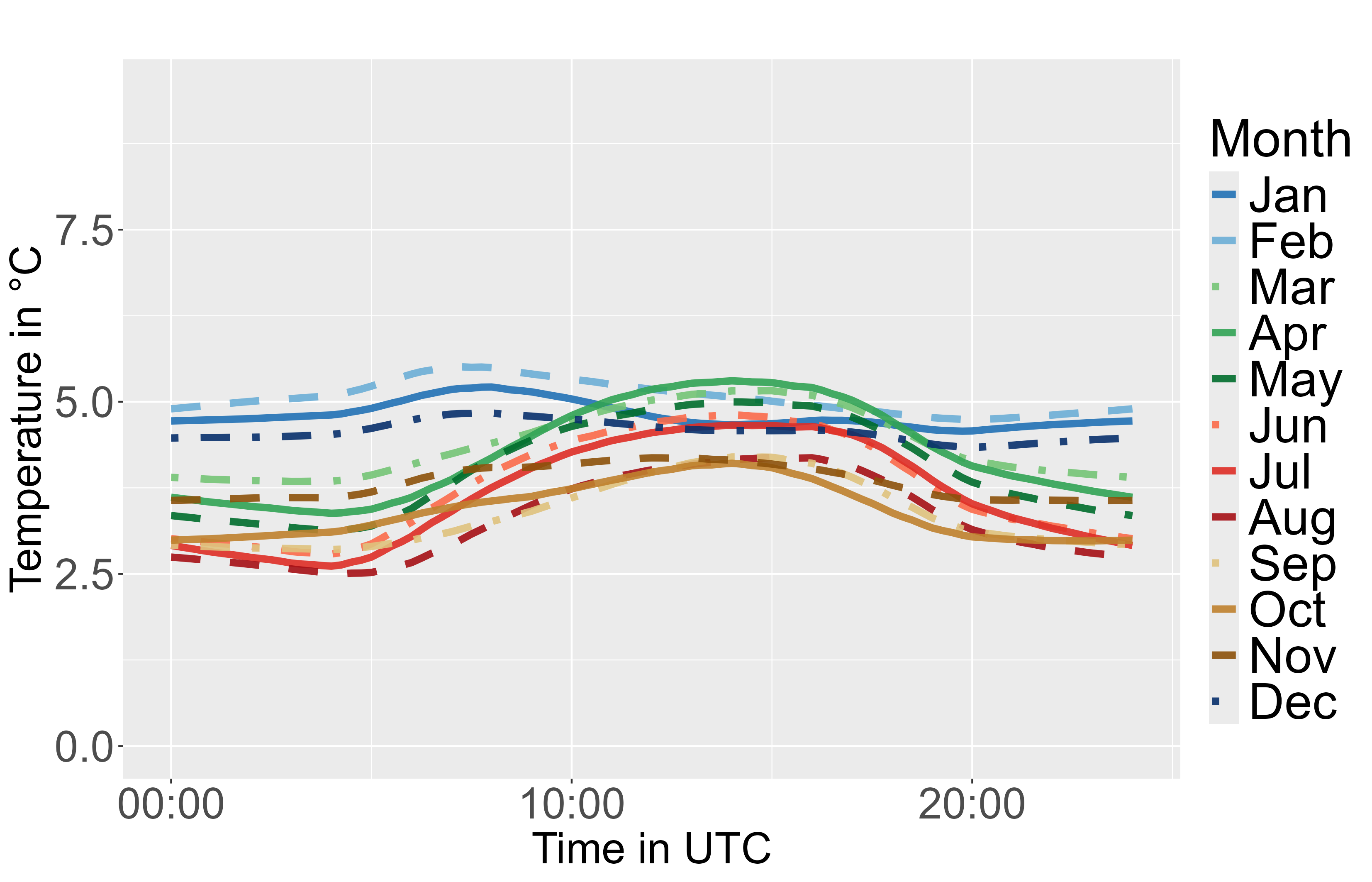}
        \caption{$1952$--$1972$.}
    \end{subfigure}
\caption{Berlin (Germany): Estimation of lagged $0$ standard deviation function.}
    \label{fig:sd function application}
\end{figure}

\begin{figure}[t!]
    \centering
    \begin{subfigure}{0.49\linewidth}
        \includegraphics[width=\linewidth]{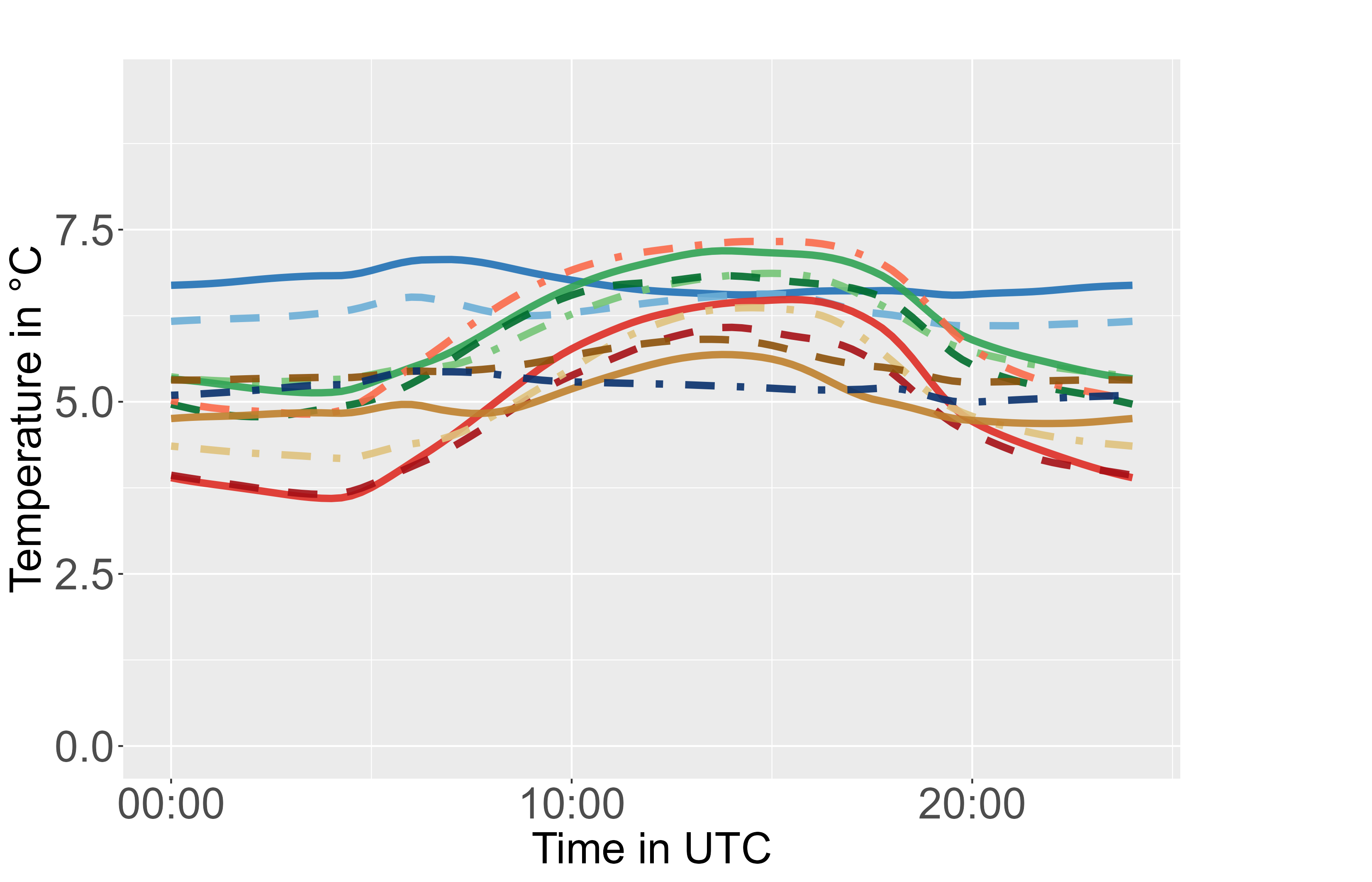}
        \caption{$2000$--$2025$.}
    \end{subfigure}
    \hfill
    \begin{subfigure}{0.49\linewidth}
        \includegraphics[width=\linewidth]{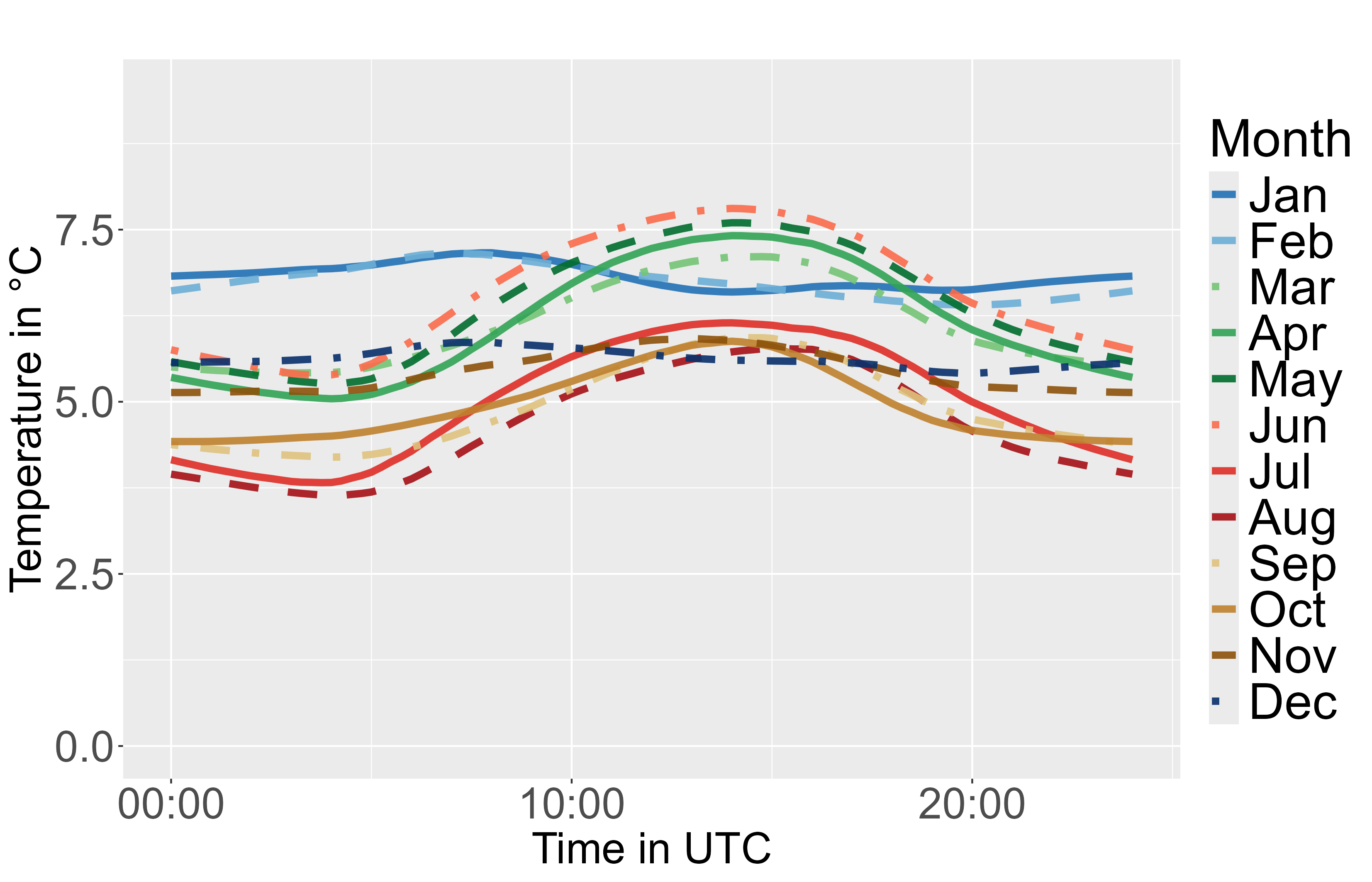}
        \caption{$1952$--$1972$.}
    \end{subfigure}
\caption{Berlin (Germany): Estimation of long run standard deviation function.}
    \label{fig:sd long run function application}
\end{figure}

\textbf{Robustness check of Holm decisions.} For each month a robustness check of the p-values of the global tests of $\int\delta = 0$ and $\Delta = 0$  across different bandwidths ($0{.}75\,h$ and $1{.}5\,h$; $h_d$ is analogously scaled) from Table \ref{tab:bwcities} is provided in Table \ref{tab: robust check of global tests} with Holm-adjustment applied over the 12 months for each city and hypothesis separately. The results show almost no effect on Holm decisions of $\int\delta = 0$ since the integral of local polynomial weights hardly depends on the bandwidths, and only one change in Holm decision, namely for $\Delta = 0$ in September in Hamburg.

\begin{table}[h!]
    \centering{
    \scriptsize
    \renewcommand{\arraystretch}{1.3}
    \begin{tabular}{c|c|c|c|c|c|c|c|c|c|c|c|c}
    \multicolumn{13}{c}{Berlin (Germany): Time periods $2000$--$2025$ and $1952$--$1972$} \\
    \hline
    $H_0:\int \delta = 0$          & Jan       & Feb       & Mar       & Apr       & May       & Jun       & Jul      & Aug        & Sep       & Oct       & Nov      & Dec          \\
    \hline
    \rule{0pt}{2.5ex} $1{.}5 \, h$ & $0{.}001$ & $0{.}002$ & $0{.}002$ & $0{.}001$ & $0{.}002$ & $0{.}011$   & $0{.}002$& $0{.}001$& $0{.}002$ & $0{.}074$ & $0{.}002$& $0{.}022$    \\
    \rule{0pt}{2.5ex} $0{.}75 \,h$ & $0{.}001$ & $0{.}002$ & $0{.}002$ & $0{.}001$ & $0{.}002$ & $0{.}011$   & $0{.}002$& $0{.}001$& $0{.}002$ & $0{.}074$ & $0{.}002$& $0{.}022$    \\
    \hline
    $H_0:\Delta = 0$               & Jan       & Feb       & Mar       & Apr       & May       & Jun       & Jul      & Aug        & Sep       & Oct       & Nov      & Dec           \\
    \hline
    \rule{0pt}{2.5ex} $1{.}5 \, h$ & $0{.}002$ & $0{.}001$ & $0{.}002$ & $0{.}001$ & $0{.}001$ & $0{.}001$ & $0{.}001$& $0{.}001$  & $0{.}003$ & $0{.}003$ & $0{.}003$  & $0{.}114$   \\
    \rule{0pt}{2.5ex} $0{.}75 \,h$ & $0{.}001$ & $0{.}001$ & $0{.}001$ & $0{.}001$ & $0{.}001$ & $0{.}001$ & $0{.}001$& $0{.}001$  & $0{.}002$ & $0{.}001$ & $0{.}001$  & $0{.}061$   \\

    \multicolumn{13}{c}{ } \\
    \multicolumn{13}{c}{Frankfurt (Germany): Time periods $2000$--$2025$ and $1952$--$1972$} \\
    \hline
    $H_0:\int \delta = 0$         & Jan       & Feb       & Mar       & Apr       & May       & Jun       & Jul      & Aug        & Sep       & Oct       & Nov      & Dec        \\
    \hline
    \rule{0pt}{2.5ex} $1{.}5 \,h$ & $0{.}002$ & $0{.}04$  & $0{.}031$ & $0{.}04$  & $0{.}035$ & $0{.}001$ & $0{.}001$& $0{.}001$  & $0{.}021$ & $0{.}04$ & $0{.}001$ & $0{.}04$   \\
    \rule{0pt}{2.5ex} $0{.}75\,h$ & $0{.}002$ & $0{.}04$  & $0{.}031$ & $0{.}04$  & $0{.}035$ & $0{.}001$ & $0{.}001$& $0{.}001$  & $0{.}021$ & $0{.}04$ & $0{.}001$ & $0{.}04$   \\
    \hline
    $H_0:\Delta = 0$              & Jan       & Feb       & Mar       & Apr       & May       & Jun       & Jul      & Aug        & Sep       & Oct       & Nov      & Dec        \\
    \hline
    \rule{0pt}{2.5ex} $1{.}5 \,h$ & $0{.}006$ & $0{.}002$ & $0{.}001$ & $0{.}001$ & $0{.}001$ & $0{.}001$ & $0{.}001$& $0{.}001$ & $0{.}001$  & $0{.}08$  & $0{.}001$& $0{.}003$ \\
    \rule{0pt}{2.5ex} $0{.}75\,h$ & $0{.}005$ & $0{.}001$ & $0{.}001$ & $0{.}001$ & $0{.}001$ & $0{.}001$ & $0{.}001$& $0{.}001$ & $0{.}001$  & $0{.}105$ & $0{.}001$& $0{.}001$ \\

    \multicolumn{13}{c}{ } \\
    \multicolumn{13}{c}{Hamburg (Germany): Time periods $2000$--$2025$ and $1952$--$1972$} \\
    \hline
    $H_0:\int \delta = 0$         & Jan       & Feb       & Mar       & Apr       & May       & Jun       & Jul      & Aug        & Sep       & Oct       & Nov      & Dec        \\
    \hline
    \rule{0pt}{2.5ex} $1{.}5\,h$  & $0{.}001$ & $0{.}001$ & $0{.}001$ & $0{.}001$ & $0{.}001$ & $0{.}004$ & $0{.}001$& $0{.}001$  & $0{.}001$ & $0{.}016$ & $0{.}001$& $0{.}01$   \\
    \rule{0pt}{2.5ex} $0{.}75\,h$ & $0{.}001$ & $0{.}001$ & $0{.}001$ & $0{.}001$ & $0{.}001$ & $0{.}004$ & $0{.}001$& $0{.}001$  & $0{.}001$ & $0{.}016$ & $0{.}001$& $0{.}01$   \\
    \hline
    $H_0:\Delta = 0$              & Jan       & Feb       & Mar       & Apr       & May       & Jun       & Jul      & Aug        & Sep       & Oct       & Nov      & Dec        \\
    \hline
    \rule{0pt}{2.5ex} $1{.}5 \,h$ & $0{.}848$ & $0{.}848$ & $0{.}227$ & $0{.}002$ & $0{.}001$ & $0{.}001$ & $0{.}001$& $0{.}001$  & $0{.}011$ & $0{.}848$ & $0{.}828$& $0{.}227$   \\
    \rule{0pt}{2.5ex} $0{.}75\,h$ & $0{.}934$ & $0{.}934$ & $0{.}083$ & $0{.}002$ & $0{.}001$ & $0{.}001$ & $0{.}001$& $0{.}001$  & $0{.}062$ & $0{.}934$ & $0{.}934$& $0{.}289$   \\

    \multicolumn{13}{c}{ } \\
    \multicolumn{13}{c}{Munich (Germany): Time periods $2000$--$2025$ and $1952$--$1972$} \\
    \hline
    $H_0:\int \delta = 0$         & Jan       & Feb       & Mar       & Apr       & May       & Jun       & Jul      & Aug        & Sep       & Oct       & Nov      & Dec        \\
    \hline
    \rule{0pt}{2.5ex} $1{.}5\,h$  & $0{.}001$ & $0{.}001$ & $0{.}001$ & $0{.}001$ & $0{.}001$ & $0{.}001$ & $0{.}001$& $0{.}001$  & $0{.}001$ & $0{.}001$ & $0{.}001$& $0{.}001$    \\
    \rule{0pt}{2.5ex} $0{.}75\,h$ & $0{.}001$ & $0{.}001$ & $0{.}001$ & $0{.}001$ & $0{.}001$ & $0{.}001$ & $0{.}001$& $0{.}001$  & $0{.}001$ & $0{.}001$ & $0{.}001$& $0{.}001$    \\
    \hline
    $H_0:\Delta = 0$              & Jan       & Feb       & Mar       & Apr       & May       & Jun       & Jul      & Aug        & Sep       & Oct       & Nov      & Dec        \\
    \hline
    \rule{0pt}{2.5ex} $1{.}5 \,h$ & $0{.}001$ & $0{.}001$ & $0{.}001$ & $0{.}001$ & $0{.}001$ & $0{.}001$ & $0{.}001$& $0{.}001$  & $0{.}001$ & $0{.}001$ & $0{.}001$& $0{.}001$    \\
    \rule{0pt}{2.5ex} $0{.}75\,h$ & $0{.}001$ & $0{.}001$ & $0{.}001$ & $0{.}001$ & $0{.}001$ & $0{.}001$ & $0{.}001$& $0{.}001$  & $0{.}001$ & $0{.}001$ & $0{.}001$& $0{.}001$    \\
    \end{tabular}
    }
    \caption{Berlin, Frankfurt, Hamburg and Munich (Germany): Robustness check of p-values of the global tests of $\int\delta = 0$ and $\Delta = 0$ for each month with different bandwidths ($0{.}75\,h$ and $1{.}5\,h$) from Table \ref{tab:bwcities}, Holm-adjusted over the 12 months, separately for each city and each hypothesis.}
    \label{tab: robust check of global tests}
\end{table}

\section{Data preprocessing and homogeneity}\label{sec:preprocessing}

\textbf{Time base.} According to DWD years since $2000$ are timestamped in Coordinated Universal Time (UTC) and years before $2000$ in Central European Time (CET = UTC+1)\footnote{https://opendata.dwd.de/climate\_environment/CDC/observations\_germany/climate/10\_minutes/air\_temperature/ DESCRIPTION\_obsgermany\_climate\_10min\_air\_temperature\_en.pdf \href{https://opendata.dwd.de/climate_environment/CDC/observations_germany/climate/10_minutes/air_temperature/DESCRIPTION_obsgermany_climate_10min_air_temperature_en.pdf}{[pdf]}} without any Central European Summer Time (CEST = UTC+2) corrections after $1949$ and before $1980$ in Germany. We therefore shift the historical series by one hour, so that both periods are analyzed in UTC, and all times of day in the paper refer to UTC.\\

\textbf{Missing data and interpolation.} Temperature air values in a time series that are missing for up to one hour are linearly interpolated. In addition, missing data between one and three hours are spline-interpolated. Days with gaps of more than three hours are removed. Table~\ref{tab:missingvalues} provides a detailed overview of this process.

\begin{table}[h!]
    \centering{
    \scriptsize
    \renewcommand{\arraystretch}{1.3}
    \begin{tabular}{c|cc|cc|cc|cc}
     & \multicolumn{2}{c}{Berlin (Germany)} & \multicolumn{2}{c}{Frankfurt (Germany)} & \multicolumn{2}{c}{Hamburg (Germany)} & \multicolumn{2}{c}{Munich (Germany)}\\ \hline

    timeperiod & $2000$--$2025$ & $1952$--$1972$  & $2000$--$2025$ & $1952$--$1972$  & $2000$--$2025$ & $1952$--$1972$  & $2000$--$2025$ & $1952$--$1972$ \\ \hline

    NAs$\leq1$h        & $0{.}76\%$ & $0\%$ & $0{.}58\%$ & $0{.}01\%$ & $1{.}5\%$ & $0\%$  & $0{.}66\%$ & $0\%$ \\
    $1$h$<$NAs$\leq3$h & $0{.}06\%$ & $0\%$ & $0{.}11\%$ & $0\%$      & $0{.}29\%$ & $0\%$ & $0{.}14\%$ & $0\%$ \\
    NAs$ > 3$h         & $0{.}36\%$ & $0\%$ & $0{.}14\%$ & $0\%$      & $0{.}24\%$ & $0\%$ & $0{.}17\%$ & $0\%$  \\
    \end{tabular}
    }
    \caption{Amount of specific missing values in real data application.}
    \label{tab:missingvalues}
\end{table}

\textbf{Homogeneity at station transitions.} At  all four stations for the current period, $2000$–$2025$, electronic resistance thermometers (PT100) at 2\,m are recorded in the DWD metadata, which have been in use since $1993$--$1997$ without any change of instruments; for $1952$--$1972$, no instrument metadata are available. The potential change in the measurement method between the two time periods may affect the daily cycle and cannot be checked.

Due to continuity gaps in the weather station records in the four cities, the data for the time periods from $1952$ to $1972$ and from $2000$ to $2025$ have been compiled. Table \ref{tab: Homogeneity} provides a detailed list and furthermore shows relocations of used weather stations during this time.

\begin{table}[h!]
    \centering{
    \scriptsize
    \renewcommand{\arraystretch}{1.3}
    \begin{tabular}{c|c|c|c|c|c}

    City & Station ID	&  Latitude	& Longitude	& Elevation (m)	 & Period used \\
    \hline

    \multirow{3}{*}{Berlin}
    & $00424$ & $52.5000$ & $13.4667$ &  $36$   & $1971\text{--}1972$    \\

    & $00433$ & $52.4676$ & $13.4020$ &  $48$   & $1952\text{--}1957$-$04$-$30$    \\

    & $00433$ & $52.4830$ & $13.3914$ & $51$    & $1957$-$05$-$01\text{--}1970$; $2000\text{--}2025$ \\

    \hline

    \multirow{4}{*}{Frankfurt}
     & $01425$ & $50.1214$ & $8.6604$  &  $103$  & $1952\text{--}1961$ \\

     & $01421$ & $50.1474$ & $8.6746$  &  $125$  & $1962\text{--}1972$ \\

     & $01420$ & $50.0449$ & $8.5979$  &  $112$  & $2000\text{--}2014$-$10$-$21$ \\

     & $01420$ & $50.0259$ & $8.5213$  &  $100$  & $2014$-$10$-$22\text{--}2025$ \\

    \hline
    \multirow{3}{*}{Hamburg}
       & $01975$ & $53.6380$ & $9.9994$  &  $12$   & $1952\text{--}1968$-$07$-$11$    \\

       & $01975$ & $53.6382$ & $9.9953$  &  $13$   & $1968$-$07$-$12$--$1972$    \\

       & $01975$ & $53.6332$ & $9.9881$  &  $11$   & $2000\text{--}2025$    \\

    \hline

    \multirow{2}{*}{Munich}
    & $03390$ & $48.1369$ & $11.7094$ &  $529$  & $1952\text{--}1972$ \\

        & $03379$ & $48.1632$ & $11.5429$ &  $515$  & $2000\text{--}2025$ \\

    \end{tabular}
    }
    \caption{Overview of weather stations.}
    \label{tab: Homogeneity}
\end{table}

Given the assumption of homogeneity for time series Table \ref{tab: global test historical period} provides estimated p-values of global difference tests with Holm correction for weather station changes and relocations.
Considering station changes for Berlin, a comparison of the periods $1952$ to $1970$ against $1971$ to $1972$ reveals no significant shifts in $10$ out of $12$ months. At level $5\%$ a significant averaged difference is observed in February and July, although $\delta$ remains significant only in July and $\Delta$ is insignificant for all months.
In contrast, a significant shift is observed only in December for Frankfurt when comparing $1952$--$1961$ to $1962$--$1972$.
Next, by considering weather station relocation of Frankfurt an analysis of $2014$-$10$-$22$--$2025$ with $2000$--$2014$-$10$-$21$ reveals no significant averaged difference, while the difference $\delta$ is only significant for June and furthermore in $6$ out of $12$ months $\Delta$ is significant.
This change may have affected the daily cycle, which is why the results regarding intraday trends for Frankfurt should be interpreted with this in consideration.
Lastly, for Berlin examining time periods $1957$-$05$-$01$--$1972$ against $1952$--$1957$-$04$-$30$ and for Hamburg $1968$-$07$-$12$--$1972$ against $1952$--$1968$-$07$-$11$ no significant changes are observed.
These mostly isolated monthly deviations, except for the recent time period in Frankfurt, cannot be clearly attributed to changes in measuring stations but are more likely due to genuine annual centered fluctuations or trend variability.
For weather station changes of Frankfurt and Munich no further information is available to test homogeneity for $1952$--$1972$ together with $2000$--$2025$. Therefore, the estimated differences in these two cities could be partly attributed to the change of weather station.

\begin{table}[h!]
    \centering{
    \scriptsize
    \renewcommand{\arraystretch}{1.3}
    \begin{tabular}{c|c|c|c|c|c|c|c|c|c|c|c|c}

    \multicolumn{13}{c}{Berlin:$1971$--$1972$ against $1952$--$1970$} \\
    \hline
                         $H_0$              & Jan       & Feb            & Mar       & Apr       & May       & Jun       & Jul           & Aug        & Sep       & Oct     & Nov      & Dec        \\
    \hline
        \rule{0pt}{2.5ex} $\int \delta = 0$ & $1{.}0$   & $\mb{0{.}042}$ & $1{.}0$   & $1{.}0$   & $1{.}0$   & $1{.}0$   & $\mb{0{.}042}$& $0{.}861$  & $1{.}0$   & $1{.}0$ & $1{.}0$  & $0{.}506$   \\
        \rule{0pt}{2.5ex} $\delta = 0$      & $1{.}0$   & $0{.}051$      & $1{.}0$   & $0{.}509$ & $1{.}0$   & $1{.}0$   & $\mb{0{.}031}$& $0{.}509$  & $1{.}0$   & $1{.}0$ & $1{.}0$  & $0{.}377$   \\
        \rule{0pt}{2.5ex} $\Delta = 0$      & $1{.}0$   & $1{.}0$        & $0{.}997$ & $1{.}0$   & $0{.}997$ & $0{.}188$ & $1{.}0$       & $0{.}188$  & $0{.}997$ & $1{.}0$ & $1{.}0$  & $1{.}0$   \\

    \multicolumn{13}{c}{ } \\
    \multicolumn{13}{c}{Berlin: $1957$-$05$-$01$--$1972$ against $1952$--$1957$-$04$-$30$} \\
    \hline
                    $H_0$               & Jan       & Feb       & Mar       & Apr       & May       & Jun       & Jul      & Aug        & Sep       & Oct       & Nov      & Dec        \\
    \hline
    \rule{0pt}{2.5ex} $\int \delta = 0$ & $1{.}0$   & $1{.}0$   & $1{.}0$   & $1{.}0$   & $1{.}0$   & $1{.}0$   & $1{.}0$  & $1{.}0$    & $1{.}0$   & $1{.}0$   & $1{.}0$  & $0{.}858$    \\
    \rule{0pt}{2.5ex} $ \delta = 0$     & $1{.}0$   & $1{.}0$   & $1{.}0$   & $1{.}0$   & $1{.}0$   & $1{.}0$   & $1{.}0$  & $1{.}0$    & $1{.}0$   & $1{.}0$   & $1{.}0$  & $0{.}815$    \\
    \rule{0pt}{2.5ex} $\Delta = 0$      & $0{.}765$ & $0{.}695$ & $1{.}0$   & $1{.}0$   & $1{.}0$   & $0{.}08$  & $0{.}277$& $0{.}148$  & $0{.}139$ & $0{.}277$ & $1{.}0$  & $1{.}0$   \\

    \multicolumn{13}{c}{ } \\
    \multicolumn{13}{c}{Frankfurt (Germany): $2014$-$10$-$22$--$2025$ against $2000$--$2014$-$10$-$21$} \\
    \hline
                    $H_0$               & Jan           & Feb       & Mar            & Apr       & May       & Jun           & Jul       & Aug        & Sep       & Oct      & Nov      & Dec        \\
    \hline
    \rule{0pt}{2.5ex} $\int \delta = 0$ & $0{.}748$     & $0{.}696$ & $0{.}954$      & $1{.}0$   & $1{.}0$   & $0{.}073$     & $0{.}696$ & $0{.}209$  & $0{.}696$ & $1{.}0$  & $1{.}0$  & $0{.}05$   \\
    \rule{0pt}{2.5ex} $ \delta = 0$     & $0{.}698$     & $0{.}494$ & $1{.}0$        & $1{.}0$   & $1{.}0$   & $\mb{0{.}049}$& $0{.}397$ & $0{.}09$   & $0{.}494$ & $1{.}0$  & $1{.}0$  & $0{.}05$    \\
    \rule{0pt}{2.5ex} $ \Delta = 0$     & $\mb{0{.}02}$ & $0{.}415$ & $\mb{0{.}048}$ & $0{.}187$ & $0{.}415$ & $0{.}101$     & $\mb{0{.}032}$  & $\mb{0{.}001}$    & $\mb{0{.}001}$   & $0{.}187$   & $\mb{0{.}04}$  & $0{.}187$    \\

    \multicolumn{13}{c}{ } \\
    \multicolumn{13}{c}{Frankfurt (Germany): $1962$--$1972$ against $1952$--$1961$} \\

    \hline
                         $H_0$              & Jan       & Feb       & Mar       & Apr       & May       & Jun       & Jul      & Aug        & Sep       & Oct       & Nov      & Dec        \\
    \hline
        \rule{0pt}{2.5ex} $\int \delta = 0$ & $1{.}0$   & $1{.}0$   & $0{.}209$ & $1{.}0$   & $1{.}0$   & $1{.}0$   & $1{.}0$  & $1{.}0$    & $1{.}0$   & $1{.}0$   & $1{.}0$  & $\mb{0{.}02}$    \\
        \rule{0pt}{2.5ex} $\delta = 0$      & $1{.}0$   & $1{.}0$   & $0{.}087$ & $1{.}0$   & $1{.}0$   & $1{.}0$   & $1{.}0$  & $1{.}0$    & $1{.}0$   & $1{.}0$   & $1{.}0$  & $\mb{0{.}023}$     \\
        \rule{0pt}{2.5ex} $ \Delta = 0$     & $1{.}0$ & $0{.}17$    & $0{.}17$  & $0{.}251$ & $0{.}077$ & $0{.}092$ & $0{.}401$ & $0{.}092$ & $0{.}235$ & $0{.}392$ & $1{.}0$  & $1{.}0$    \\

    \multicolumn{13}{c}{ } \\
    \multicolumn{13}{c}{Hamburg (Germany): $1968$-$07$-$12$--$1972$ against $1952$--$1968$-$07$-$11$} \\

    \hline
                         $H_0$              & Jan       & Feb       & Mar       & Apr       & May       & Jun      & Jul      & Aug       & Sep       & Oct      & Nov      & Dec        \\
    \hline
        \rule{0pt}{2.5ex} $\int \delta = 0$ & $1{.}0$   & $1{.}0$   & $1{.}0$   & $0{.}948$ & $1{.}0$   & $1{.}0$  & $1{.}0$  & $0{.}14$  & $1{.}0$   & $1{.}0$  & $1{.}0$  & $1{.}0$    \\
        \rule{0pt}{2.5ex} $ \delta = 0$     & $1{.}0$   & $1{.}0$   & $1{.}0$   & $0{.}558$ & $1{.}0$   & $1{.}0$  & $1{.}0$  & $0{.}29$  & $1{.}0$   & $1{.}0$  & $0{.}766$& $1{.}0$    \\
        \rule{0pt}{2.5ex} $ \Delta = 0$     & $0{.}984$ & $0{.}231$ & $0{.}806$ & $0{.}302$ & $0{.}806$ & $0{.}984$& $0{.}209$& $0{.}275$ & $0{.}984$ & $0{.}984$& $0{.}452$& $0{.}806$    \\

    \end{tabular}
    }
    \caption{Station changes and relocations: p-values of the global tests of $\int\delta = 0,\delta = 0$ and $\Delta = 0$ between the periods before and after a change, for each month, Holm-adjusted over the 12 months, separately for each comparison and hypothesis.}
    \label{tab: global test historical period}
\end{table}

\textbf{Trends within periods.} The mean functions $\mu^{\dense}$ and $\mu^{\sparse}$ of a month are mean curves averaged over the years of the respective period. If temperatures trend within a period, $\delta$ compares these period averages, and the confidence bands remain valid for them: the variation of the yearly mean curves around the period average enters the bootstrap processes in (\ref{eq:multiboot}), which are centered at the period mean, as additional variation, so that the bands become wider. To check for trends, Table~\ref{tab:withinperiod} compares the halves of each period, 2013--2025 with 2000--2012 and 1963--1972 with 1952--1962, by the same tests, Holm-adjusted over the 12 months, separately for each city, comparison and hypothesis. At level 5\%, there is no significant change within the historical period. Within the recent period, for all four cities there is no month in which the hypotheses $\int \delta = 0$ and $\Delta=0$ are rejected at the same time. Individually, $\int \delta = 0$ is rejected for December in all four cities and for June in Berlin, while $\Delta=0$ is rejected for January in Hamburg and for August and September in Frankfurt. In addition to $\int \delta = 0$, the hypothesis $\delta=0$ is rejected for December in Munich. A comparison of halves has little power against a slow trend, so these tests cannot exclude one; however, apart from December, the rejections show no common pattern across the cities.

The change in Holm decision is noticeable within the historical period of December in Frankfurt, which involves a cut in $1961$--$1962$ compared to $1962$--$1963$: The averaged temperature in December $1962$ is about $-2{.}6\,^\circ$C, which is $4{.}5\,^\circ$C below the averaged temperature in December of $1952$--$1972$, that is $1{.}9\,^\circ$C. Such observations can shift the estimated mean difference considerably and therefore also the difference. This extreme temperature event and the change in weather station occurred in the same year, which is why the two effects cannot be separated.

\begin{table}[h!]
    \centering{
    \scriptsize
    \setlength{\tabcolsep}{4pt}
    \renewcommand{\arraystretch}{1.3}
    \begin{tabular}{c|c|c|c|c|c|c|c|c|c|c|c|c}
    $H_0$ & Jan & Feb & Mar & Apr & May & Jun & Jul & Aug & Sep & Oct & Nov & Dec \\
    \hline
    \multicolumn{13}{c}{\rule{0pt}{2.5ex}Berlin: 2013--2025 against 2000--2012} \\
    \hline
    \rule{0pt}{2.5ex}$\int \delta = 0$ & $0{.}338$ & $0{.}338$ & $1{.}0$ & $1{.}0$ & $1{.}0$ & $\mb{0{.}031}$ & $0{.}96$  & $0{.}269$ & $0{.}338$ & $0{.}26$  & $1{.}0$ & $\mb{0{.}004}$ \\
    \rule{0pt}{2.5ex}$\delta = 0$      & $0{.}557$ & $0{.}389$ & $1{.}0$ & $1{.}0$ & $1{.}0$ & $0{.}141$      & $1{.}0$   & $0{.}226$ & $0{.}141$ & $0{.}356$ & $1{.}0$ & $0{.}065$ \\
    \rule{0pt}{2.5ex}$\Delta = 0$      & $1{.}0$   & $1{.}0$   & $1{.}0$ & $1{.}0$ & $1{.}0$ & $0{.}768$      & $0{.}433$ & $0{.}208$ & $0{.}553$ & $1{.}0$   & $1{.}0$ & $0{.}304$ \\
    \hline
    \multicolumn{13}{c}{\rule{0pt}{2.5ex}Berlin: 1963--1972 against 1952--1962} \\
    \hline
    \rule{0pt}{2.5ex}$\int \delta = 0$ & $0{.}194$ & $1{.}0$ & $1{.}0$ & $1{.}0$ & $1{.}0$ & $1{.}0$ & $1{.}0$  & $1{.}0$   & $1{.}0$ & $1{.}0$ & $1{.}0$ & $1{.}0$ \\
    \rule{0pt}{2.5ex}$\delta = 0$      & $0{.}203$ & $1{.}0$ & $1{.}0$ & $1{.}0$ & $1{.}0$ & $1{.}0$ & $1{.}0$  & $1{.}0$   & $1{.}0$ & $1{.}0$ & $1{.}0$ & $1{.}0$ \\
    \rule{0pt}{2.5ex}$\Delta = 0$      & $1{.}0$   & $1{.}0$ & $1{.}0$ & $1{.}0$ & $1{.}0$ & $1{.}0$ & $0{.}14$ & $0{.}173$ & $1{.}0$ & $1{.}0$ & $1{.}0$ & $1{.}0$ \\
    \hline
    \multicolumn{13}{c}{\rule{0pt}{2.5ex}Frankfurt: 2013--2025 against 2000--2012} \\
    \hline
    \rule{0pt}{2.5ex}$\int \delta = 0$ & $0{.}669$ & $0{.}695$ & $1{.}0$ & $1{.}0$   & $1{.}0$ & $0{.}472$ & $0{.}167$ & $0{.}695$      & $0{.}695$      & $0{.}695$ & $1{.}0$ & $\mb{0{.}012}$ \\
    \rule{0pt}{2.5ex}$\delta = 0$      & $0{.}761$ & $0{.}761$ & $1{.}0$ & $1{.}0$   & $1{.}0$ & $0{.}58$  & $0{.}292$ & $0{.}761$      & $0{.}58$       & $0{.}761$ & $1{.}0$ & $0{.}067$ \\
    \rule{0pt}{2.5ex}$\Delta = 0$      & $0{.}076$ & $0{.}348$ & $0{.}2$ & $0{.}2$   & $0{.}348$ & $0{.}2$ & $0{.}058$ & $\mb{0{.}013}$ & $\mb{0{.}001}$ & $0{.}2$ & $0{.}078$ & $0{.}348$ \\
    \hline
    \multicolumn{13}{c}{\rule{0pt}{2.5ex}Frankfurt: 1963--1972 against 1952--1962} \\
    \hline
    \rule{0pt}{2.5ex}$\int \delta = 0$ & $1{.}0$ & $1{.}0$ & $1{.}0$     & $1{.}0$   & $1{.}0$   & $1{.}0$   & $1{.}0$   & $1{.}0$   & $1{.}0$ & $1{.}0$   & $1{.}0$ & $0{.}906$ \\
    \rule{0pt}{2.5ex}$\delta = 0$      & $1{.}0$ & $1{.}0$ & $0{.}786$   & $1{.}0$   & $1{.}0$   & $1{.}0$   & $1{.}0$   & $1{.}0$   & $1{.}0$ & $1{.}0$   & $1{.}0$ & $0{.}786$ \\
    \rule{0pt}{2.5ex}$\Delta = 0$      & $1{.}0$ & $0{.}586$ & $0{.}421$ & $0{.}724$ & $0{.}307$ & $0{.}173$ & $0{.}341$ & $0{.}307$ & $0{.}421$ & $1{.}0$ & $1{.}0$ & $1{.}0$ \\
    \hline
    \multicolumn{13}{c}{\rule{0pt}{2.5ex}Hamburg: 2013--2025 against 2000--2012} \\
    \hline
    \rule{0pt}{2.5ex}$\int \delta = 0$ & $1{.}0$        & $0{.}794$ & $1{.}0$ & $1{.}0$ & $1{.}0$ & $0{.}151$ & $1{.}0$ & $1{.}0$   & $0{.}906$ & $0{.}419$ & $1{.}0$ & $\mb{0{.}018}$ \\
    \rule{0pt}{2.5ex}$\delta = 0$      & $1{.}0$        & $0{.}622$ & $1{.}0$ & $1{.}0$ & $1{.}0$ & $0{.}333$ & $1{.}0$ & $1{.}0$   & $0{.}805$ & $0{.}805$ & $1{.}0$ & $0{.}12$ \\
    \rule{0pt}{2.5ex}$\Delta = 0$      & $\mb{0{.}018}$ & $1{.}0$   & $1{.}0$ & $1{.}0$ & $1{.}0$ & $1{.}0$   & $1{.}0$ & $0{.}437$ & $1{.}0$   & $1{.}0$   & $1{.}0$ & $1{.}0$ \\

    \hline
    \multicolumn{13}{c}{\rule{0pt}{2.5ex}Hamburg: 1963--1972 against 1952--1962} \\
    \hline
    \rule{0pt}{2.5ex}$\int \delta = 0$ & $0{.}998$ & $1{.}0$ & $1{.}0$ & $1{.}0$   & $1{.}0$ & $1{.}0$ & $1{.}0$ & $1{.}0$   & $1{.}0$ & $1{.}0$ & $1{.}0$   & $1{.}0$ \\
    \rule{0pt}{2.5ex}$\delta = 0$      & $1{.}0$   & $1{.}0$ & $1{.}0$ & $1{.}0$   & $1{.}0$ & $1{.}0$ & $1{.}0$ & $1{.}0$   & $1{.}0$ & $1{.}0$ & $1{.}0$   & $1{.}0$ \\
    \rule{0pt}{2.5ex}$\Delta = 0$      & $1{.}0$   & $1{.}0$ & $1{.}0$ & $0{.}922$ & $1{.}0$ & $1{.}0$ & $1{.}0$ & $0{.}409$ & $1{.}0$ & $1{.}0$ & $0{.}813$ & $1{.}0$ \\
    \hline
    \multicolumn{13}{c}{\rule{0pt}{2.5ex}Munich: 2013--2025 against 2000--2012} \\
    \hline
    \rule{0pt}{2.5ex}$\int \delta = 0$ & $0{.}415$ & $0{.}657$ & $0{.}771$ & $0{.}986$ & $0{.}748$ & $0{.}748$ & $0{.}051$ & $0{.}748$ & $0{.}657$ & $0{.}748$ & $0{.}803$ & $\mb{0{.}006}$ \\
    \rule{0pt}{2.5ex}$\delta = 0$      & $0{.}647$ & $0{.}776$ & $0{.}776$ & $0{.}819$ & $0{.}776$ & $0{.}776$ & $0{.}08$  & $0{.}776$ & $0{.}647$ & $0{.}776$ & $0{.}819$ & $\mb{0{.}012}$ \\
    \rule{0pt}{2.5ex}$\Delta = 0$      & $0{.}642$ & $0{.}642$ & $0{.}642$ & $0{.}449$ & $0{.}449$ & $0{.}172$ & $0{.}572$ & $0{.}091$ & $0{.}319$ & $0{.}606$ & $0{.}642$ & $0{.}449$ \\
    \hline
    \multicolumn{13}{c}{\rule{0pt}{2.5ex}Munich: 1963--1972 against 1952--1962} \\
    \hline
    \rule{0pt}{2.5ex}$\int \delta = 0$ & $1{.}0$ & $1{.}0$ & $1{.}0$ & $1{.}0$ & $1{.}0$ & $1{.}0$ & $1{.}0$ & $1{.}0$ & $1{.}0$ & $1{.}0$ & $0{.}806$ & $1{.}0$ \\
    \rule{0pt}{2.5ex}$\delta = 0$      & $1{.}0$ & $1{.}0$ & $1{.}0$ & $1{.}0$ & $1{.}0$ & $1{.}0$ & $1{.}0$ & $1{.}0$ & $1{.}0$ & $1{.}0$ & $0{.}857$ & $1{.}0$ \\
    \rule{0pt}{2.5ex}$\Delta = 0$      & $1{.}0$ & $1{.}0$ & $1{.}0$ & $1{.}0$ & $1{.}0$ & $1{.}0$ & $0{.}724$ & $1{.}0$ & $1{.}0$ & $1{.}0$ & $1{.}0$ & $1{.}0$ \\
    \end{tabular}
    }
    \caption{Halves of each period: p-values of the global tests of $\int\delta = 0,\delta = 0$ and $\Delta = 0$ for each month, Holm-adjusted over the 12 months, separately for each city, comparison and hypothesis; values below $0{.}05$ in bold.}
    \label{tab:withinperiod}
\end{table}

\end{document}